\documentclass{ws-ijgmmp}

\usepackage{amsfonts}
\usepackage{array}
\usepackage{color}
\newcommand{\ran}{{\mathrm{Ran\,}}}
\newcommand{\vol}{{\mathrm{vol\,}}}   
\newcommand{\dom}{{\mathrm{Dom\,}}}
\newcommand{\Fix}{{\mathrm{Fix\,}}}
\newcommand{\pr}{{\mathrm{pr\,}}}

\newcommand{\spec}{{\mathrm{Spec \,}}}   
\newcommand{\mitad}{{\frac{1}{2}}}   
\newcommand{\Cl}{{\mathrm{Cl\,}}} 
\newcommand{\mirror}{{\mathrm{mirror}}} 
\newcommand{\ext}{{\mathrm{ext}}}   

\newcommand{\R}{{\mathbb{R}}}         %%real numbers 
\newcommand{\C}{{\mathbb{C}}}         %%complex numbers
\newcommand{\Z}{{\mathbb{Z}}}   

\newcommand\be{\begin{equation}}
\newcommand\ee{\end{equation}}
\newcommand\I{\mathbb{I}}

\newcommand{\beq}{\begin{eqnarray}}
\newcommand{\eeq}{\end{eqnarray}}
\newcommand{\Dsl}{{\,\raise0.7pt\hbox{/}\mkern-12.5mu D}}
\newcommand{\dsl}{{\,\raise0.7pt\hbox{/}\mkern-14mu D}}
\newcommand{\D}{\mathcal{D}}
\newcommand{\norm}[1]{\|#1\|}
\newcommand{\pO}{{\partial \Omega}}

\newcommand{\scalar}[2]{\langle#1\,,#2\rangle}
\newcommand{\matriz}[2]{\left( \begin{array}{#1} #2 \end{array}\right) }
\newcommand{\pois}[2]{\{\,#1,\,#2\,\}}     
\newcommand{\ellip}{{\mathrm{ellip}}}   
\newcommand{\Tr}{{\mathrm{Tr \,}}} 

\newcommand{\set}[1]{\{\,#1\,\}}

\renewcommand{\sp}[1]{{\mathcal{ #1}}}
\renewcommand{\d}{\mathrm{d}}
\renewcommand{\H}{{\mathbb{H}}}

\begin{document}

\catchline{}{}{}{}{}
%
%%%%%%%%%%%%%%%%%%%%%%%%%%%%%%%%%%%%%%%%%%%%%%%%%%%%%

%\ExecuteOptions{a4paper,10pt,twoside,onecolumn,final}

\title{THE TOPOLOGY AND GEOMETRY OF SELF-ADJOINT AND ELLIPTIC BOUNDARY CONDITIONS FOR DIRAC AND LAPLACE OPERATORS}
  
\author{M. Asorey}
\address{Departamento de F\'{\i}sica Te\'orica, Facultad de Ciencias,ŧ 
Universidad de Zaragoza,  \\ E-50009 Zaragoza, Spain.\\
\email{asorey@unizar.es}}

\author{A. Ibort}
\address{ICMAT \& Depto. de Matem\'aticas, Universidad Carlos III de
Madrid, \\ Avda. de la Universidad 30, E-28911 Legan\'es, Madrid, Spain. \\
\email{albertoi@math.uc3m.es}}
 
\author{G. Marmo}
\address{Dipartimento di Fisica, INFN Sezione di Napoli, Universit\'a di
Napoli, \\ I-80125 Napoli, Italy. \\
\email{marmo@na.infn.it}}

\maketitle

\begin{history}
\received{(Day Month Year)}
\revised{(Day Month Year)}
\end{history}
 
 %%%%%%%%%%%%%%%%%%%%%%
 %%%%%%%%%%%%%%%%%%%%%%
 
\begin{abstract}
The theory of self-adjoint extensions of first and second order elliptic
differential operators on manifolds with boundary is studied via its most representative instances: Dirac and
Laplace operators.   

The theory is developed by exploiting the geometrical
structures attached to them and, by using an adapted Cayley transform
on each case, the space $\mathcal{M}$ of such extensions is shown to have a canonical group
composition law structure.     

The obtained results are compared with von Neumann's Theorem 
characterising the self-adjoint extensions of densely defined symmetric
operators on Hilbert spaces.  The 1D case is thoroughly investigated.

The geometry of the submanifold of elliptic self-adjoint extensions $\mathcal{M}_\ellip$ is 
studied and it is shown that it is a Lagrangian submanifold of the universal Grassmannian $\mathbf{Gr}$.
The topology of $\mathcal{M}_\ellip$ is also explored and it is shown that
there is a canonical cycle whose dual is the Maslov class of the manifold.
Such cycle, called the Cayley surface, plays a relevant role in the study of
the phenomena of topology change.

Self-adjoint extensions of Laplace operators are discussed in the path
integral formalism, identifying a class of them for which both treatments
leads to the same results.

A theory of dissipative quantum systems is proposed based on this theory
and a unitarization theorem for such class of dissipative systems is proved.

The theory of self-adjoint extensions with symmetry of Dirac operators is 
also discussed and a reduction theorem for the self-adjoint elliptic Grasmmannian
is obtained.

Finally, an interpretation of spontaneous symmetry breaking is offered from the
point of view of the theory of self-adjoint extensions.

\end{abstract}

%\vskip 1.5cm

%PACS numbers: %11.30.Pb


%\newpage

\tableofcontents

\markboth{M. Asorey, A. Ibort, G. Marmo}
{The Topology and Geometry of Self-adjoint and Elliptic Boundary Conditions}

%%%%%%%%%%%%%%%%%%%%%%%%%
%%%%%%%%%%%%%%%%%%%%%%%%%
%%%%%%%%%%%%%%%%%%%%%%%%%

\section{Introduction}

%%%%%%%%%%%%%%%%%%%%%%%%%%

The construction and discussion of quantum mechanical systems requires a
detailed analysis of the boundary conditions (BC) imposed on the
system.  Often such boundary conditions are imposed by the observers and their
experimental setting or are inherent to the system.
This is a common feature of all quantum systems, even the simpler ones. 

The outcomes of measurable quantities of the system will change with
the choice of BC because the spectrum of the quantum observables will vary with the
different self-adjoint extensions obtained for them depending on the chosen BC.  

The Correspondence Principle provides a useful guide to the analysis of quantum mechanical
systems, but it is not obvious how it extends to include boundary conditions
both in classical and quantum systems.  For instance, Dirichlet's boundary
conditions corresponds to impenetrability of the classical walls
determining the boundary of the classical system in configuration space,
but, what are the corresponding classical conditions for mixed BC?   
Then, we are facing the problem of determining the classical limit of
quantum BC.   Conversely, we can address the question of `quantizing' classical boundary conditions.  In
particular we can ask if the classical determination of BC is enough to
fully describe a `quantized' system.   

As the experimental and observational capabilities are getting more and more powerful, we are
being forced to consider boundary conditions beyond the
standard ones, Dirichlet, Neumann, etc.   For instance, in Condensed matter, models
with `sticky' boundary conditions are seen to be useful to understand
certain aspects in the Quantum Hall effect \cite{Jo95}; in quantum
gravity, self-adjoint extensions are used to understand signature change
\cite{Eg95}.  Even more fundamental, topology change in quantum systems
is modelled using dynamics on BC's \cite{Ba95}.    Physical implications
of the problem were already analyzed in \cite{As12}.  % (Add more references here from BAl'smanifesto).

Following Dirac's approach, we can develop a canonical quantization
program for classical systems with boundary.   Such program requires
a prior discussion on the dynamics of Hamiltonian systems with
boundary.   

Without entering a full discussion of this important
problem, we may assume that a classical Hamiltonian system with
boundary is specified by a Hamiltonian function $H$ defined on the phase
space $T^*\Omega$ of a configuration space $\Omega$ with boundary
$\partial \Omega$, and a canonical transformation $S$ of the symplectic
boundary $T^*\partial \Omega$, obtained by considering the quotient of the restriction $T_{\partial\Omega}^*\Omega$ of the
cotangent bundle of $\Omega$ to $\partial \Omega$ and then, taking its quotient by the characteristic distribution of the restriction of the canonical symplectic structure on $T^*\Omega$ to it.  Thus, classical
boundary conditions (CBC's) are defined by canonical maps
\begin{equation}\label{CBCs}
S \colon T^*\partial \Omega \to T^*\partial \Omega \, ,
\end{equation}
and form a group, the group of symplectic
diffeomorphisms of $T^*(\partial\Omega)$\footnote{Similar
considerations can be made for more general classical phase space.}.

Dirac's quantization rule will be stated as follows:  Given a
CBC $S$, and two classical observables $f,g$ on $T^*\Omega$, determine
a quantum boundary condition (QBC) $\hat{S}$ and
two self-adjoint operators $\hat{f}_S$, $\hat{g}_S$ depending on
$\hat{S}$, such that
\begin{equation}\label{DQ1}
[\hat{f}_S , \hat{g}_S ] = i\hbar \widehat{\pois f g}_S \, ,
\end{equation}
and 
\begin{equation}\label{DQ2} 
\hat{S} \circ \hat{R} = \widehat{S R} ,
\end{equation}
where the composition on the left is a group composition on the space of
QBC's to be discussed later on.  It is obvious that as in the
boundaryless situation such quantization rule could not be implemented
for all observables and all CBC's.  So, one important question for such
program will be how to select subalgebras of classical observables and
subgroups of CBC's suitable for quantization.  

Before embarking in such
enterprise, some relevant aspects of the classical and the quantum
picture need to be clarified.  For instance, we have to understand the
structure of the self-adjoint extensions of the operator corresponding to a
classical observable.  The most important class of operators arising in
the first quantization of classical systems are first and second
order elliptic differential operators: the Laplace-Beltrami operator
when quantizing a classical particle without spin, the Dirac operator
for the quantization of particles with spin.  

This family of operators are
certainly the most fundamental of all elliptic operators, and in fact,
in a sense all elliptic operators are obtained from them.
Thus for Dirac and Laplace operators we will like to understand their
self-adjoint realizations in terms of CBC's.  For that we need to
understand first their self-adjoint realizations in terms of QBC's.  

Von Neumann developed a general theory of self-adjoint extensions of
symmetric operators \cite{Ne29}.  Such theory is often presented in the
realm of abstract Hilbert space theory that causes that the relevant
features attached to the geometry of the operators is
lost.  Thus we will proceed following a direct approach to the theory
of self-adjoint extensions exploiting the geometry of first and second order elliptic
differential operators, and obtaining a fresh interpretation of von
Neumann's theory directly in terms of boundary data.  A consequence of this
analysis will be an interpretation of QBCs in terms of the unitary group at the boundary,
that provides the group composition law needed for the implementation of Dirac's quantization rule
Eq. (\ref{DQ2}).

Elliptic differential operators have been exhaustively studied
culminating with the celebrated Atiyah-Singer
index theorem that relates the analytical index of such operators with the
topological invariants of the underlying spaces \cite{At68}.   Such analysis extends
to the boundary situation provided that appropriate elliptic boundary
conditions are used.  A remarkable example is provided by the global
elliptic boundary conditions introduced in \cite{At75}, or APS conditions,
that allow to describe the index of Dirac operators on manifolds with boundary.  
Such extensions have been adequately generalized for higher order
elliptic operators giving rise to interesting constructions of boundary
data \cite{Fr97}.   We should also recall here the 
important contributions by Lesch and Wojciechowski on the theory
of elliptic boundary conditions and spectral invariants and the geometry
of the elliptic Grassmannian (see for instance \cite{Le96}).

However, from a physical viewpoint another family of boundary conditions
has been considered for quark fields in bag models of quark confinement in QCD.
The theory of chiral boundary conditions is not well developed although many
physical applications have been analyzed from this perspective \cite{mit}, \cite{chiral}, \cite{As13}, \cite{As15}.

The program sketched so far concerns
exclusively  with first quantization of classical systems, but second
quantization is needed to truly deal with their physical nature. 
First quantization of classical systems,
requires to consider the quantization of boundary conditions, which
leads automatically to consider a collection of QBC's for the first
quantized system.   Even for very simple systems, like a fermion
propagating on a disk we need to consider `all' self-adjoint extensions
of Dirac's operator on the disk.   Thus, to proceed to second
quantization we need to understand the global structure of such space of
extensions.   We will show that such space lies naturally in the
infinite dimensional Grassmannian manifold and defines a Lagrangian
submanifold of it, that will be called the self-adjoint Grassmannian. 

Such infinite dimensional Grassmannian was introduced in the study of
the KdV and KP integrable hierarchies of nonlinear partial differential equations \cite{Se85}.    
It represents a `universal phase space'
for a large class of integrable evolution problems.  Lately, such
infinite dimensional Grassmannian was introduced again as the phase space of
string theory and its quantization was discussed \cite{Al87},
\cite{Wi88}.     

Our approach here is different, the infinite dimensional
Grassmannian appears as the natural setting to discuss simultaneously
all QBC's for a first quantized classical system of arbitrary
dimension.   In fact, the relevant QBC's are contained in a Lagrangian submanifold of
the elliptic Grassmannian that should be subjected to second
quantization.  Lagrangian submanifolds of symplectic manifolds play the
role of `generalized functions', thus such program would imply
quantizing a particular observable of the Grassmannian, making contact
again with string theory.  We should stress here that string theory is
genuinely 2D whereas we are dealing with a classical point-like theory
for classical objects in arbitrary dimensions.

The realization of such program would eventually introduce quantum
dynamical effects on the space of QBC's, suggesting that QBC's could
actually change (as suggested in \cite{Ba95}), or that there is the possibility of non-vanishing
probability amplitudes between states characterized by different
boundary conditions.   This observation imply that probably our first
statement concerning the structure of classical systems with boundary is
not totally correct, and we should extend it as follows:  In the
boundaryless situation, a Hamiltonian function defines a flow of
symplectic diffeomorphisms $\phi_t$ on the phase space of the system. 
We must replace then the specification of CBC's by a fixed symplectic
diffeomorphism on the boundary and to allow for a flow of symplectic
diffeomorphisms on the boundary too, i.e., by a boundary Hamiltonian
$H_B$.  Upon quantization such Hamiltonian will define a quantum
Hamiltonian on the space of QBC's that eventually will lead to a
propagator on second quantization.  

Apart from the difficulties and
physical implications of such ideas, we must point out an immediate
consequence related to the topology of the systems studied. 
Changing the classical and quantum boundary conditions implies that
possibility of changes in the topology of the system.  
For instance the quantization of a fermion moving
on a disk with changing boundary conditions can change the topology of
the disk and evolve into a surface of higher genus.  
Such process was analyzed by Asorey, Ibort and Marmo \cite{As05} and 
it was shown that in the first quantized scheme such change can only occur if the trajectory
in the space of self-adjoint extensions of the system cuts a
submanifold where the spectrum of the operator diverges. Such
submanifold is called the Cayley submanifold and describes its
topology.  Thus, in this sense the Cayley submanifold acts as a
`wall' in the space of QBC's for a first quantized system with
boundary, even though this will not preclude tunnelling in second
quantization (see also \cite{Sh12} for arguments along similar lines).  
  
There are many other physical phenomena involving the `boundary'
of physical systems that it would be simply impossible to enumerate here.
We would just like to mention a few recent contributions related to
different physical problems: the analysis of 
boundary conditions and the Casimir effect in \cite{As06}, \cite{As07}, the role of boundary conditions in topological insulators \cite{As13}, and boundary conditions generated entanglement \cite{Ib14b}.

In this work we will try to offer a panorama of the field by presenting
a discussion on some fundamental aspects of the theory of self-adjoint extensions
of Laplace and Dirac-like operators, as well as
a variety of related ideas and problems still on development.  There is a vast 
literature on the subject covering its mathematical flank and it would be impossible to reproduce here.
We must cite however, apart from the reference to the work by von Neumann already quoted, 
the pioneering work by Weyl \cite{We09}, Friedrichs' construction \cite{Fr34} or the contributions by Krein and Naimark\cite{Kr47,Na43} to
the general theory as a few historical landmarks.   

On the other hand the theory of extensions of elliptic operators has attracted a lot of attention 
and, apart from the index theorem related works cited above, there is a number of fundamental 
results on the fields.  We will just quote the analysis 
of non-local extensions of elliptic operators by Grubb \cite{Gr68} and the theory of
singular perturbations of differential operators by Albeverio  and Kurasov \cite{Al99} because of
their influence on this work (see also \cite{Ib14} for 
a quadratic forms based analysis of the extensions of the Laplace-Beltrami operator and \cite{Ib12}
where the reader will find a reasonable list of references on the subject).

Thus, we will proceed by reviewing first the geometric theory of self-adjoint extensions of Dirac and Laplace
operators where the role of the Cayley transform at the boundary will be emphasized and the group structure
of the space of self-adjoint extensions will be described explicitly. This will be the content of Sects. \ref{section:Dirac} and \ref{section:Laplace}. Afterwards, Section \ref{section:Neumann} the relation with von Neumann's theorem will be discussed and the one-dimensional situation will be discussed thoroughly.   The spectral function for arbitrary boundary conditions will 
be obtained explicitly and the quantum analogue of Kirchoff's laws will be obtained.

In Section \ref{section:Feyman} we will begin to discuss the semiclassical theory of
self-adjoint extensions and its implementation in the path integral representation of quantum systems.
Section \ref{section:Grassmann} will be devoted to study the global structure of the space of
extensions both elliptic and self-adjoint.  The infinite-dimensional Grassmannian will be 
discussed as well as the Lagrangian submanifold of elliptic self-adjoint extensions.

We will introduce the study of dissipative quantum systems via non-self-adjoint 
boundary conditions in Section \ref{section:dissipative} where a unitarization theorem
for dissipative systems will be stated.

The problem of dealing with symmetries of quantum systems will be addressed in Section \ref{section:symmetry},
that is, if the quantum system has a symmetry, which are the self-adjoint extensions compatible with it.
This problem will be stated and partially solved (see also \cite{Ib14c}) and some examples will be exhibited.   
Particular attention will be paid to the space of self-adjoint extensions of the quotient Dirac operator of a Dirac
operator with symmetry and its characterization as the symplectic manifold of fixed points of the self-adjoint Grassmannian.

It could also happen that even if we have a symmetry of a symmetric operator, the self-adjoint extension describing the corresponding quantum system will not share the same symmetry.   We will say then that there is a spontaneous breaking of the symmetry.  This situation will be succinctly dealt with in Section \ref{section:breaking}. 

\newpage

%%%%%%%%%%%%%%%%%%%%%%%%%%%%%%
%%%%%%%%%%%%%%%%%%%%%%%%%%%%%%
%%%%%%%%%%%%%%%%%%%%%%%%%%%%%%

\section{Self-adjoint extensions of first order elliptic operators: Dirac operators}\label{section:Dirac}

%%%%%%%%%%%%%%%%%%%%%%%%%%%%%%
%%%%%%%%%%%%%%%%%%%%%%%%%%%%%%

\subsection{Dirac operators}\label{sec:Dirac_operators}

As it was indicated before, Dirac operators constitute an important class of 
first order elliptic operators, to the extent that most relevant elliptic operators 
arising in geometry and physics are directly related to them.  
Let us set the ground for them.  We will consider a
Riemannian manifold $(\Omega, \eta )$ with smooth boundary
$\partial\Omega$.  We denote by $\Cl(\Omega)$ the Clifford
bundle over $\Omega$, defined as the algebra bundle whose fibre at $x\in
\Omega$ is the Clifford algebra $\Cl(T_x\Omega)$ generated by vectors $u$ in
$T_x\Omega$ with relations
$$
u\cdot v + v\cdot u = - 2\eta (u,v)_x \, , \qquad \forall u,v\in T_x\Omega \, .
$$ 

Let $\pi \colon S\to \Omega$ be a $\Cl (\Omega)$-complex vector bundle over $\Omega$, i.e.,
for each $x\in \Omega$, the fibre $S_x = \pi^{-1}(x)$ is a
$\Cl(\Omega)_x$-module, or in other words, there is a representation of
the algebra $\Cl(\Omega)_x = \Cl (T_x\Omega)$ on the complex space $S_x$.  We will indicate by $u\cdot \xi$ the action of the element $u\in \Cl(\Omega)_x$ on the vector $\xi \in S_x$.  We will also assume in
addition that the bundle $S$ carries a hermitian metric denoted by
$(\cdot, \cdot)$ such that Clifford multiplication by unit vectors in $T\Omega$
is unitary, that is,
\begin{equation}\label{unit_cl} 
( u\cdot \xi, u\cdot \zeta )_x = ( \xi,\zeta )_x \, ,  \qquad 
\forall \xi,\zeta \in S_x, \, u\in T_x\Omega \, , \,\, ||u||^2 = 1 \, .
\end{equation}
Finally, we will assume that there is a Hermitean connection $\nabla$ on $S$ such
that 
\begin{equation}\label{derD} 
\nabla ( V \cdot \xi ) = (\nabla_\eta V)\cdot \xi +
V\cdot \nabla \xi ,
\end{equation}
where $V$ is a smooth section of the Clifford bundle $\Cl
(\Omega)$, $\xi \in \Gamma
(S)$ and $\nabla_\eta$ denotes the canonical connection on
$\Cl(\Omega)$ induced by the Riemannian structure $\eta$ on $\Omega$.  

A bundle $\pi\colon S\to \Omega$ with the structure described above is commonly called a Dirac
bundle \cite{La89} and they provide the natural framework to define Dirac
operators.  Thus, if $\pi\colon S\to \Omega$ is a Dirac bundle, we can define a
canonical first-order differential operator $\Dsl\colon \Gamma (S) \to
\Gamma (S)$ by setting:
\begin{equation}\label{dirac_op}
 \Dsl\xi = e_j \cdot \nabla_{e_j} \xi \, ,
\end{equation}
where $e_j$ is any orthonormal frame at $x\in \Omega$. 

There is a natural inner product on the space of smooth sections $\Gamma (S)$ of the Dicrac bundle $S$ induced from the
pointwise inner product $(\cdot, \cdot )$ on $S$ by setting
$$ 
\langle \xi , \zeta \rangle = \int_\Omega  (\xi (x), \zeta (x))_x \,  \vol_\eta (x) \, ,
$$ 
where $\vol_\eta$ is the volume form defined by the
Riemannian structure on $\Omega$.   We will denote the associated norm by
$|| \cdot ||_2$ and $L^2 (\Omega, S)$ is the  corresponding Hilbert space of square integrable sections 
of $S$.  

The Dirac operator $\Dsl$ is defined on the Frechet space of smooth
sections of $S$,
$\Gamma (S)$, however this is not the largest domain where it can be
defined.  The largest domain in $L^2(\Omega,S)$ where this operator can be defined
consists of the completion of
$\Gamma (S)$ with respect to the Sobolev norm $|| \cdot ||_{1,2}$ 
defined by
\begin{equation}\label{sobov_k} 
||\xi ||_{k,2}^2 = \int_\Omega (\xi (x), (I + \nabla^\dagger\nabla)^{k/2}
\xi (x) )_x \vol_\eta (x) \, ,
\end{equation}
with $k =1$,
where $\nabla^\dagger$ is the adjoint operator to $\nabla$ in $L^2(\Omega,S)$. 
In fact, such space is the space of sections of $S$ possessing first
weak derivatives which are in $L^2(\Omega,S)$.  This
Hilbert space will be denoted by
$H^1(\Omega,S)$ (see for instance \cite{adams03}). 

The operator $\Dsl$ defined on $H^1 (\Omega, S)$ is not self-adjoint as
we will discuss in what follows.
However it is immediate to check that the Dirac operator is
symmetric in $\Gamma_0 ({\buildrel \circ \over S})$, the space of smooth 
sections of $S$ with compact support contained in $\buildrel \circ \over 
\Omega$, the interior of $\Omega$.  In fact, after an
integration by parts we obtain immediately,
\begin{equation}\label{symD} 
\langle \Dsl \sigma , \rho \rangle = \langle \sigma, \Dsl \rho
\rangle \, , \qquad  \forall \sigma, \rho \in \Gamma_0 ({\buildrel \circ \over
S}) \, .
\end{equation}
The operator $\Dsl$ with this domain
is closable on $H^1 (\Omega, S)$ and its closure is the completion of $\Gamma_0
({\buildrel \circ \over S})$ with respect to the norm $||\cdot ||_{1,2}$. 
Such domain will be denoted by $H_0^1 (\Omega, S)\subset H^1(\Omega,S)$.  The elements
of $H_0^1 (\Omega, S)$ are sections vanishing on $\partial \Omega$ with
$L^2$-weak derivatives.  Notice that both $H_0^1(\Omega,S)$ and $H^1(\Omega,S)$ are dense
subspaces in $L^2(\Omega,S)$.

 If we denote by $\Dsl_0$ the closure of $\Dsl$ on
$H_0^1 (S)$, it is simple to check that $\Dsl_0^\dagger = \Dsl$ with domain
$H^1(\Omega,S)$, where $\Dsl_0^\dagger$ denotes the adjoint of $\Dsl_0$ in $L^2 (S)$.  The domains of the different self-adjoint extensions, if any, of $\Dsl_0$ will be
linear dense subspaces of $H^1 (\Omega, S)$ containing $H_0^1(\Omega,S)$ such
that the boundary terms obtained in the integration by parts procedure
will vanish.  We will denote in what follows by $\dom (T)$ the
domain of the operator $T$ and by $\ran (T)$ its range, then the symmetric
extensions $\Dsl_s$ of $\Dsl_0$ will be defined on subspaces $\dom (\Dsl_s)$ such
that 
$$
H_0^1 (\Omega, S) = \dom (\Dsl_0) \subset \dom (\Dsl_s) \subset \dom (\Dsl_s^\dagger)
\subset \dom (\Dsl) = H^1 (\Omega, S) \, , 
$$
and $\Dsl_s\xi = \Dsl \xi$ for any $\xi\in \dom (\Dsl_s)$.  A self-adjoint extension of $\Dsl$ is 
a symmetric extension $\Dsl_s$ of $\Dsl_0$ such that $\dom (\Dsl_s) = \dom (\Dsl_s^\dagger)$.
 
We will characterize such self-adjoint extensions by using first the geometry of some Hilbert
spaces defined on the boundary of $\Omega$ and, later on, we will compare
with the classical theory of self-adjoint extensions of densely defined symmetric operators developed by von Neumann \cite{Ne29}.  To do that, we will repeat
first the well-known integration by part process used to derive formula
(\ref{symD}).

Let $x\in \Omega$ and $\{ e_j\}$ an orthonormal frame in a neighborhood of
$x$ so that $\nabla_{e_j} e_i = 0$ at $x$ for all $i,j$.   If
$\xi$, $\zeta$ are sections of $S$, then they define a vector field
$X$ in a neighborhood of $x$ by the condition
$$
\eta ( X, Y ) = - ( \xi, Y\cdot \zeta ) \, , \qquad \forall Y \, .
$$
Then, at $x$, we get:
$$\langle \Dsl\xi, \zeta \rangle = \langle e_j \nabla_{e_j} \xi, \zeta
\rangle = - \nabla_{e_j} \langle \xi, e_j\zeta \rangle + \langle \xi , \Dsl\zeta
\rangle ,$$
but, by definition $\mathrm{div\, } (X) = \eta (\nabla_{e_j}X, e_j )$,
hence,
$$
\mathrm{div\,} (X) = \nabla_{e_j} \eta ( X, e_j ) - \eta (X, \nabla_{e_j}e_j) =  - (
e_j\cdot \xi , e_j\cdot \zeta ) .$$ 
Namely,
$$ 
( \Dsl\xi (x), \zeta (x) )_x - (\xi (x) , \Dsl\zeta (x))_x =  \mathrm{div}_x
(X) \,.
$$ 
Integrating both parts of the equation, we will find,
$$ 
\langle \Dsl\xi, \zeta\rangle - \langle \xi , \Dsl\zeta \rangle = \int_\Omega
\mathrm{div}_x (X) \vol_\eta (x) = \int_\Omega d (i_X \vol_\eta) = \int_{\partial
\Omega} i^* (i_X \vol_\eta) ,$$
where we denote by $i\colon \partial
\Omega \to \Omega$ the canonical inclusion. 
If $\nu$ denotes
the inward unit vector on the normal bundle to
$\partial \Omega$, the volume form $\vol_\eta$ can be written as
$\theta\wedge \vol_{\partial\eta}$, where $\vol_{\partial\eta}$ is an
extension of the volume form defined on
$\partial \Omega$ by the restriction of the Riemannian metric $\eta$ to it, and $\theta$ is a
1-form such that $\theta (Y) = \eta (Y, \nu)$.  Then we easily get,
$$ 
i_X \vol_\eta = (i_X\theta) \vol_{\partial\eta} = \eta (X,\nu )
\vol_{\partial \eta} = ( \xi, \nu\cdot \zeta )\, \vol_{\partial\eta} \, ,
$$
and, finally, we get
\begin{equation}\label{int_part_D1} 
\langle \Dsl\xi, \zeta\rangle - \langle \xi ,
\Dsl\zeta \rangle =  \int_{\partial \Omega} i^* (\nu\cdot \xi,
\zeta )\,  \vol_{\partial \eta} \, .
\end{equation}

%%%%%%%%%%%%%%%%%%%%%%%%%
%%%%%%%%%%%%%%%%%%%%%%%%%
%%%%%%%%%%%%%%%%%%%%%%%%%
%%%%%%%%%%%%%%%%%%%%%%%%%

\subsection{The geometric structure of the space of boundary data}\label{sec:Dirac_bc}

We will denote by $S_{\partial \Omega}$ the restriction of the Dirac bundle $S$
to $\partial\Omega$, i.e., $\pi_{\partial \Omega} \colon S_{\partial \Omega}= S\mid_{\partial \Omega} = i^*S \to \partial \Omega$, $\pi_{\partial \Omega} (\xi) = \pi(\xi)$ for any $\xi \in S_x$, $x\in \partial \Omega$. 
Notice that $S_{\partial \Omega} $ becomes a Dirac bundle over $\partial
\Omega$.  It inherits an inner product $(\cdot, \cdot )$ induced from the Hermitean
product on $S$ as well as an Hermitean
connection $\nabla_{\partial\Omega}$, defined again by restricting the
connection $\nabla$ on $S$ to sections along $\partial\Omega$.  Thus the
boundary Dirac bundle $S_{\partial \Omega}$, carries a canonical
Dirac operator denoted by $\Dsl_{\partial \Omega}$.  

Notice that $\partial
\Omega$ is a manifold without boundary, thus the boundary Dirac operator
is essentially self-adjoint and has a unique self-adjoint extension
(see \cite{La89}, Thm. 5.7).  We shall use this fact later on.   

We will denote as before by $L^2 (\partial \Omega,  S_{\partial \Omega})$ the Hilbert space of
square integrable sections of $S_{\partial \Omega}$ and by $\langle
\cdot ,\cdot \rangle_{\partial\Omega}$ its Hilbert product structure:
\begin{equation}\label{bounD} 
\langle \phi , \psi \rangle_{\partial \Omega} =
\int_{\partial \Omega} (\phi(x), \psi (x) )_x \, \vol_{\partial\eta} (x) \, , \qquad \psi, \phi \in L^2(\partial \Omega,  S_{\partial \Omega})\, .
\end{equation}
Because of Lions trace theorem \cite{Li72} the restriction map $i^* \colon \Gamma (S) \to \Gamma (S_{\partial \Omega})$,
$\sigma \mapsto \phi =i^* \sigma$, extends to a continuous linear map:
$$
b\colon H^1(\Omega,S) \to H^{1/2} (\partial \Omega, S_{\partial \Omega}) \, .
$$

The Hilbert space $H^{1/2} (\partial \Omega, S_{\partial\Omega})$ will be called the Hilbert space of
boundary data for the Dirac operator $\Dsl$ and will be denoted in what follows by $\sp
H_D$.  It carries an interesting extra geometrical structure induced by
the boundary form  defined by Green's formula (\ref{int_part_D1}) responsible for the non self-adjointeness of the Dirac operator $\Dsl$ in $H^1 (\Omega, S)$. 
\begin{equation}\label{int_part_D2} 
\Sigma (\phi , \psi ) = \int_{\partial \Omega}
( \nu(x)\cdot \phi (x), \psi (x) )_x \, \vol_{\partial \eta} (x) \,.
\end{equation}

The normal vector field $\nu$ defines an automorphism of the Dirac bundle $S_{\partial \Omega}$ and
an isomorphism:
$$
\nu\colon \Gamma (S_\pO ) \to \Gamma (S_\pO ) \, ,
$$
by $\nu (\phi ) (x) = \nu (x) \cdot \phi (x)$, for all $x\in \partial
\Omega$, $\phi\in \Gamma (S_\pO )$.   Such automorphism extends to a
continuous complex linear operator of $\sp H_D$ denoted now by $J$.  
Because $\nu^2 = - 1$ in the Clifford algebra, such operator $J$ verifies
$J^2 = -I$. In addition, because of the Dirac bundle structure, eq.
(\ref{unit_cl}), $J$ is an isometry of the Hilbert space product, 
\begin{equation}
 \langle J\phi, J\psi \rangle_{\partial\Omega} =  \langle \phi,\psi
\rangle_{\partial\Omega}, ~~~~~ \forall \phi, \psi \in \sp H_D \, ,
\label{complexstructure}
\end{equation}
i.e., $J$ defines a compatible complex structure on
$\sp H_D$.  In general, a complex Hilbert space
$\sp H$ with inner product $\langle \cdot, \cdot \rangle$ and a compatible
complex structure $J$ defines a new continuous bilinear form $\omega$ by
setting,
$$ 
\omega (\varphi,\psi) = \langle J\varphi, \psi\rangle \, , \qquad \forall
\varphi,\psi \in \sp H \, .
$$
Such structure is skew-hermitian in the sense that 
$$
\omega (\varphi , \psi ) = - \overline{\omega (\psi, \varphi) } \,.
$$
(Notice that in the real case $\omega$ will define a symplectic structure on $\sp H$.)  
We will call such structure $\omega$ symplectic-Hermitean and the corresponding 
linear space a symplectic hermitian linear space (see for instance \cite{Ha00} for 
a discussion of finite-dimensional Hermitian symplectic geometry).

The compatible complex structure allows to decompose the Hilbert space
$\sp H$ as 
\begin{equation}\label{H+H-}
\sp H = \sp H_+ \oplus \sp H_- \, ,
\end{equation} 
where $\sp H_\pm$ are the $\mp i$
eigenspaces of $J$, that is $\phi_\pm \in \sp H_\pm$ if $J\phi_\pm = \mp
i \phi_\pm$.  The subspaces $\sp H_\pm$ are orthogonal as the following
computation shows. 
$$\langle \phi_+ , \phi_- \rangle = \langle J\phi_+ , J\phi_- \rangle = 
\langle i\phi_+ , -i\phi_- \rangle = - \langle \phi_+ , \phi_- \rangle$$
The Hilbert space
$\sp H$ will carry in this way a natural decomposition in two orthogonal
infinite dimensional closed subspaces, i.e., a polarization.   Notice that 
the Hilbert space $\sp H_D$ carries already another complex structure, 
denoted by $J_0$, multiplication by $i$, and both complex structures are 
also compatible in the sense that $[J,J_0] = 0$.    
because of the bundle $S$ is a $\Cl(\Omega)$ complex bundle.  In addition 
this implies that
$$J(\sp H_+ \pm i\sp H_-)\subset \sp H_\pm. $$   

Hence, the Hilbert space of boundary data $\sp H_D$ for the Dirac operator
$\Dsl$ is a polarized Hilbert space carrying a compatible complex
structure $J_D$ and the corresponding symplectic-hermitian structure $\omega_D$.  Using
these structures the boundary form $\Sigma$ is written as,
\begin{equation}\label{bound_D} 
\Sigma (\xi, \zeta)= \omega (b(\xi), b(\zeta)) = \langle Jb(\xi), b(\zeta) \rangle_{\partial\Omega} \, , \qquad \forall \xi, \zeta\in \sp H^1 (S) .
\end{equation}

From this characterization we immediately see that self-adjoint
extensions $\Dsl_s$ of $\Dsl$ will be obtained in domains $\dom (\Dsl_s)$ such that
their boundary image, $b(\dom (\Dsl_s))$ will be isotropic subspaces of
$\omega_D$, thus vanishing the r.h.s. of Eq. (\ref{bound_D}).  Moreover to
be self-adjoint, such domains must verify that $b(\dom (\Dsl_s)) = b(\dom
(\Dsl_s^\dagger))$, thus, they must be maximal subspaces with this property. 
We have thus proved the first part of the following theorem,

\begin{theorem}\label{unit_D}  Self-adjoint extensions of the Dirac
operator $\Dsl$ are in one to one correspondence with maximally closed
$\omega_D$-isotropic subspaces of the boundary Hilbert space $\sp H_D$. 
The domain of any of these extensions is the inverse image by the
boundary map $b$ of the corresponding isotropic subspace.   Moreover, 
each maximally closed $\omega_D$-isotropic subspace $W$ of $\sp H_D$
defines an isometry $U\colon \sp H_+\to \sp H_-$, called the Cayley 
transform of $W$ and conversely.
\end{theorem}

\begin{proof} Let $W$ be a closed $\omega_D$-isotropic subspace of $\sp H_D$. 
Then, $b^{-1} (W)$ is a closed subspace of $H^1 (\Omega, S)$ containing $H_0^1
(\Omega, S)$.  Let $\Dsl_W$ be the extension of $\Dsl$ defined on  $b^{-1} (W)$ and
compute $\Dsl_W^\dagger$.  If $b(\xi),b(\zeta) \in W$, then $\langle
\Dsl_W^\dagger\xi, \zeta\rangle = \langle \xi, \Dsl_W \zeta \rangle + \omega
(b(\xi), b(\zeta)) = \langle \xi, \Dsl_W \zeta \rangle$ because $W$ is
$\omega_D$ isotropic.  This shows that $b^{-1}(W) \subset \dom
(\Dsl_W^\dagger)$. If there were $\xi\in \dom (\Dsl_W^\dagger) - b^{-1}(W)$,
then, the same computation shows that $\omega_D (b(\xi), \phi) = 0$
for all $\phi\in W$, and the subspace $W^\prime = W \oplus \langle
b(\xi) \rangle$ will be $\omega_D$-isotropic, which is contradictory.  Thus
$\dom (\Dsl_W) = \dom (\Dsl_W^\dagger)$ and the extension is self-adjoint.  
The converse is proved similarly.

Let us consider now a closed maximal $\omega_D$-isotropic subspace $W$.
Let us show that $W$ is transverse to $\sp H_\pm$.  Let $\phi\in W\cap
\sp H_\pm$, then $0= \omega_D (\phi, \phi) = \langle J\phi, \phi \rangle
= \mp i || \phi ||^2$, then $\phi = 0$.  Then, the subspace $W$ defines
the graph of a continuous linear operator $U\colon \sp H_+ \to \sp
H_-$ and vectors $\phi = \phi_+ + \phi_- \in W$ have the form $\phi_- =
U\phi_+$.  Then, the $\omega_D$-isotropy of $W$ implies,
\begin{eqnarray*}
 0 &=& \omega_D (\phi_+ + U\phi_+, \psi_+ + U\psi_+) \\ &=&  \langle i\phi_+
- iU\phi_+ , \psi_+ + U\psi_+ \rangle = -i \langle \phi_+, \psi_+
\rangle + i\langle U\phi_+ , U\psi_+ \rangle \, ,
\end{eqnarray*}
for every $\phi_+, \psi_-\in \sp H_+$,
that proves that $U$ is unitary.  
\end{proof}

%%%%%%%%%%%%%%%%%%%%%%%%%%%%%

\subsection{The Cayley transform on the boundary}\label{Cayley_D}

Theorem \ref{unit_D} is the boundary analogue of Von Neumann's description 
of selfadjoint extensions by means of isometries between the deficiency 
spaces $\sp N_\pm$ (see later on).
In spite of the inherent interest of this result, it is well-known that sometimes
it is more useful to have an alternative description of such extensions in 
terms of selfadjoint operators.  For that we
will use the
Cayley transform.  We shall define the Cayley transform on the
polarized boundary Hilbert space $\sp H_D = \sp H_+ \oplus \sp H_-$ as the
isomorphism
$C \colon \sp H_D \to \sp H_D$ defined by $C(\phi_1, \phi_2) = (\phi_1
+ i\phi_2, \phi_1 - i \phi_2)$, for every $\phi_1\in \sp H_+$, $\phi_2
\in \sp H_-$.  The complex structure $J_D$ is transformed into
$C^{-1}J_D C (\phi_1, \phi_2) = (-\phi_2, \phi_1)$ and the
symplectic-hermitian structure $\omega_D$ is transformed into the
bilinear form
\begin{equation}\label{omega_C} \sigma_D (\phi_1,\phi_2; \psi_2, \psi_2) = 2i (\langle
\phi_2,
\psi_2 \rangle - \langle \phi_1 , \psi_2 \rangle ) .
\end{equation}
    
Let $U$ be a unitary operator $U\colon \sp H_+ \to \sp H_-$ such that
$I - U$ is invertible.   Then we have $\phi_- = U\phi_+$ and using the
Cayley transform, $\phi_\pm = \phi_1 \pm i\phi_2$, we will obtain that, 
$$\phi_1 = i\frac{I+U}{I-U} \phi_2 ,$$
that defines an operator $A_U\colon \sp H_+ \to \sp H_-$.
In general it will be more convenient to consider graphs of operators
$\sp H_+ \to \sp H_-$ because 
$C$ acts on $\sp H_D$, then it actually maps subspaces of $\sp H_D$ in subspaces
of $\sp H_D$.  

If $W$ is a subspace of $\sp H_D$, then the adjoint of
$W$ will be the subspace $W^\dagger$ defined by setting,
$$ W^\dagger = \set{(\psi_1, \psi_2)\in \sp H_+\oplus \sp H_- \mid
\langle \phi_1, \psi_2 \rangle = \langle \phi_2, \psi_1 \rangle, ~~
\forall (\phi_1, \phi_2)\in \sp H_+ \oplus \sp H_-} .$$
The subspace $W$ is said to be symmetric if
$W\subset W^\dagger$ and self-adjoint if $W = W^\dagger$. 
(see \cite{Ar61} for more details on
operational calculus with closed subspaces of a Hilbert space). 
We will say that an operator $A\colon \sp H_+ \to \sp H_-$ is selfadjoint 
if its graph is a selfadjoint subspace of $\sp H_D$.
In this sense it is obvious that the Cayley transform operator $A_U$ of 
$U$ is selfadjoint.   
Moreover, it is clear that
self-adjoint subspaces are maximally isotropic subspaces of the bilinear
form $\sigma_D$ given by Eq. (\ref{omega_C}).  But $\sigma_D$ is the
transformed bilinear form on $\sp H_D$ by the Cayley transform, then,
maximally $\sigma_D$-isotropic subspaces correspond to maximally
$\omega_D$-isotropic subspaces, that is the Cayley transform maps
one-to-one graphs of isometries $U\colon \sp H_+ \to \sp H_-$ into
self-adjoint subspaces of $\sp H_D$.

Let us denote by $K\colon \sp H_+ \to \sp H_-$ a compact operator, we
denote by $W+K$ the subspace of $\sp H_D$ given by $\set{(\phi,\psi + 
K(\phi) ) \mid (\phi,\psi) \in W}$ and we will call it a compact deformation 
of $W$.   If we denote by $\sp M_D$ the space of self-adjoint subspaces of $\sp H_D
= \sp H_+ \oplus \sp H_-$.  We will introduce a topology in $\sp M_D$ as 
follows.  We shall define base for the topology by the family of balls 
given by the sets of subspaces of the form $\sp O_{W,\epsilon} = \set{ 
W+ K \mid \parallel K \parallel < \epsilon, (W+K)^\dagger = W + K }$.  
Hence $\sp M_D$ becomes a topological space (in fact as we will see later a 
smooth manifold).

We can summarize the previous discussion in
the following theorem.

\begin{theorem}\label{sa_Dirac}  The Cayley transform defines a homeomorphism between the 
space of isometries $U(\sp H_+, \sp H_-)$ from $\sp H_+$ to $\sp H_-$ with 
the operator topology
and the space of self-adjoint subspaces of $\sp H_D = \sp H_+ \oplus \sp
H_-$.  Moreover, the self-adjoint
extensions of the Dirac operator $\Dsl$ are in one-to-one correspondence
with the self-adjoint subspaces in $\sp M_D$.
\end{theorem}

%%%%%%%%%%%%%%%%%%%%%%%
%%%%%%%%%%%%%%%%%%%%%%%

\subsection{Simple examples: The Dirac operators in 1 and 2 dimensions}

Consider $\Omega$ a 1-dimensional manifold with boundary, hence a connected
component may be either the half-line or a closed interval in the compact case. 
We may assume without loss of generality that the metric is constant.
Then the Clifford bundle becomes the trivial bundle complex line
bundle over $\Omega$, as $\mathrm{Cl}(1)$ is algebra isomorphic to $\mathbb{C}$.
Hence Dirac's operator is simply $\Dsl = i\partial/\partial x$, i.e., the momentum operator.

Notice that if $\Omega$ is the half-line $[ 0, +\infty )$, there are no self-adjoint extensions of $D$.
The Hilbert space of boundary data is simply $\mathcal{H}_D = \mathbb{C}$.  Then,
Eq. \eqref{int_part_D1} becomes:
$$
\langle \Dsl\xi, \zeta\rangle - \langle \xi , \Dsl\zeta \rangle =  - i \bar\xi (0) \zeta(0) \, ,
$$
but, the only symmetric extensions are given by the condition $\xi ( 0 ) = 0$, i.e.,
functions in $H_0^1(\Omega)$, but such domain defines just a symmetric operator
in full agreement with von Neuman's theorem (see later Sect. \ref{section:Neumann}).

However, if $\Omega$ is a closed interval $[a,b]$, then $\mathcal{H}_D \cong \mathbb{C}^2$ and
$$
\langle \Dsl\xi, \zeta\rangle - \langle \xi , \Dsl\zeta \rangle = i (\bar\xi(b)\zeta(b) - \bar\xi(a)\zeta(a) )\, .
$$
Notice that $\mathcal{H}_\pm \cong \mathbb{C}$
and the space of self-adjoint extensions is given by unitary maps from $\mathcal{H}_+$ to $\mathcal{H}_-$,
that is $U(1)$ (see also Sections \ref{squaring} and \ref{Laplace1D} for a discussion on the relation between the self-adjoint extensions of
$\Dsl$ and the Laplace operator $\Delta$)

Dirac operators in 2D have been the subject of exhaustive research both
from the mathematical and physical side.   We will consider $\Omega$ to
be now a two dimensional compact orientable manifold, i.e., a Riemann 
surface $\Sigma$ with boundary $\partial \Sigma = \cup_{\alpha = 1}^r 
S_\alpha$ where $S_\alpha \cong S^1$.  Consider now a Dirac bundle $S\to 
\Sigma$.   Notice that because of the general previous considerations $S$ 
carries a representation of the Clifford algebra bundle $\Cl(\Sigma )$.  

Take a point $p\in \Sigma$ and a local chart $(U,p)$ around it with 
local complex coordinates $z = x + iy$.   The tangent space $T_p\Sigma$ is 
spanned by the local vector fields $e_1 = \partial /\partial x$, $e_2 = \partial / 
\partial y$, and the two vectors can be taken to be orthonormal.  Hence the 
Clifford algebra at $p$ is generated by $e_1, e_2$ with the relation:
$$ 
e_1 \cdot e_2 + e_2 \cdot e_1 = 0\, , \qquad  e_1^2 = -1\, , \quad  e_2^2 = -1 \, .
$$
The generators $e_1,e_2$ of the Clifford algebra act on the tangent bundle 
as the $2\times 2$ matrices 
$$
\sigma_1 = 
\left( \begin{array}{rr} 0 & 1 \\ -1 & 0 \end{array}\right) \, , \qquad \sigma_2 = \left( \begin{array}{rr} 0 & i \\ i & 0 \end{array}\right)  \, .
$$
This representation happens to be the spinor representation of the spin 
group $\mathrm{Spin} (1)$ which is a double covering of $U(1)$ (the covering $U(1) 
\to U(1)$ given by the map $z\mapsto z^2$).  Thus if we take the 
complexified tangent space $T\Sigma\otimes\C$ as a Dirac bundle then we 
will have for the Dirac operator:
$$ \Dsl = e_1 \cdot \nabla_{e_1} + e_2\cdot \nabla_{e_2} = \matriz{cc}{0 & 
\bar{\partial} + A \\ \partial + \bar{A} & 0} ,$$
with $\partial$, $\bar{\partial}$ being the Cauchy--Riemann operators on 
the variables $x,y$, and $A$ the homogeneous part of the Levi--Civita 
connection corresponding to the given Riemannian metric on $\Sigma$.  

The $\Cl (2)$ representation above is reducible because $\Gamma_5 = i 
e_1e_2$ anticommutes with $e_1$ and $e_2$.  In fact, $\Gamma_5 = \sigma_3$ 
which has eigenvalues $\pm 1$.   The representation space of $\Cl(2)$ 
decomposes into two subspaces $S^\pm$ of eigenvectors of $\pm 1$ 
respectively.   The operator $\Gamma_5$ is known as the chirality operator, 
and the spaces of sections of $S^\pm$ are the right-handed and left-handed 
fermions respectively.   Notice that $\Gamma_5 \psi^\pm = \pm \psi^\pm$, 
and the chiral projectors, i.e., the orthogonal projectors into the 
orthogonal subspaces $S^\pm$ are given by $\Pi_\pm = 1/2 (I \pm \Gamma_5)$.  
Thus given any section $\psi$ of $T\Sigma^\C$, we have $\psi = \psi^+ + 
\psi^-$, $\psi^\pm = \Pi_\pm \psi$.  Notice finally that $[\Dsl,\Gamma_5]_+ = 
0$, where $[\cdot, \cdot]_+$ denotes the anticommutator, hence $\Dsl\colon 
S^\pm \to S^\mp$, and $\Dsl$ exchanges the irreducible representation 
spaces as it is apparent from the block structure of the matrix 
representing $\Dsl$ above.  

Assuming that the boundary is connected, i.e.,  $\partial \Sigma \cong S^1$, then 
the space of boundary data $\mathcal{H}_D$ is given as $H^{1/2}(S^1, S_{\partial \Omega})$ which
is isomorphic to the space of functions of Sobolev class $1/2$ on $S^1$ with values in $\mathbb{C}^2$.
Clifford multiplication by $\nu$, the outward normal determines a polarisation $\mathcal{H}_D = \mathcal{H}_+ \oplus \mathcal{H}_-$ as 
described in Eq. \eqref{H+H-}.  Thus, according with Thm. \ref{unit_D} the space of self-adjoint extensions of $\Dsl$ 
is given by the family of unitary operators $\mathcal{U}(\mathcal{H}_+ , \mathcal{H}_-)$.   
Notice that each one of the subspaces $\sp H_\pm$ is isomorphic to $H^{1/2}(S^1)$, hence
the space of self-adjoint extensions  can be identified with the group of 
unitary operators of the Sobolev space $H^{1/2}(S^1)$.

\newpage

%%%%%%%%%%%%%%%%%%%%%%%%%%%%%%
%%%%%%%%%%%%%%%%%%%%%%%%%%%%%%
%%%%%%%%%%%%%%%%%%%%%%%%%%%%%%

\section{Self-adjoint extensions of Laplace operators}\label{section:Laplace}

%%%%%%%%%%%%%%%%%%%%%%%%%%%%%%
%%%%%%%%%%%%%%%%%%%%%%%%%%%%%%

%\bigskip

\subsection{The covariant Laplacian}%	\label{section:Laplace}

We will start the discussion of the theory of self-adjoint extensions of
second order differential operators considering a particular situation
of big physical interest.  Later on, we will extend this considerations
to more general families of second order operators and we will find the
relation with the theory of extensions for Dirac operators studied in
the first part.

We will consider a physical system 
subject to the action of a Yang--Mills field.   The configuration space
of the system will be a compact connected smooth Riemannian manifold
$(\Omega ,\eta)$ with smooth boundary $\partial \Omega$ \footnote{The
theory can be generalized for boundaries with singularities and noncompact
manifolds under appropriate regularity conditions}.   We will consider a
hermitian vector bundle $E\to \Omega$ with hermitian product $(\cdot , \cdot)$. 
The Yang--Mills potential will be a connection 1-form $A$ with values on
$\mathrm{End} (E)$.   The space of smooth sections of
$E$ will be denoted by $\Gamma (E)$ and the covariant differential
operator $d_A = d + A$ will map $\Gamma (E)$ to $\Gamma (E) \otimes
\Lambda^1 (\Omega )$, i.e., it will map sections of $E$ to 1--forms in
$\Omega$ with values on $E$.   

We will define the $L^2$-product $\langle .,. \rangle$ in $\Gamma
(E)$ as usual,
\begin{equation}\label{riemp}  \langle \psi_1, \psi_2 \rangle =
\int_\Omega \, (\psi_1(x), \psi_2(x))_x \, \vol_\eta (x),
\end{equation}
where $\vol_\eta$ denotes the Riemannian volume defined
by the metric $\eta$.   We define
$L^2 (E)$ as the completion with respect to the norm $\parallel .
\parallel_2$ of the space $\Gamma (E)$.

Similarly, we will define the product of two $k$--forms $\alpha$, $\beta$
on $\Omega$ with values on $E$ by the formula,
$$ \langle \alpha, \beta \rangle = \int_\Omega (\alpha_{i_1\ldots
i_k}(x), \beta^{i_1\ldots i_k}(x) ) \vol_\eta (x) ,$$
where we have used the metric $g$ to raise the subindexes on the
components of $\beta$.  We define the Hodge operator $\star$ as a map
from $\Gamma (E) \otimes \Lambda^k (\Omega)$ to $\Gamma (E) \otimes
\Lambda^{n-k} (\Omega)$, defined by
$$\star \alpha (x) \wedge \beta (x)  = (\alpha_{i_1\ldots i_k}(x),
\beta^{i_1\ldots i_k}(x) ) \eta_g (x) ,$$
and thus,
$$ \langle \alpha, \beta \rangle = \int_\Omega \star \alpha \wedge \beta
.$$ We will consider the completion of $\Gamma (E) \otimes \Lambda^1
(\Omega )$ with respect to the norm $||\alpha ||_2^2 =\langle
\alpha, \alpha \rangle$ and then define the adjoint to the operator $d_A$
with respect to this Hilbert space structure, i.e.,
$$ \langle d_A^\dagger \alpha, \psi \rangle = \langle \alpha, d_A \psi
\rangle ,$$ for all $\psi \in \Gamma (E)$ (notice that $d_A^\dagger$ is
well defined because $\Gamma (E)$ is dense in $L^2(\Omega,E)$).

The quantum Hamiltonian for a particle moving in the presence of the
Yang--Mills potential $A$ is formally given by the second order
differential operator,
\begin{equation}\label{qham} \H  = -\frac{1}{2} \Delta_A + V,
\end{equation}
where 
$$\Delta_A =d_A^\dagger d_A + d_A d_A^\dagger ,$$ 
is the covariant Laplace-Beltrami operator and $V$ is a smooth function on
$\Omega$.

Clearly, the operator $\Delta_A$ is closable and symmetric on
$\Gamma_0^\infty (E)$.  Its closure is defined on the domain $H_0^2 (\Omega, E)$, 
i.e., the completion of $\Gamma_0^2({\buildrel \circ \over E})$ with
respect to the Sobolev norm $||  \cdot ||_{2,2}$ defined in Eq.
(\ref{sobov_k}) with $k=2$. In fact, such space is the space of sections
of $E$ vanishing on
$\partial\Omega$, such that they possess first and second weak derivatives
which are in $L^2(\Omega,E)$.  We will denote this operator by $\Delta_0$,
$\dom \Delta_0 = H_0^2(\Omega,E)$.

The operator $\Delta_A$ has another extension, the largest space where
it can be defined, the closure of $\Gamma (E)$ with respect to
$||.||_{2,2}$.  This domain is the Sobolev space $H^2(\Omega,E)$ and it is easy
to check that the adjoint of $\Delta_0$ is precisely $\Delta_A$ with
domain $H^2(\Omega,E)$, thus $\dom \Delta_0^\dagger = H^2(\Omega,E)$,
$\Delta_0^\dagger = \Delta_A$.  On the other hand the smooth function
$V$ defines a essentially self-adjoint operator on $H^2(\Omega,E)$.  
If we denote by $\H_0$ the operator defined by $\H$ with domain $H_0^2(\Omega,E)$
it is possible to check that $\H_0^\dagger = \H$, i.e., the domain
of the adjoint operator of $\H_0$ is $H^2(\Omega,E)$.  
Hence, the operator $\H$ defined on $H^2 (E)$ is not self--adjoint in
general.  We will be interested in finding extensions $\H_s$ of $\H_0$
with domain $\dom \H_s$ such that 
$$
H_0^2 (\Omega, E) = \dom \H_0 \subset \dom \H_s = \dom \H_s^\dagger \subset
\dom \H = H^2(\Omega,E),$$
and $\H_s \psi = \H_0 \psi$, for any $\psi\in H_0^2 (\Omega, E)$.
To do that, we will integrate by parts as follows.
\begin{equation}\label{simet} \langle \psi_1, \H \psi_2\rangle =
\langle \H \psi_1, \psi_2\rangle + \mitad\Sigma\left(\psi_1,
\psi_2\right),
\end{equation}
for any smooth sections $\psi_1,\psi_2\in \Gamma (E)$, where by Stokes
theorem and the parallelism of $(.,.)$ with respect to the connection
defined by $A$, we get:
\begin{equation}\label{bound_term} \Sigma\left(\psi_1, \psi_2\right) =
\int_{\partial\Omega} \,
i^*[\left( \star d_A\psi_1 , \psi_2 \right) - 
\left( \psi_1 ,\star d_A \psi_2 \right) ].
\end{equation}
Notice that $\star d_A \psi$ is an
$(n-1)$--form on $\Omega$ with values in $E$, and then $\left( \star d_A
\psi_1, \psi_2\right)$ denotes the $(n-1)$--form on $\Omega$ obtained
by taking the product of the values of $\star d_A \psi_1$ and $\psi_2$
on the fibres of $E$. 

The boundary term $\Sigma (\psi_1,\psi_2) $ has a relevant physical
interpretation. It measures  the net total probability flux across the
boundary. If the operator $\H $ were self-adjoint this flux would have
to vanish: the incoming flux would have to be equal to  the outcoming
flux because the evolution operator $\exp it\H$ in such a case would be
unitary and preserve probability. 

We will characterize first self-adjoint extensions of $\H_0$ by using the
geometry of some Hilbert spaces defined on the boundary
$\partial \Omega$ of $\Omega$ and later on, we will compare with the
classical theory of von Neumann.

%%%%%%%%%%%%%%%%%%%%%%
%%%%%%%%%%%%%%%%%%%%%%

\subsection{Complex structures on the Hilbert space of boundary data}\label{section:maximal}

We will denote again by $E_\pO $ the restriction of the bundle $E$ to
$\partial \Omega$.   Let us denote by $\vol_{\partial\eta}$ the
Riemannian volume form defined by the restriction $\partial \eta$ of the
Riemannian metric $\eta$ to the boundary $\partial \Omega$.  Then we will
denote by $L^2(\Omega,E_\pO)$ the Hilbert space of square integrable
sections of the bundle $E_\pO $ and by $\langle \cdot, \cdot \rangle_{\partial
\Omega}$ its Hilbert product structure given by,
\begin{equation}\label{boundD}
\langle \psi, \varphi \rangle_{\partial\Omega} = \int_{\partial \Omega} (\psi(x), \varphi(x)) \mathrm{vol}_{\partial\Omega} \, .
\end{equation}

To any section $\Psi \in \Gamma (E)$ we can associate two sections
$\varphi, \dot{\varphi} \in \Gamma (E_\pO )$ as follows,  
\begin{eqnarray}\label{nor}
i^*\Psi &=& \varphi \nonumber \\
i^*(\star d_A\Psi) &=& \dot\varphi \,
\vol_{\partial \eta} ,\end{eqnarray}
in other words, the function $\varphi$
is the restriction to $\partial \Omega$ of $\Psi$, and
$\dot{\varphi}$ is the normal derivative of $\Psi$ along the boundary. 

Thus the restriction map $b\colon \Gamma (E) \to \Gamma (E_\pO )
\times \Gamma (E_\pO )$, assigning to each section its Dirichlet
boundary data, $\Psi \mapsto b(\Psi) = (\varphi, \dot\varphi)$, extends
to a continuos linear map,
$$
b \colon H^2(\Omega,E)\to H^{3/2} (\partial \Omega, E_\pO )Ê\oplus H^{1/2} (\partial \Omega, E_\pO ) \, ,
$$ 
from the Hilbert space $H^2(\Omega,E)$ to the direct sum Hilbert space of
boundary data $H^{3/2}(\partial \Omega,E_\pO)$ and
$H^{1/2}(\partial \Omega,E_\pO)$ equipped with the canonical direct sum product,
\begin{equation}\label{hilb}
\langle (\varphi_1,\dot\varphi_1), (\varphi_2,\dot\varphi_2) \rangle =
\langle \varphi_1, \varphi_2 \rangle_{\partial \Omega}  + \langle
\dot{\varphi}_1, \dot{\varphi}_2 \rangle_{\partial \Omega} .
\end{equation} 
We will call the
Hilbert space $H^{3/2}(\Omega,E_\pO) \oplus H^{1/2}(\Omega,E_\pO)$ the Hilbert space
of boundary data of the Laplacian operator and we will denote it by $\sp
H_L$.

To avoid the difficulties arising from the different spaces contributing
as factors to $\sp H_L$ we will restrict the second factor to
$H^{3/2}(\partial\Omega,E_\pO)\subset H^{1/2} (\partial \Omega, E_\pO )$, or we imbed both of them into
$L^2(\Omega,E_\pO)$.   As it happens with the boundary data space $\sp
H_D$ for the Dirac operator, the boundary Hilbert space
$\sp H_L$ has an interesting extra geometrical structure, a canonical
compatible complex structure $J_L$.   Apart from the natural complex
structure inherited from $L^2(\Omega,E_\pO)$ there is another complex
structure defined by
\begin{equation}\label{J_B}
J_L (\varphi, \dot\varphi) = (\dot\varphi, - \varphi ) \, ,
\end{equation}
which clearly is unitary and verifies $J^\dagger = - J$.
Then the composition with the Hilbert space product defines a
skew-pseudo-hermitian product $\omega_L$ on $\sp H_L$, i.e.,
\begin{equation}\label{sigma_B}
\omega_L ((\varphi_1,\dot\varphi_1), (\varphi_2,\dot\varphi_2)) = \langle 
J_L(\varphi_1,\dot\varphi_1), (\varphi_2,\dot\varphi_2)\rangle = 
\langle \dot\varphi_1,\varphi_2\rangle - \langle
\varphi_1,\dot\varphi_2 \rangle .
\end{equation}
For practical purposes it is sometimes convenient to redefine $\omega_L$
as $\sigma_L = -\frac{i}{2}\omega_L$, then $\sigma_L$ defines a
pseudo-Hermitean structure on $\sp H_L$. 
The skew-pseudo-Hermitean structure $\omega_L$ defined on the boundary data
has been sometimes called the Lagrange form because of Lagrange's
identity in the one--dimensional case \cite{Na68}.

\bigskip

Then, the boundary term for the Laplace operator obtained in eq.
(\ref{bound_term}) can be simply written as
$$ \Sigma (\psi_1, \psi_2) = \omega_L (b(\psi_1), b(\psi_2)) ,$$
showing explicitly its geometrical nature and allowing a direct
characterization of all self-adjoints extensions of $\H_0$.

Let $\H_s$ be a self--adjoint extension of $\H_0$.  It is obvious that
because of Eq. (\ref{simet}), $b(\dom (\H_s))$ is a subspace of $\sp H_L$
such that $\omega_L\mid_{b(\dom (\H_s))} = 0$, i.e., $b(\dom (\H_s))$ is
an isotropic subspace with respect to the skew-pseudo-hermitian structure
$\omega_L$.  

\begin{theorem}\label{maximal}
There is a one-to-one correspondence between self-adjoint extensions of
the operator $\H_0$ and maximal $\omega_L$-isotropic
subspaces of $\sp H_L$.   Moreover, such self-adjoint extensions are in
one-to-one correspondence with self-adjoint subspaces of $\sp H_L$ and
the eigenspaces of $J_L$ provide a polarization of the boundary Hilbert
space $\sp H_L = \sp H_+\oplus \sp H_-$ such that the self-adjoint
extensions of $\H_0$ are in one-to-one correspondence with the unitary
operators $\mathcal{U} (\sp H_+, \sp H_-)$. 
\end{theorem}

\begin{proof} The first part of the theorem is proved
very much like Thm. \ref{unit_D} If
$\H^\prime$ is self--adjoint, we have seen that $b(\sp D^\prime)$ is isotropic.  If it
were not maximal, this would imply that there is an isotropic subspace
$\sp I$ that contains properly the domain of $\H^\prime$, $b(\sp D^\prime)
\subset \sp I$.  Then
$b^{-1}(\sp I)$ is the domain of an extension ($b$ is continuous) of
$\H_0$ and it verifies
$$ \sp D^\prime \subset b^{-1}(\sp I ) \subset b^{-1}(\sp I) \subset
(\sp D^\prime)^\dagger .$$ 
In fact, the isotropy of $\sp I$ implies that,
$$\langle \psi_1 , \H^\prime \psi_2 \rangle = \langle \H^\prime \psi_1,
\psi_2 \rangle , ~~~~~\forall \psi_1, \psi_2 \in b^{-1}(\sp I) ,$$
because $\Sigma_B \mid_{\sp I} = 0$.

Conversely, if $\sp I$ is a maximal isotropic subspace of $\sp H_B$ is
evident that the subspace $b^{-1}(\sp I)$ defines the domain of a
self--adjoint extension of $\H_0$.  In fact, if $b^{-1}(\sp I) \subset
(b^{-1}(\sp I))^\dagger$, then $b((b^{-1}(\sp I))^\dagger)$ is an
isotropic subspace of $\sp H_B$ containing $\sp I$.  Thus, because $\sp
I$ is maximal, they must coincide.  Then, $D_{\sp I} = b^{-1}(\sp I)$
defines a self--adjoint extension of $\H_0$. 

Before proving the other statements, we will introduce the analogue of
the Cayley transformation discussed in \S \ref{Cayley_D} for the Laplace
operator.
\end{proof}

%%%%%%%%%%%%%%%%%
%%%%%%%%%%%%%%%%%

\subsection{The Cayley transformation on the boundary Hilbert space for
Laplacian operators}\label{section:Cayley}

We will concentrate now in the boundary Hilbert space $\sp H_L$ with its
skew-pseudo-hermitian structure $\omega_L$.  It is obvious that the
subspaces 
$\sp L_+ = L^2 (E_\pO ) \times \{ {\bf 0} \}$ $= \set{(\varphi, {\bf
0}) \mid \varphi \in L^2 (E_\pO )}$ and $\sp L_- = \{ {\bf 0}
\} \times L^2 (E_\pO )$ $= \set{({\bf 0}, \dot\varphi) \mid
\dot\varphi \in L^2 (E_\pO )}$ are isotropic and they paired by
$\omega_L$.  In fact, 
$$\omega_L (\varphi_1,{\bf 0}; {\bf 0},\dot{\varphi}_2) = \mitad
\langle \varphi_1, \dot{\varphi}_2 \rangle_{\partial \Omega}  .$$
Thus the block structure of $\omega_L$ with respect to the isotropic
polarization $\sp H_L = \sp L_+ \oplus \sp L_-$ is
$$\omega_L = \matriz{c|c}{ 0 & \mitad \langle
.,.\rangle_{\partial \Omega} \\ \hline -\mitad \langle
.,.\rangle_{\partial \Omega} & 0} .$$

It will be convenient for further purposes to introduce another
pseudo-Hermitean product on $\sp H_L$ that corresponds to the
diagonalization of $\omega_L$, hence of $J_L$.  We shall define the Cayley
transform on the boundary as the map $C \colon \sp H_L \to\sp H_L $,
given by, 
\begin{equation}\label{cayley} C (\varphi, \dot\varphi ) =
(\varphi + i \dot\varphi, \varphi - i \dot\varphi), ~~~~~ \forall
(\varphi, \dot{\varphi}) \in \sp H_L.
\end{equation}
We will denote by $\varphi^\pm = \varphi \pm i \dot{\varphi}$ the
components of $C(\varphi, \dot{\varphi})$.  It is clear that $C$ is an
isomorphism of linear spaces and transforms the pseudo-Hermitean product
$\omega_L$ into the skew-pseudo-hermitian product
$\langle .,. \rangle_{\partial \Omega} \ominus
\langle .,. \rangle_{\partial \Omega}$, i.e., 
\begin{equation}\label{sigma}
\sigma_L (\varphi_1^+, \varphi_1^-; \varphi_2^+, \varphi_2^- ) = 
\langle \varphi_1^+,\varphi_2^+ \rangle_{\partial \Omega} - \langle
\varphi_1^-, \varphi_2^- \rangle_{\partial \Omega} , ~~~~~\forall
(\varphi_a^+, \varphi_a^-) \in \sp H_L, a = 1,2 .
\end{equation}
and 
$$ \omega_L (\varphi_1, \dot\varphi_1 ; \varphi_2, \dot\varphi_2 ) =
\sigma_L (\varphi_1^+, \varphi_1^-; \varphi_2^+, \varphi_2^- ) .$$ 

\medskip

Thus, the complex structure $J_L$ is transformed into diagonal form,
i.e., $CJ_L C^{-1} (\varphi^+, \varphi^-) = (-i\varphi^+, i\varphi^-)$,
and the eigenspaces of $J_L$ define a new polarization of $\sp H_L = \sp
H_+ \oplus \sp H_-$, where $\sp H_\pm = \set{\varphi \pm i\dot\varphi =
0}$.  Then the previous
discussion can be summarized saying that the Cayley transform $C$ is an
pseudounitary transformation among the pseudo-Hermitean spaces $(\sp H_L
= \sp L_+\oplus \sp L_-,
\omega_L)$ and $(\sp H_L = \sp H_+ \oplus \sp H_-, \sigma_L )$. 

It is obvious that because $C$ is an isometry between $\omega_L$ and
$\sigma_L$ it will carry $\omega_L$-isotropic subspaces into
$\sigma_L$-isotropic subspaces.   Then the Cayley transform $C$ induces a
homeomorphism between the corresponding spaces of isotropic subspaces
with respect to the natural topology in the set of subspaces of $\sp
H_L$.  We will discuss the topology of such spaces later on.

Now we will continue the proof of Thm. \ref{maximal} by showing that the set $\sp
M$ of all closed maximal $\sigma_L$-isotropic subspaces of $\sp H_L$ can
be identified with the group of unitary operators
$\sp U$ of $L^2(\partial\Omega,E_\pO)$.

\begin{proof}  (Cont.) It is a straightforward consequence of
the diagonal expressing of the pseudo-Hermitean product $\Sigma$. 
Any maximal linear isotropic subspace $N$ is transverse to the
subspaces $\sp H_\pm$.  Thus $N$ defines a
linear operator $U\colon \sp H_+ \to \sp H_-$ and 
\begin{equation}\label{form} N = \set{(\varphi^+, U\varphi^+ )\in \sp H
\mid \varphi^+ \in \sp H_+} = \mbox{graph} U .\end{equation}
It is also an immediate
consequence of the isotropy of $N$ that 
$$\langle \varphi_1^+, \varphi_2^+ \rangle_{\partial \Omega}
= \langle U\varphi_1^+, U\varphi_2^+ \rangle_{\partial \Omega} ,$$
and we conclude that $U$ is an isometry.  The maximality of $N$ implies
that $U$ is an isomorphism.  Then we have that $\sp M \cong \mathcal{U} (\sp H_+,
\sp H_-) \cong \mathcal{U} (L^2 (\pO, E_\pO ))$ and the identification is
continuous because of the continuous embedding of the space of bounded
operators on the space of closed subspaces.
\end{proof}

\begin{remark}  We should warn the reader that in the statement of
the previous results, we have replaced the Sobolev spaces 
$H^{3/2}(\pO, E_\pO)$ and $H^{1/2}(\pO, E_\pO)$ by $L^2(\pO, E_\pO)$.
Even if this is incorrect, there is a natural way of expression the results
in this way by using the appropriate theory of Hilbert scales (see \cite[Ch. 2.2]{Ca09}).
Now it is obvious that with the above characterization the set
of self-adjoint extensions of $\H $ inherits  the group structure of the
group of unitary operators  $U \left(L^2(\Omega,E_\pO)\right)$.
\end{remark}

\medskip

Notice that the domain of the self-adjoint extension of $\H_0$ defined by
$U$ is the linear subspace of sections $\psi$ in $H^2(\Omega,E)$ such that
their boundary values $\varphi$, $\dot\varphi$ satisfy the condition 
\begin{equation}\label{U_con} 
\varphi - i\dot\varphi = U(\varphi + i \varphi) .
\end{equation}
We shall denote such subspace by $\sp D_U = \set{\psi \in \sp D \mid
b(\psi) \in C^{-1}(\mbox{graph} U)}$ or equivalently
\begin{equation}\label{D_U}
\sp D_U = \set{\psi \in W^2(E) \mid
\varphi - i\dot\varphi = U(\varphi + i \varphi)}.
\end{equation}  

Now it is not hard to see that $\sp M_L$ is the
space of self-adjoint subspaces of $\sp H_L$. 
If the self-adjoint subspace $W$ is transverse to the subspace $\sp
L_\mp$ then,  there will exists a self-adjoint operator $A_\pm \colon
\sp L_\pm \to \sp L_\mp$, such $W$ will be the graph of $A_\pm$ and
we will denote it by $W_{A_\pm}$.   Notice that on the other hand, the
self-adjoint subspaces of $\sp H_L$ are precisely the closed maximal
$\omega_L$-isotropic subspaces of $\sp H_L$ with respect to the
skew-pesudo-hermitian structure
$\omega_L$.  \hfill$\Box$

\bigskip

This characterization is very useful because if allows to write the
boundary conditions that characterizes the self-adjoint extension of
$\H$ quite easily.  If $A_+$ is a self-adjoint operator in the boundary,
then the self-adjoint subspace defining the self-adjoint extension of
$\H$ is given by the space of $(\varphi, \dot\varphi)$ such that
\begin{equation}\label{self_bc}  \dot\varphi = A_+ \varphi   ,
\end{equation}
and similarly for $A_-$, the boundary condition will then be
\begin{equation}\label{self_bcm}  A_-\dot\varphi = \varphi  .
\end{equation}
The corresponding unitary operators are obtained by the Cayley
transform on the boundary, i.e., the unitary operator $U_+$ whose
graph is the isotropic subspace image by $C$ of the self-adjoint
subspace defined by $A_+$,
\begin{equation}\label{varphiA} \varphi - i A_+\varphi = U_+ (\varphi + i A_+
\varphi ) ,
\end{equation}
hence,
$$
U_+ = \frac{I - i A_+}{I + i A_+} \,,
$$
and similarly for $A_-$.
$$
U_- = \frac{I + i A_-}{I - i A_-} \, .
$$

Conversely if $U$ is a unitary
operator on the boundary, we can define its Cayley transform as the
self-adjoint subspace of $\sp H_B$ defined by $C^{-1} N_U$, i.e., the
inverse image by $C$ of the isotropic subspace defined by $U$.  It is
clear that not always such subspace will be of the form $W_{A_\pm}$ for
some self-adjoint operator.  If either $A_+$ or $A_-$ exist, it will be
given by the formula
\begin{equation}\label{AU} 
A_\pm = i\frac{I \pm U}{I \mp U} \, .
\end{equation}
It is obvious that $A_\pm$ will
exist if and only if $\pm1$ does not belong to the spectrum of $U$.
Notice that $A_-$ is the inverse Cayley transform (if it exists) of
$U^\dagger$.   

\bigskip

The previous considerations show that there is a distinguished set of
self-adjoint extensions of $\H$ for which the expression of the
boundary conditions defining their domain are not expressible in the
simple form given by Eq. (\ref{self_bc}) or (\ref{self_bcm}).  These
self-adjoint extensions correspond to the case that $\pm 1$ is an
eigenvalue of the corresponding unitary operator or equivalently that
the self-adjoint subspace defined by the unitary operator is not the
graph of a self-adjoint operator on the boundary.   We will call them
the Cayley surfaces of the space of self-adjoint extensions and they
can be described equivalently as follows:
\begin{equation}\label{cayley_sur} \sp C_\pm = \set{U\in U\left( L^2 (\partial
E)\right) \mid \pm 1 \in \sigma (U)} = \set{ W\in \sp M_L \mid \dim
W\cap \sp L_\mp >0} .
\end{equation}

Notice that the unitary operators $\pm I$ are in the Cayley surfaces
$\sp C_\pm$ respectively and because of Eq. (\ref{varphiA}) $I$ defines
Dirichlet boundary conditions:
\begin{equation}\label{dirich}  \varphi = 0 .
\end{equation}
Moreover the unitary operator $-I$ is not in the Cayley surface $\sp
C_+$ and corresponds to the self-adjoint operator $A = 0$ because of eq.
(\ref{AU}), thus using (\ref{self_bc}) it defines Neumann boundary
conditions
\begin{equation}\label{neuma}  \dot\varphi = 0 .
\end{equation}

These observations provide us a way to classify all self-adjoint
extensions of $\H$ in those which are not in the Cayley surface and
those which are inside.  We will proceed later on to perform this
analysis.  

\subsection{Squaring the space of boundary conditions of Dirac operators
and boundary conditions for Laplace operators}\label{squaring}

In the previous sections we have repeated step by step the theory of
self-adjoint extensions of the Hodge Laplacian for Dirac operators.  The results are
similar but subtly different due to the fundamental different
nature of both families of operators.   

In both theories a crucial role is played by a skew-pseudo-Hermitean structure on the Hilbert space of
boundary data.   In fact, the space of self-adjoint extensions is
identified with a Lagrangian submanifold of the infinite dimensional
Grassmannian naturally defined by such structure.   On the other hand,
we have already noticed that self-adjoint extensions arise in a twisted
way in both cases, i.e., the Cayley transform exchanges unitary and
self-adjoint subspaces in reverse order.   

On the other hand, Dirac operators $\Dsl$ lead naturally to elliptic second
order differential operators, the Dirac Laplacian $\Dsl^2$, which are closely
related to Hodge Laplacians by Bochner's identities that play a
fundamental role in understanding topological properties of Riemannian
manifolds.  Such process of squaring a Dirac operator, should lead to a
precise link between the theory of self-adjoint extensions of Laplace
operators and Dirac operators.  We will discuss this link in this
section finding out that squaring of Dirac operators has a natural
abstract counterpart in the boundary Hilbert spaces of the operators. 
This construction will be called squaring a Hilbert space with a
skew-pseudo-Hermitean structure.

\medskip

Let $\Dsl$ denote a Dirac operator on the Dirac bundle $S\to \Omega$ and
$\sp H_D$ the Hilbert space of boundary data for $\Dsl$ equipped with the
canonical compatible complex structure $J_D$.   We shall consider now
the Dirac Laplacian $\Dsl^2\colon \Gamma_0 (S) \to \Gamma_0 (S)$.  The
Dirac Laplacian is obviously a symmetric operator in the domain $H_0^2
(S)$, which is the closure of $\Gamma_0 (S)$ with respect to the Sobolev
norm $||.||_{2,2}$.   Notice that the operator $\Dsl^2$ with domain
$H_0^2(S)$ is the closure of $\Dsl^2$ defined on $\Gamma_0 (S)$.  We shall
study now the self-adjoint extensions of $\Dsl^2$.  Then we integrate by
parts the $L^2$-product $\langle \Dsl^2 \xi, \zeta \rangle$,
\begin{eqnarray}
\langle \Dsl^2 \xi, \zeta \rangle &=& \langle \Dsl \xi, \Dsl \zeta \rangle
 + \Sigma (\Dsl\xi, \zeta ) \nonumber  \\ 
 \label{int_part_Dsquare} &=& \langle  \xi, \Dsl^2 \zeta \rangle + \Sigma
(\xi, \Dsl\zeta )
 + \Sigma (\Dsl\xi, \zeta ) \, ,
 \end{eqnarray}
where we have used the computation in eqs. (\ref{int_part_D2}).  
But because of Eq. (\ref{bound_D}) we have that,
\begin{equation}\label{split_D} \Sigma (\Dsl\xi, \zeta) = \langle J b(\Dsl\xi), b(\zeta)
\rangle .
\end{equation}
The effect of the boundary map $b$ on $\Dsl\xi$ will be computed
as follows. We choose a collar neighborhood $U_\epsilon$ of the boundary $\partial \Omega$.  Notice that because $\partial \Omega$ is compact, 
there exists $\epsilon > 0$ such that the map $x \mapsto (s,y)$, with $y$ the point in $\partial \Omega$ closest to $x$
and $s = \mathrm{dist}(x,y)]$ defines a diffeomorphism between  $U_\epsilon$ and $(-\epsilon,0]\times \partial
\Omega$.    Using the decomposition of the tangent space defined by the previous diffeomorphism 
we can write $\Dsl$ on $U_\epsilon$ as 
\begin{equation}\label{DpartialOmega}
\Dsl = J (\partial_\nu + \Dsl_{\partial \Omega}) 
\end{equation}
where the vector field $\nu$ on $U_\epsilon$ is defined by $\partial / \partial s$.  Notice that when restricted to $\partial \Omega$, $\nu$ is just the normal vector to the boundary $\partial \Omega$.  Thus, the covariant derivative in the normal
direction is given by partial derivative with respect to this coordinate
and is represented by $\partial_\nu$.   Then, choosing a orthonormal
frame in $U$ with the first vector the extended vector field
$\tilde{\nu}$, the Dirac operator $\Dsl$ splits as in Eq. (\ref{split_D}). 
The symbol $J$ corresponds to Clifford multiplication by $\tilde{\nu}$
and restricted to $\partial \Omega$ gives the complex structure $J_D$. The
operator $\Dsl_{\partial \Omega}$ is the transversal Dirac operator that
restricted to $\partial\Omega$ is the Dirac operator of the Dirac bundle
$S_\pO $ discussed in \S \ref{section:Dirac}.  Thus, 
\be
b(\Dsl\xi) = b((J\partial_\nu +
\Dsl_{\partial\Omega})\xi) = Jb(\partial_\nu \xi) +
\Dsl_{\partial\Omega}b(\xi). \label{split}
\ee
Then,
\begin{equation}\label{first_int_D} \Sigma (\Dsl\xi, \zeta) = \langle J(Jb(\partial_\nu
\xi) + \Dsl_{\partial\Omega}b(\xi)), b(\zeta)\rangle = - \langle
b(\partial_\nu
\xi), b(\zeta) \rangle + \langle J \Dsl_{\partial\Omega} b(\xi), b(\zeta)
\rangle .
\end{equation}
On the other hand, we can compute similarly the last term in the r.h.s.
of Eq. (\ref{int_part_Dsquare}), and we obtain,
\begin{equation}\label{second_int_D} \Sigma (\xi, \Dsl\zeta) = \langle Jb(\xi), b(\Dsl\zeta)
\rangle = \langle Jb(\xi), Jb(\partial_\nu\xi) +
\Dsl_{\partial\Omega}b(\zeta)
\rangle = 
\end{equation}
$$ = \langle b(\xi), b(\partial_\nu\xi) \rangle + 
\langle Jb(\xi),  \Dsl_{\partial\Omega}b(\zeta)
\rangle .$$ 
Collecting the last terms in eqs. (\ref{first_int_D}),
(\ref{second_int_D}), we obtain,
$$ \langle J \Dsl_{\partial\Omega} b(\xi), b(\zeta)
\rangle + \langle Jb(\xi),  \Dsl_{\partial\Omega}b(\zeta)
\rangle = $$
$$\langle J\Dsl_{\partial\Omega} b(\xi), b(\zeta)
\rangle + \langle \Dsl_{\partial\Omega} Jb(\xi),  b(\zeta)
\rangle$$ 
$$ =  \langle J\Dsl_{\partial\Omega} b(\xi), b(\zeta)
\rangle - \langle J\Dsl_{\partial\Omega} b(\xi),  b(\zeta)
\rangle = 0 ,$$
where we have used that $\Dsl_{\partial\Omega}$ is essentially self-adjoint
on the boundaryless manifold $\partial\Omega$ and the $J$ anticommutes
with $\Dsl_{\partial\Omega}$.  Denoting $b(\xi) = \phi$, $b(\partial_\nu
\xi) =\dot{\phi}$, $b(\zeta) = \psi$, $b(\partial_\nu \zeta) =
\dot{\psi}$, we finally obtain,
$$\Sigma (\Dsl\xi,\zeta) + \Sigma (\xi, \Dsl\zeta) = \langle \phi, \dot{\psi}
\rangle - \langle \dot{\phi}, \psi \rangle ,$$
which is exactly the same boundary term we obtained in our analysis of
self-adjoint extensions of the Laplace operator.  We can abstractly
describe the analysis of self-adjoint extensions of the Dirac Laplacian
by considering the Hilbert space $\sp H_B = \sp H_D\oplus \sp H_D$ with
the direct sum inner product $\langle .,. \rangle_B = \langle .,.
\rangle \oplus \langle .,. \rangle $ and the compatible complex
structure $J_B$ on $\sp H_B$ defined by as in Eq. (\ref{J_B}).  Notice
that $J_B$ is not the direct sum of $J_D \oplus J_D$ which becomes
another compatible complex structure on $\sp H_B$.  This process is what
will be called squaring the boundary data and is summarized in the
following table.

\bigskip

\begin{center}
\begin{tabular}{||c|c||} \hline\hline  $\Dsl$ & $\Dsl^2$ \\ \hline
$\sp H_D = \sp H_+ \oplus \sp H_-$ &  $\sp H_B = \sp H_D \oplus \sp H_D$
\\ $\langle .,. \rangle$ & $\langle .,. \rangle_B$ \\ $J_D$ & $J_B$ \\
$\omega_D$ & $B$ \\ \hline $U(\sp H_+, \sp H_-)$ & $\sp M_B$ \\ $\sp
M_D$ & $U(\sp H_D, \sp H_D)$ \\ \hline\hline 
\end{tabular}
\end{center}

The previous discussion implies immediately that the self-adjoint
extensions of the Dirac Laplacian $\Dsl^2$ are described in exactly the
same terms as the self-adjoint extensions of the Hodge Laplacian
described earlier.   We will denote as before by $\sp M_B$ the
self-adjoint Grassmannian of $\Dsl^2$ which is diffeomorphic to the
self-adjoint Grassmannian of $\nabla$.  Notice that because of Bochner's
identity, this was expected in advance. In fact, the Dirac Laplacian
$\Dsl^2$ and the Hodge Laplacian $\nabla$ differ only on a curvature
term that, because is a zero order operator, is self-adjoint, then
both theories should agree.  We shall discuss other aspects concerned
with these facts later on. 

\medskip

It is only natural what is the relation between self-adjoint extensions
of $\Dsl^2$ and self-adjoint extensions of $D$.   As we know from the
previous discussions such extensions are always self-adjoint spaces of
the boundary Hilbert space with respect to a skew-pseudo-hermitian
structure induced by a compatible complex structure.  As we have shown
in this section these structures are related by the ``square
construction'', thereby there must be a definite relation between them.  

First we shall study the self-adjoint extensions of $\Dsl^2$ induced from
those of $\Dsl$.  Let $W$ be a self-adjoint subspace of $\sp H_D = \sp H_+
\oplus \sp H_-$.   If $\Dsl_W$ is self-adjoint, it will be $\Dsl_W^2$.  Notice
that the boundary data for $\Dsl^2$ are decomposed as $\phi = \phi_+ +
\phi_-$, and $\dot\phi = \dot\phi_+ + \dot\phi_-$ where all the factors
are mutually orthogonal because of the decomposition
\begin{equation}\label{decom_LD} \sp H_L = \sp H_D \oplus \sp H_D = (\sp H_+ \oplus
\sp H_-) \oplus  (\sp H_+ \oplus \sp H_-) .
\end{equation}

Because $\Dsl_W$ is self-adjoint we have
$$ \langle \Dsl_W^2 \xi , \zeta \rangle = \langle \Dsl_W\xi, \Dsl_W\zeta \rangle
+ \omega_D (\dot\phi, \psi ) ,$$
and the last term in the r.h.s. must vanish for all $(\dot\phi, \psi)
\in W_{2,3}$.  Notice that we are identifying $W\subset \sp H_D$ with the
diagonal subspace $W_{2,3}$ obtained taking the second and third terms in
the decomposition of $\sp H_L$ given in Eq. (\ref{decom_LD}).  In the same
way, repeating the integration by parts, we will obtain that the term
$\omega_D (\phi, \dot\psi ) = 0$ for all $(\phi, \dot\psi )\in W_{1,4}$,
where $W_{1,4} = \set{ (\phi,0,0,\dot{\psi}) \in \sp H_L \mid
(\phi,\dot\psi)\in W }$.  Thus, to the self-adjoint extension $\Dsl_W$ we
associate the self-adjoint extension, still denoted by $\Dsl_W^2$, with
self-adjoint subspace $\tilde{W} = W_{2,3} \oplus W_{1,4}$.  

Conversely, it is clear that $\Dsl^2$ has far more self-adjoint extensions
than those defined by self-adjoint subspaces of $\sp H_L$ of the form
described above $\tilde{W}$.   However, it is easy to show that if
$\tilde{W}$ defines a self-adjoint extension of $\Dsl^2$, i.e., $\tilde{W}$
contains the boundary conditions $(\phi, \dot\phi)$ that makes $\Dsl^2$
self-adjoint, then, the subspace of $\sp H_D$ defined by the first
component, is going to define an extension of $\Dsl$.  We shall denote by
$\pi_1\colon \sp H_L \to \sp H_D$ the projection on the first factor. 
Then, if $W$ denotes the projected subspace $\pi_1 (\tilde{W})$, then,
we can ask when $W$ will be self-adjoint. 

%%%%% 

%FINISH THE PROOF OF THIS THEOREM:  The diagonal embedding of
%self-adjoint extensions of $D$ on $D^2$.
The characterization of  selfadjoint extensions of $\Dsl^2$ which are induced
by those of $\Dsl$ can be achieved in simple terms.
Let us consider the Clifford algebra element $e_{n+1}=e_1\cdot e_2\cdots\cdot e_n$ which in
even dimensional manifolds is always non-trivial. In  odd dimensional manifolds
one can consider the non-trivial representation of the  Clifford algebra induced from the
$n+1$-dimensional Clifford algebra and then $e_{n+1}$ becomes also non-trivial. 
In both cases due to the special properties of 
Clifford algebra we have that $e_{n+1}\cdot e_i+e_i\cdot e_{n+1}=0$ for any $i=1,2,\cdots, n$
and $e_{n+1}\cdot {\nu}+{\nu}\cdot e_{n+1}=0$.
The selfadjoint extensions of $\Dsl$ can then be characterized by  unitary operators $U$ of
$H^1(\Omega,S_\pO)$ which commute with $\nu$. The corresponding domains are given by
\be
{\cal{D}}_{\dsl_U}=\{ \psi\in L^2(\Omega,S);  P_-\psi|_{_{S_\pO }}= U e_{n+1}P_+\psi|_{_{S_\pO }}\},
\ee
where $P_\pm$ are the projectors
$P_\pm=\frac12(\I\pm \nu)$. The domain of the corresponding Dirac Laplacian $\Dsl_U^2$
is given by the subdomain of ${\cal{D}}_{\dsl_U}$ of spinors $\psi\in {\cal{D}}_{\dsl_U}$ such that $ \Dsl\psi\in {\cal{D}}_{\dsl_U}$.
This requirement imposes further constraints on the normal derivative,
\beq
{\cal{D}}_{\dsl^2_U}&=&\{ \psi\in H^2(S);  P_-\psi|_{_{S_\pO }}= U e_{n+1}P_+\psi|_{_{S_\pO }},\nonumber\\
&& P_-(\I+U e_{n+1})\partial_\nu \psi|_{_{S_\pO }}
= P_-(\I-U e_{n+1}) \Dsl_{_{S_\pO }} \psi|_{_{S_\pO }} \}
\eeq
as required by the boundary conditions of the second order differential operator $\Dsl^2$ 
which give rise to  selfadjoint extensions. 
Notice, that these boundary conditions are not the most general ones that make $\Dsl^2$ selfadjoint, but are
the only ones which guarantee that $\Dsl^2$ is the square of the selfadjoint Dirac operator $\Dsl$.
To illustrate this let us consider a couple of examples. Let  $U$  be of the form
\begin{equation}\label{ubc}
	U=e^{2 i\arctan e^{\alpha} \nu},
\end{equation}
which because of the identities
\begin{eqnarray*}\label{ibc}
	\!\!\!\!\frac{I-U e_{n+1}}{I+U e_{n+1}}\!\!&=&\!\!\frac{I+U }{I-U}(1-e_{n+1})+\frac{I-U }{I+U}(1+e_{n+1})\\
	\!\!&=&\!\!i\cot(\arctan e^{\alpha}) (1-e_{n+1})-i \tan(\arctan e^{\alpha}) (1+e_{n+1})\\
	\!\! &=&\!\!i e^{-\alpha} (1-e_{n+1})-i e^{\alpha} (1+e_{n+1})= -i e^{\alpha e_{n+1}}e_{n+1},
\end{eqnarray*}
corresponds  to the chiral bag boundary conditions:
\begin{equation}\label{cbc}
	\frac12\left(1-i e_{n+1}e^{-\alpha e_{n+1}}\nu\right)  \psi=0.
\end{equation}
In the massive case these extensions can give rise to the existence of edge states  \cite{As13,As15}.
Another example corresponds to the Atiyah-Patodi-Singer boundary conditions  \cite{At75} which are given by
$P_+ \psi|_{S_\pO }=0,$
where $P_+$ is the orthogonal projector on the subspace of $H^1(\Omega,S_\pO)$ corresponding to the
positive spectrum of the selfadjoint operator $\Dsl_{S_\pO }$ (see later on Sect.  \ref{section:elliptic}). In that case the spinors of the domain of 
$\Dsl^2$ must satisfy the extra requirement that also involves the normal derivatiative
\beq
{\cal{D}}_{\dsl^2}&=&\{ \psi\in H^2(S);  P_+ \psi|_{S_\pO }=0,  P_+ \partial_\nu\psi=-P_+ \Dsl_{S_\pO }\}.
\eeq

%%%%%%%%%%%%%%%%%%%%%%%%%%%
%%%%%%%%%%%%%%%%%%%%%%%%%%%

\section{Von Neumann's theorem and boundary conditions revisited}\label{section:Neumann}

%%%%%%%%%%%%%%%%%%%%%%%%%%%

\subsection{Von Neumann's theorem vs. unitary operators at the boundary}

A theory of self-adjoint extensions of symmetric differential operators based on the
geometrical structures induced by them in the corresponding spaces of
boundary data  has been sketched along the previous sections.   
However a general solution to this problem was set up
by von Neumann \cite{Ne29} in the abstract realm of symmetric operators
with dense domains in Hilbert spaces.    

We will show the exact nature
of the link between both approaches, the one discussed in this work based on geometrical boundary
data and von Neumann's based on global information on the bulk.

We have already mentioned the fact that the theory of extensions for the Dirac Laplacian and the covariant Laplacian are the same because Bochner's identity (\cite{La89}, Thm. II.8.2) implies that the difference between both operators is a zeroth order operator.  Hence we will keep the discussion
in this section to the covariant Laplacian without loss of generality. 

We have characterized self-adjoint extensions of the covariant Laplacian
as unitary operators $\sp U\left(\sp H_+, \sp H_-\right)$ whereas the
standard characterization of self--adjoint extensions provided by von
Neumann's theorem \cite{Ne29} is by means of unitary operators $K\colon \sp
N_i \to \sp N_{-i}$ between the  deficiency subspaces $\sp N_{\pm i} =
\{\psi\in L^2(\Omega) \mid \H_0^\dagger \psi= \pm i\psi\}$.  As we have
stressed before, the advantage of the former characterization 
is that it is directly related to conditions that the wave functions must
satisfy on the boundary
$\partial\Omega$.

Before we shall discuss the exact relation between von Neumann's theory
and Thm. \ref{maximal} we should mention that there has been several
refinements of von Neumann's theory extending it to more general
situations.   For instance, if ${}^*$ denotes a conjugation on a Hilbert space, i.e., an
antilinear involution such that $\langle \psi_1^*, \psi_2^* \rangle =
\langle \psi_2, \psi_1\rangle$, a linear operator $A$ with dense domain
is said to be $*$--symmetric if
${}^* A^* \subset A^\dagger$.  If ${}^* A^* = A^\dagger$, then $A$ is
called $*$--self-adjoint.  If $A$ is $*$--real, i.e., ${}^* A^* =
A$, then $A$ $*$--self-adjoint implies that $A$ is self--adjoint in the standard sense.  Then it was proved in
\cite{Ga62} that any $*$--symmetric operator with dense domain has a
$*$--self-adjoint extension.  

There is also a generalization of von Neumann's theory of
extensions of symmetric operators with dense domain to formally normal
operators with non-dense domains \cite{Co73}.  Both generalizations can be
discussed from the viewpoint of the geometry
of boundary conditions.  We will not insist on this and we
will restrict for clarity to the simpler case
of self--adjoint extensions of symmetric operators with dense domains.

We denote as usual by $- \Delta_A^\dagger$ the adjoint, with domain
$\sp D_\mathrm{max} = H^2( E)$, of the Laplacian $- \Delta_A$ with domain
$\sp D_\mathrm{min} = H^2_0 (E)$.  We will also use the notation $\mathcal{D}_0$ for the domain $\mathcal{D}_\mathrm{min}$.  Given any $\lambda\in \C$, $\mathrm{Im}\lambda > 0$,  we define the deficiency
spaces $\sp N_{\lambda}$, $\sp N_{\bar\lambda}$,
by,
\begin{eqnarray}\label{deficiency} \sp N_{\lambda}  &=& \ran
( -\Delta_A + \lambda I)^\perp = \ker (-\Delta_A^\dagger + \bar\lambda I), \\
\sp N_{\bar\lambda}  &=& \ran ( -\Delta_A + \bar\lambda I)^\perp = \ker
(-\Delta_A^\dagger + \lambda I) \, ,
\end{eqnarray}
that are closed spaces of the Hilbert space $L^2(\Omega,E)$.

It is then true that for any nonreal $\lambda$,
\begin{equation}\label{DDNN}
\sp D_\mathrm{max} = \sp D_\mathrm{min} + \sp N_\lambda + \sp N_{\bar\lambda} \, ,
\end{equation}
and the sum is direct as vector spaces. 
Von Neumann's theorem states that:

\begin{theorem}\label{vonNeuman}\cite{Ne29} There exists a one-to-one correspondence between
self-adjoint extensions of $-\Delta_A$ and unitary operators $K$ from
$\sp N_\lambda$ to $\sp N_{\bar\lambda}$, for any $\lambda\in \mathbb{C}$, $\mathrm{Im}\lambda > 0$.
\end{theorem}

The domain of the self--adjoint extension corresponding to the operator
$K$ is $\sp D_0 + \ran (I+K)$ and is defined for a function of the form
$\Psi = \Psi_0 + (I+K) \xi_+$,
$\Psi_0 \in \sp D_0$, $\xi_+\in N_+$, by
$$ \Delta_A^K \psi = \Delta_A \Psi_0 + \bar{\lambda} \xi_+ +  \lambda K\xi_+
\, .$$ 
Notice that the theorem implies that all deficiency spaces $\mathcal{N}_\lambda$ with $\mathrm{Im}\lambda > 0$
are isomorphic.

Different presentations of this theorem, and of Eq. \eqref{DDNN}, can be found for instance in
\cite{Du63}, \cite{Ak63}, \cite{Na68} \cite{Yo65}, \cite{Re75} and \cite{We80} and, as it was already expressed in
the introduction, there exists an abundant literature on the subject.

Once that Eq. \eqref{DDNN} is established, the main idea of the proof is to show that there is a one-to-one
correspondence between extensions of the symmetric operator $-\Delta_A$ and
extensions of its Cayley transform $U_\Delta \colon \ran (-\Delta_A + \bar\lambda I)
\to \ran (-\Delta_A + \lambda I)$ defined by
$$ 
U = \frac{-\Delta_A + \lambda I}{-\Delta_A + \bar\lambda I} \, .
$$ 

To compare with our previous results it will be convenient to 
describe von Neumann extension theorem in the
setting of skew-Hermitean spaces.

We define the total Hilbert deficiency space $\sp N =Ê\sp N_\lambda
\oplus \sp N_{\bar\lambda}$.   Similarly to the results obtained in Sections \ref{section:maximal}, \ref{section:Cayley},
unitary operators from $K\colon \sp N_\lambda \to \sp N_{\bar\lambda}$ are in
one--to--one correspondence with maximal isotropic subspaces of $\sp
N$ with respect to the natural pseudo-Hermitean structure
$\omega_{\sp N}$ defined on $\sp N$ by:
\begin{equation}
\sigma_{\sp N} ((\Psi_1^+, \Psi_1^-),(\Psi_2^+,\Psi_2^-)) = \langle \Psi_1^+
,\Psi_2^+ \rangle - \langle \Psi_1^-,
\Psi_2^- \rangle \, ,  
\end{equation}
for all $\Psi_\alpha^+ \in \mathcal{N}_\lambda, \Psi_\alpha^-
\in \mathcal{N}_{\bar\lambda}$, $\alpha = 1,2$.
Now we can try to identify the deficiency space in the bulk
$\sp H_{VN}$ with the boundary space $\sp H_L$.

\medskip 

The boundary map $b$ restricts to $\sp N \subset H^2 (E)$, and
moreover $b$ restricts to the closed subspaces $\sp N_\lambda$, $\sp
N_{\bar\lambda}$.  We compose $b$ with the Cayley transform on the
boundary $C$ to obtain a continuous linear map $j\colon \sp N
\to \sp H_L$ defined as follows.  Let $j_\pm (\Psi^\pm) = \varphi \pm i
\dot{\varphi}$, where $(\varphi, \dot\varphi) = b (\Psi )$, where $j_\pm$ denote the
restriction of $j$ to $\sp N_\pm$ respectively. Then, $j =
j_+ \oplus j_-$.   We will denote $j_\pm (\Psi^\pm )$ as usual by $\phi^\pm$.   Then,

\begin{equation}\label{j}  
j(\Psi^+, \Psi^-) = ( \phi^+, \phi^- ) .
\end{equation}
The following Lemma will
show that $j$ preserves the skew-Hermitean structures.

\begin{lemma}  With the above notation the map $j\colon \sp N
\to \sp H_L$ verifies
$$ \sigma_{\sp N} ((\Psi_1^+, \Psi_1^-),(\Psi_2^+,\Psi_2^-)) = \sigma_L
((\varphi_1^+, \varphi_1^-),(\phi_2^+,\phi_2^-)) .$$ 
\end{lemma}

{\it Proof:}  We consider $\lambda = i$, the proof for general
$\lambda$ proceeds in the same way. 
We consider first $\Psi_1^+, \Psi_2^+ \in \sp N_i$, then
$-\Delta_A^\dagger \Psi_\alpha^+ = i\Psi_\alpha^+$, $\alpha =1,2$.

Then it is clear that,
\begin{eqnarray*}
0 &=& \langle \Psi_1^+, (-\Delta_A^\dagger - i)\Psi_2^+ \rangle =
\langle \Psi_1^+, -\Delta_A \Psi_2^+ \rangle -i \langle \Psi_1^+,
\Psi_2^+ \rangle \\
&=& \langle -\Delta_A\Psi_1^+, \Psi_2^+ \rangle -i \Sigma_B
(b(\Psi_1^+),b(\Psi_2^+)) - i \langle \Psi_1^+, \Psi_2^+ \rangle \\
&=& \langle (-\Delta_A -i)\Psi_1^+, \Psi_2^+ \rangle - 2i \langle
\Psi_1^+, \Psi_2^+ \rangle - i\Sigma_B
(\phi_1^+, \phi_2^+) \\  &=& - 2i \langle
\Psi_1^+, \Psi_2^+ \rangle - i\Sigma_B (\phi_1^+,\phi_2^+).
\end{eqnarray*}
Hence,
$$
\sigma_{\sp N} ((\Psi_1^+,0),(\Psi_2^+,0)) = \langle \Psi_1^+, \Psi_2^+ \rangle
= -\frac{1}{2} \sigma_L (\phi_1^+,\phi_2^+) = -\frac{1}{2}\sigma (\phi_1^+,0; \phi_2^+,0) \, .
$$
Similarly, it is shown, 
$$
\sigma_{\sp N} (0,\Psi_1^-;0,\Psi_2^-) = \sigma
(0,\phi_1-;0,\phi_2^-)  \, ,
$$
that together with the ortogonality of $\sp N_i$ and $\sp N_{-i}$ with respect to
$\sigma_L$ proves the result. \hfill $\Box$

\bigskip

To show that $j$ is onto we will
need the following facts from the existence and uniqueness of solutions of the following Dirichlet's
problem.

\begin{proposition}\label{existence} For every $\varphi \in \Gamma^\infty(\partial
E)$, and for every $\lambda\in \mathbb{C}$ there is a unique solution to the
equations,
\begin{equation}\label{delta_i}  - \Delta_A \Psi + \bar\lambda \Psi = 0 \, , 
\qquad - \Delta_A \Psi + \lambda \Psi = 0 \, ,
\end{equation}  
with boundary condition
$$
\Psi \mid_{\partial \Omega} = \varphi \, .
$$
\end{proposition}

\begin{proof}   We prove first uniqueness.  If there were two
solutions $\Psi_1$, $\psi_2$, then because the operator $-\Delta_A + \bar\lambda$ is
elliptic, then because $\varphi$ is smooth, by elliptic regularity they will be both smooth.   Then, $\Psi =
\Psi_1 - \Psi_2$ also satisfies Eq. (\ref{delta_i}) with the boundary
condition $\Psi \mid_{\partial \Omega} = 0$,  which is impossible by the
uniqueness of the solution of the Dirichlet's problem.   Moreover we can argue as follows. If we look
for solutions $\Psi$ of the equation (\ref{delta_i}) such that
$\Psi \mid_{\partial \Omega} =$ constant, then, we can remove the boundary
identifying all their points and looking for the solutions of eq.
(\ref{delta_i}) on the closed manifold $\Omega^\prime$ obtained in this
way.  But now, $-\Delta_A$ is essentially self-adjoint on
$\Gamma (E^\prime)$ where $E^\prime$ is the fibre bundle obtained from $E$
identifying all the fibres over $\partial \Omega$\footnote{Notice that
the compactness of $\Omega$ is crucial in this statement.}, and then it has
not imaginary eigenvalues.   

\medskip

Let us now prove the existence of solutions.  Let $\widetilde{\Psi}$ be any
section in $\Gamma^\infty (E )$ such that $\widetilde{\Psi}
\mid_{\partial \Omega} = \varphi$.  Then, there exists a unique section
$\zeta \in \Gamma (\Omega )$ such that
$$ - \Delta_A \zeta + \bar\lambda \zeta = \Delta_A \tilde{\Psi} - \lambda
\tilde{\Psi} ,$$  
with Dirichlet boundary conditions, $\zeta\mid_{\partial\Omega} = 0$ which is a consequence of the
solution of the Dirichlet boundary value problem for elliptic
operators.  Then, the section $\Psi = \zeta + \widetilde{\Psi}$ verifies Eq.
(\ref{delta_i}) and the boundary condition $\Psi \mid_{\partial
\Omega} = \varphi$. 
\end{proof}

Notice that because of Lions' Theorem \cite{Li72} the previous theorem can be refined assuming that $\varphi \in H^{3/2}(\partial\Omega)$.  In that case the solution $\Psi$ will lie in $H^2(\Omega)$.

An alternative argument to the previous proof will be to consider the operator $T = (- \Delta_A  + \bar\lambda )(- \Delta_A + \lambda) = \Delta^2 - 2\mathrm{Re}\lambda \Delta + | \lambda |^2$ which is a positive 4th order elliptic differential operator and solve the equation $T \Psi = 0$ with boundary conditions $\Psi\mid_{\partial \Omega} = \varphi$, and $\left([- \Delta_A + \lambda]\Psi\right)\mid_{\partial \Omega} = 0$ which are elliptic boundary conditions.  Hence there is a unique solution to the system $\Psi$ which obviously solves the system Eq. \eqref{delta_i}.

\begin{theorem}\label{ident} The total deficiency space on the bulk $\sp N$ with its natural skew-Hermitean structure $\sigma_{\sp N}$ is
isometrically isomorphic to the boundary data space $\sp H_L$ with its
natural skew-Hermitean structure $\sigma_L$ as skew-Hermitean
spaces.
\end{theorem}

\begin{proof}  We will have to show that the map $j$ is
onto. We can solve the boundary problems
\begin{eqnarray}   -\Delta_A \Psi^+ + \bar\lambda \Psi^+ &=& 0 \, , \qquad
\Psi^+\mid_{\partial \Omega} = \phi^+ \\ -\Delta_A \Psi^- + \bar\lambda \Psi^- &=& 0 \, , \qquad 
\Psi^-\mid_{\partial \Omega} = \phi^- , 
\end{eqnarray} 
for given $(\phi^+, \phi^-) \in \Gamma^\infty (\partial \Omega\times \partial\Omega)$.
notice that such solutions will lie necessarily on $\sp N_\lambda$, $\sp N_{\bar\lambda}$. 
Proposition \ref{existence} shows that such solutions $\Psi^+$, $\Psi^-$ exist and
they are unique.  They define the inverse of the map $j$ on the dense
subspace $\Gamma^\infty (\partial \Omega\times \partial\Omega)$, and it is continuous as the resolvents of the operators $-\Delta_A  + \bar\lambda$, $-\Delta_A + \lambda$ are compact.  Thus $j$ is an isometry onto. 
\end{proof} 

The following result follows immediately from the previous discussion.

\begin{corollary}\label{unit_vn}
There is a one-to-one correspondence $K \mapsto U$ between unitary operators $K$ from
$\sp N_\lambda$ to $\sp N_{\bar\lambda}$ and unitary operators $U$ at the
boundary given as:
$$
\mathrm{graph}(U) = b (\mathrm{graph}(K)) \, , 
$$
or, in other words, 
\begin{equation}\label{KU}
U = b \circ K \circ j \, .
\end{equation}
\end{corollary}

\begin{remark}  A few remarks concerning the previous results are in order.
\begin{enumerate}

\item Notice first that the previous theorem can also be seen as offering an
alternative proof of von Neumann's theorem.  

\item The previous statement could also be expressed in terms of the
``direct'' boundary data $\varphi$ and $\dot\varphi$.  In fact they will
be obtained from the expressions:
$$
\varphi = \frac{1}{2} (\phi^+ + \phi^-); ~~~~~  \dot\varphi =
-\frac{i}{2} (\phi^+ - \phi^-) \, .
$$ 

\item The previous correspondence between unitary operators can be extended to any operator
$L \colon \mathcal{N}_+ \to \mathcal{N}_-$ not necessarily unitary by means of formula  Eq. \eqref{KU},
that is given the (non-unitary) operator $L$ we define the (non-unitary) operator from $\mathcal{H}_+$ to $\mathcal{H}_-$:
$$
A = b\circ L \circ j \, .
$$
This fact will be of consequence later on in Sect. \ref{section:dissipative} when discussing dissipative quantum systems and non-self-adjoint boundary conditions.

\item Notice that Thm. \ref{ident} was stated in a way that the specific form of the
operator $\H$ was not appearing explicitly so that it suggests
the form that they would have in general.  For instance, we can use the results obtained so far to prove the
analogue of Thm. \ref{unit_vn} for the Dirac operator.
\end{enumerate}

\end{remark}

%%%%%%%%%%%%%%%%%%%%%%%%
%%%%%%%%%%%%%%%%%%%%%%%%

\subsection{Examples and applications: Boundary conditions for the Laplace operator in one-dimension}\label{Laplace1D}

%%%%%%%%%%%%%%%%%%%%%%%%

The relation with the classical boundary conditions analyzed in the
previous section becomes also  clear in the light of   this geometric
approach. The quantum extension which corresponds to the quantization
of the classical boundary condition $S^\rho_\alpha$ is precisely the one
associated to the  unitary operator given by 
\begin{eqnarray}\label{clas} && U\varphi(x)=
\varphi(\rho^{-1}(x))e^{i\alpha(x)} \\ &&
\dot\varphi(x)=-\dot\varphi(\rho^{-1}(x)) e^{i\alpha(x)} .
\end{eqnarray}
We remark, however, there are many other quantum extensions  given by
operators which are not of the form (\ref{clas}). Therefore, not all
constrained quantum systems correspond to the quantization  of a
constrained classical system.

To illustrate the utility of the above geometric approach we   
consider some simple applications to Sturm-Liouville problems. In such
a case the configuration space is constrained to an interval
$\Omega=[0,1]$ of real numbers.  The metric $g$ is the  standard 
Euclidean metric of $\R$ and the symmetric operator is  the
Sturm-Liouville second order differential operator 
$$
\H = -\frac{1}{2}\Delta = -\frac{1}{2} \frac{d}{dx^2} \, ,
$$  
defined on $C^\infty_0([0,1])$.  The boundary set is in this case discrete:
$\partial\Omega=\{0,1\}$, and
$L^2(\partial\Omega)=\C^2$. Therefore the different self-adjoint
extensions are parametrized by a $2\times 2$  unitary matrix.
\begin{equation}\label{mat} U=\left( \begin{array}{cc} u_{11}&u_{12}\\
u_{21}&u_{22}\end{array}\right) .
\end{equation}
The domain of the associated extension is given by the functions of
$ H^2([0,1])$ whose boundary values satisfy the following equations
(notice that we are in a 1-dimensional manifold, and because of Sobolev
inequalities, the functions in $H^2([0,1])$ are $C^1$ and their derivatives absolutely continuous):
\begin{equation}\label{sturm}
\left(\begin{array}{c} \varphi(0) + i\dot\varphi(0) \\
\varphi(1)+i\dot\varphi(1)\end{array}\right) = \left(\begin{array}{cc}
u_{11}&u_{12}\\ u_{21}&u_{22}\end{array}\right)
\left(\begin{array}{c} \varphi(0)-i\dot\varphi(0) \\
\varphi(1)-i\dot\varphi(1)\end{array}\right) ,
\end{equation}
where $\dot\varphi(0)=\varphi'(0)$ and $\dot\varphi(1)=-\varphi'(1)$.

Some specially interesting examples correspond to the case when the
matrix $U$ is diagonal or antidiagonal. In the first case we have
\begin{equation}\label{diagona} U= \left(\begin{array}{cc}
e^{i\epsilon} & 0\\ 0&e^{-i\gamma}
\end{array}\right) ,
\end{equation}
which corresponds to the boundary conditions
\begin{eqnarray}\label{diric} 
-\sin \frac{\epsilon}{2}\varphi(0) + \cos\frac{\epsilon}{2}\dot\varphi(0) &=& 0
\\  -\sin \frac{\gamma}{2}\varphi(1)+\cos\frac{\gamma}{2}\dot\varphi(1)
&=& 0 , \end{eqnarray}
which includes Newmann $\dot\varphi(0)=\dot\varphi(1)=0$ and
Dirichlet $\varphi(0)=\varphi(1)=0$ boundary conditions.
In the antidiagonal case  
\begin{equation}\label{antdia} U=
\left(\begin{array}{cc }0 & e^{-i\epsilon}\\ e^{i\epsilon} & 0
\end{array}\right)
\end{equation}
we have (pseudo-)periodic boundary conditions 
\begin{eqnarray}\label{per}
\varphi(1) &=& e^{i\epsilon}\varphi(0) \\ 
\dot\varphi(1) &=& e^{i\epsilon}\dot\varphi(0) 
\end{eqnarray}
$\varphi(1)=e^{i\epsilon}\varphi(0)$  with probability flux propagating
through the boundary.

\subsubsection{Self--adjoint extensions of Schr\"odinger operators in 1D}
We will concentrate our attention in 1D were we will be able to provide an elegant formula to solve the spectral problem for each self--adjoint extension.

Notice first that a compact 1D manifold $\Omega$ consists of a finite number of closed intervals $I_\alpha$, $\alpha = 1,\ldots,n$, $x_\alpha \in I_\alpha$ denoting the variable on each one of them.  Each interval will have the form $I_\alpha = [a_\alpha, b_\alpha] \subset \mathbb{R}$ and the boundary of the manifold $\Omega = \sqcup_{\alpha=1}^n [a_\alpha, b_\alpha]$ (disjoint union) is given by the family of points $\{ a_1, b_1, \ldots, a_n,b_n\}$.     Functions $\Psi$ on $\Omega$ are determined by vectors $(\Psi_1(x_1), \ldots, \Psi_n(x_n))$ of complex valued functions $\Psi_\alpha \colon I_\alpha \to \mathbb{C}$.

 A Riemannian metric $\eta$ on $\Omega$ is given by specifying a Riemannian metric $\eta_\alpha$ on each interval $I_\alpha$, this is, by a positive smooth function $\eta_\alpha(x_\alpha) > 0$ on the interval $I_\alpha$, i.e., $\eta\mid_{I_\alpha} = \eta_\alpha (x_\alpha) dx_\alpha^2$.    Then the inner product on $I_\alpha$ takes the form $\langle \Psi_\alpha , \Phi_\alpha \rangle = \int_{a_\alpha}^{b_\alpha} \bar{\Psi}_\alpha (x_\alpha) \Phi_\alpha (x_\alpha) \sqrt{\eta_\alpha (x_\alpha)} dx_\alpha$ and the Hilbert space of square integrable functions on $\Omega$ is given by $L^2(\Omega) = \bigoplus_{\alpha= 1}^n L^2(I_\alpha, \eta_\alpha)$.    Thus the Hilbert space $L^2(\partial \Omega)$ at the boundary reduces to $\mathbb{C}^{2n}$, as well as the subspaces $H^{3/2}(\partial \Omega)$ and $H^{1/2}(\partial \Omega)$.  The vectors in $L^2(\partial \Omega)$ are determined by the values of $\Psi$ at the points $a_\alpha$, $b_\alpha$ (with the standard inner product):
 $$\psi = (\Psi_1(a_1),\Psi_1(b_1), \ldots, \Psi_n(a_n),\Psi_n(b_n)).$$  
 Similarly we will denote by $\dot{\psi}$ the vector containing the normal derivatives of $\Psi$ at the boundary, this is:
 $$ \dot{\psi} = \left(  - \left. \frac{d\Psi_1}{dx}\right|_{a_1} , \left. \frac{d\Psi_1}{dx}\right|_{b_1}, \ldots, - \left.\frac{d\Psi_n}{dx}\right|_{a_n},\left.\frac{d\Psi_n}{dx}\right|_{b_n}  \right) .$$
 
Because of Thm. \ref{maximal} an arbitrary self--adjoint extension of the Schr\"odinger operator 
$$
\H = - \frac{1}{2} \bigoplus_\alpha \frac{d^2}{dx_\alpha^2} + V(x_1, \ldots, x_n) \, ,
$$ defined by the Riemannian metric $\eta$ and a regular potential function $V$ is defined by a unitary operator $V \colon \mathbb{C}^{2n} \to \mathbb{C}^{2n}$.  Its domain consists of those functions whose boundary values $\psi$, $\dot{\psi}$ satisfy Asorey's condition, Eq. (\ref{U_con}).  This equation becomes a finite dimensional linear system for the components of the vectors $\psi$ and $\dot{\psi}$.    Hence the space of self--adjoint extensions is in one--to--one correspondence with the unitary group $U(2n)$ and has dimension $4n^2$.
 
 It will be convenient for further purposes to organize the boundary data vectors $\psi$ and $\dot{\psi}$ in a different way.  Thus, we denote by $\psi_l\in \mathbb{C}^n$ (respec. $\psi_r$) the column vector whose components $\psi_l(\alpha)$, $\alpha = 1, \ldots, n$, are  the values of $\Psi$ at the left endpoints $a_\alpha$, this is $\psi_l(\alpha) = \Psi_\alpha(a_\alpha )$ (respec. $\psi_r(\alpha ) = \Psi_\alpha(b_\alpha )$ are the values of $\Psi$ at the right endpoints).   Similarly we will denote by $\dot{\psi}_l(\alpha ) =   - \frac{d\Psi_\alpha}{dx}\mid_{a_\alpha}$ and $\dot{\psi}_r(\alpha ) =   \frac{d\Psi_\alpha}{dx}\mid_{b_\alpha}$, $\alpha = 1, \ldots, n$.   Hence, the domain of the self--adjoint extension defined by the unitary matrix $U$ will be written accordingly as:
 \begin{eqnarray}\label{asorey_1D}
\psi_l - i\dot{\psi}_l &=& U^{11}(\psi_l + i\dot{\psi}_l) + U^{12}(\psi_r + i\dot{\psi}_r) \\ 
 \psi_r - i\dot{\psi}_r &=& U^{21}(\psi_l + i\dot{\psi}_l) + U^{22}(\psi_r + i\dot{\psi}_r)  \nonumber
 \end{eqnarray}
and $U$ has the block structure:
\begin{equation}\label{block}
 U = \left[  \begin{array}{c|c} U^{11} & U^{12} \\ \hline U^{21} & U^{22}  \end{array}    \right] .
 \end{equation}
 Notice that the unitary matrix $U$ is related to the unitary matrix $V$ above by a permutation, but we will not need its explicit expression here.
 
 Thus in what follows we will use the notation for the boundary data:
 $$ \psi = \left[ \begin{array}{c} \psi_l \\ \psi_r \end{array}\right]; \quad  \dot{\psi} = \left[ \begin{array}{c} \dot{\psi}_l \\ \dot{\psi}_r \end{array}\right]$$
 and Asorey's condition reads again:
 \begin{equation}\label{asorey2}
 \psi - i \dot{\psi} = U (\psi + i \dot{\psi}), \quad \quad  U \in U(2n) .
 \end{equation} 

%%%%%%%%%%%%%%%%%%%%%%%%%%%%%%%%%%%%%%%%%%%%%%%%%%%%%%%%%%%%%%%

\subsubsection{The spectral function}
Once we have determined a self--adjoint extension $H_U$ of the Schr\"odinger operator $H$, we can determine the unitary evolution of the system by computing the flow $U_t = \exp(-itH_U/\hbar )$.   It is well--known that the Dirichlet extension of the Laplace--Beltrami operator has a pure discrete spectrum because of the compactness of the manifold and the ellipticity of the operator, hence all self--adjoint extensions have a pure discrete spectrum (see \cite{We80}, Thm. 8.18).  Then the spectral theorem for the self--adjoint operator $H_U$ states:
$$ H_U = \sum_{k = 1}^\infty \lambda_k P_k ,$$
where $P_k$ is the orthogonal projector onto the finite--dimensional eigenvector space $V_k$ corresponding to the eigenvalue $\lambda_k$.  The unitary flow $U_t$ is given by:
$$ U_t = \sum_{k = 1}^\infty e^{-it\lambda_k / \hbar} P_k .$$
Hence all that remains to be done is to solve the eigenvalue problem:
\begin{equation}\label{eigenvalue_problem}
H_U \Psi = \lambda \Psi ,
\end{equation}
for the Schr\"odinger operator $H_U$.     We devote the rest of this section to provide an explicit formula to solve Eq. (\ref{eigenvalue_problem}).   

On each subinterval $I_\alpha = [a_\alpha, b_\alpha]$ the differential operator $H_\alpha = H|_{I_\alpha}$ takes the form 
of a Sturm--Liouville operator 
$$H_\alpha = - \frac{1}{W_\alpha} \frac{d}{dx} p_\alpha(x) \frac{d}{dx} + V_\alpha(x),$$ 
with smooth coefficients $W_\alpha = \frac{1}{2\sqrt{\eta_\alpha}} > 0$ (now and in what follows we are taking the physical constants $\hbar$ and $m$ equal to 1), $p_\alpha(x) = \frac{1}{\sqrt{\eta_\alpha}}$, hence the second order differential equation 
\begin{equation}\label{eigen_alpha}
 H_\alpha \Psi_\alpha = \lambda \Psi_{\alpha} 
 \end{equation}
has a two-dimensional linear space of solutions for each $\lambda$.   We shall denote a basis of solutions of such space as $\Psi_\alpha^\sigma$, $\sigma = 1,2$.  Notice that $\Psi_\alpha^\sigma$ depends differentially on $\lambda$.    Hence
a generic solution of Eq. (\ref{eigen_alpha}) takes the form:
$$ \Psi_\alpha = A_{\alpha,1} \Psi_{\alpha}^1 + A_{\alpha,2} \Psi_{\alpha}^2 .$$
Now it is clear that 
$$\psi_l(\alpha) = \Psi_\alpha (a_\alpha) = A_{\alpha,1} \psi_a^1 (\alpha) + A_{\alpha,2} \psi_a^2(\alpha) .$$
Hence:
$$ \psi_l = A_1 \circ \psi_a^1 + A_2 \circ \psi_a^2 ,$$
where $A_\sigma$, $\sigma = 1,2$, denotes the column vector
$$ A_\sigma = \left[  \begin{array}{c}  A_{1,\sigma} \\ \vdots \\ A_{n,\sigma} \end{array} \right] $$
and $\circ$ denotes the Hadamard product of two vectors, i.e., $(X\circ Y)_\alpha = X_\alpha Y_\alpha$ where $X,Y \in \mathbb{C}^n$.
We obtain similar expressions for $\psi_r$, $\dot{\psi}_l$ and $\dot{\psi}_r$.
With this notation Eqs. (\ref{asorey_1D}) become:
\begin{eqnarray}\label{asorey_eigen}
(\psi_l^1 - i\dot{\psi}_l^1)\circ A_1 + (\psi_l^2 - i\dot{\psi}_l^2)\circ A_2 &=& U^{11}(\psi_l^1 + i\dot{\psi}_l^1)\circ A_1 + U^{11}(\psi_l^2 + i\dot{\psi}_l^2)\circ A_2  \nonumber \\  +  U^{12}(\psi_r^1 + i\dot{\psi}_r^1)\circ A_1  &+& U^{12}(\psi_r^2 + i\dot{\psi}_r^2)\circ A_2
\end{eqnarray}
\begin{eqnarray*}
(\psi_r^1 - i\dot{\psi}_r^1)\circ A_1 + (\psi_r^2 - i\dot{\psi}_r^2)\circ A_2 &=& U^{21}(\psi_l^1 + i\dot{\psi}_l^1)\circ A_1 + U^{21}(\psi_l^2 + i\dot{\psi}_l^2)\circ A_2  \\ \nonumber +  U^{22}(\psi_r^1 + i\dot{\psi}_r^1)\circ A_1  &+& U^{22}(\psi_r^2 + i\dot{\psi}_r^2)\circ A_2
\end{eqnarray*}
It will be convenient to use the compact notation $\psi_{l\pm}^\sigma = \psi_{l}^\sigma \pm i \dot{\psi}_l^\sigma$, $\sigma = 1,2$, and similarly for  $\psi_{r\pm}^\sigma$.  

If $T$ is a $n\times n$ matrix and $X,Y$ arbitrary $n \times 1$ vectors, we will define $T \circ X$ as the unique matrix such that $(T\circ X)Y = T(X\circ Y)$.    The rows of the matrix $T\circ X$ are $T_i \circ X$ or alternatively, the columns of $T\circ X$ are given by $T^j X_j$ (no summation on $j$).   It can be proved easily that
\begin{equation}\label{hadamardT}
 T\circ X = T\circ (X \otimes \mathbf{1}) ,
 \end{equation}
where $\mathbf{1}$ is the vector whose components are all ones (i.e., the identity with respect to the Hadamard product $\circ$) and the Hadamard product of matrices in the r.h.s. of Eq. \eqref{hadamardT} is the trivial componentwise product of matrices. 
Using these results Eqs. (\ref{asorey_eigen}) become:
\begin{eqnarray*}
(I_n\circ \psi_{l-}^1 - U^{11}\circ \psi_{l+}^1 - U^{12}\circ \psi_{r+}^1  )A_1 + (I_n\circ \psi_{r-}^2 - U^{11}\circ \psi_{l+}^2 - U^{12}\circ \psi_{r+}^2  )A_2 &=& 0 \\
\nonumber  (I_n\circ \psi_{r-}^1 - U^{21}\circ \psi_{l+}^1 - U^{22}\circ \psi_{r+}^1  )A_1 + (I_n\circ \psi_{r-}^2 - U^{21}\circ \psi_{l+}^2 - U^{22}\circ \psi_{r+}^2  )A_2 &=& 0
 \end{eqnarray*}
Thus the previous equations define a linear system for the $2n$ unknowns $A_1$ and $A_2$.  They will have a non trivial solution if and only if the determinant of the $2n \times 2n$ matrix of coefficients $M(U, \lambda)$ below vanish:

$$ M(U, \lambda) = \left[  \begin{array}{cc}  I_n\circ \psi_{l-}^1 - U^{11}\circ \psi_{l+}^1 - U^{12}\circ \psi_{r+}^1 & I_n\circ \psi_{l-}^2 - U^{11}\circ \psi_{l+}^2 - U^{12}\circ \psi_{r+}^2  \\
I_n\circ \psi_{r-}^1 - U^{21}\circ \psi_{l+}^1 - U^{22}\circ \psi_{r+}^1  & I_n\circ \psi_{r-}^2 - U^{21}\circ \psi_{l+}^2 - U^{22}\circ \psi_{r+}^2  \end{array} \right] . $$

The fundamental matrix $M(U, \lambda)$ can be written in a more inspiring form using another operation naturally induced by the Hadamard and the usual product of matrices.    Thus, consider the $2n \times 2n$ matrix $U$ with the block structure of Eq. (\ref{block}) and the $2n \times 2$ matrices:
$$ [\psi_\pm^1 \mid \psi_\pm^2 ]  = \left[   \begin{array}{c|c} \psi_{l\pm}^1 & \psi_{l\pm}^2 \\  \hline \psi_{r\pm}^1 & \psi_{r\pm}^2 \end{array} \right] ,$$
then we define
$$ \left[  \begin{array}{c|c} U^{11} & U^{12} \\ \hline U^{21} & U^{22}  \end{array} \right]  \odot \left[   \begin{array}{c|c} \psi_{l\pm}^1 & \psi_{l\pm}^2 \\ \hline \psi_{r\pm}^1 & \psi_{r\pm}^2 \end{array} \right]  \equiv \left[   \begin{array}{c|c} U^{11}\circ \psi_{l\pm}^1 + U^{12} \circ \psi_{r\pm}^1& U^{11}\circ \psi_{l\pm}^2 + U^{12}\circ  \psi_{r\pm}^2 \\  \hline U^{21}\circ \psi_{l\pm}^1 + U^{22}\circ \psi_{r\pm}^1 & U^{21}\circ\psi_{l\pm}^2 + U^{22} \circ \psi_{r\pm}^2\end{array} \right] 
$$
and similarly
$$ I_{2n} \odot  [\psi_\pm^1 \mid \psi_\pm^2 ]  = \left[  \begin{array}{c|c}  I_n\circ \psi_{l\pm}^1 & I_n \circ \psi_{l\pm}^2 \\
\hline I_n\circ \psi_{r\pm}^1 & I_n \circ \psi_{r\pm}^2   \end{array}\right] .$$
Finally we conclude that the condition for the existence of coefficients $A_1$ and $A_2$ such that the solutions to the eigenvalue equation lie in the domain of the self--adjoint extension defined by $U$ is given by the vanishing of the spectral function $\Lambda_U (\lambda) = \det M(U, \lambda )$, that written with the notation introduced so far becomes:
\begin{equation}\label{spectral_function}
\Lambda_U(\lambda) =  \det ( I_{2n} \odot [\psi_-^1 \mid \psi_-^2 ]  - U \odot  [\psi_+^1 \mid \psi_+^2 ]  ) = 0 .
 \end{equation}
The zeros of the spectral function $\Lambda$ provide the eigenvalues $\lambda$ of the spectral problem Eq. (\ref{eigenvalue_problem}).

In the particular case $n =1$, the previous equation becomes greatly simplified, the Hadamard product becomes the usual scalar product and the Hadamard--matrix product is the usual product of matrices.  After some simple manipulations, the spectral function $\Lambda_U (\lambda)$ becomes:  
\begin{eqnarray}\label{spectral_1}
\Lambda_U (\lambda) &=& W(l,r,-,-) +U^{11}W(r,l,-,+)+U^{22}W(r,l,+,-) \\ &+& U^{12}W(r,r,-,+)  +U^{21}W(l,l,+,-)+\det U\cdot W(l,r,+,+)\notag
\end{eqnarray}
where we have used the notation:
$$ W(l,l,+,-) = \left|  \begin{array}{cc}  \psi_{l+}^1 & \psi_{l+}^2 \\ \psi_{l-}^1 & \psi_{l-}^2  \end{array} \right|, \quad W(l,r,+,-) = \left|  \begin{array}{cc}  \psi_{l+}^1 & \psi_{l+}^2 \\ \psi_{r-}^1 & \psi_{r-}^2  \end{array} \right|, \text{ etc.}$$
If we parametrize the unitary matrix $U \in U(2)$ as:
$$ U = e^{i\theta/2} \left[  \begin{array}{rc} \alpha & \beta \\ -\bar{\beta} & \bar{\alpha}  \end{array} \right], \quad | \alpha |^2 + | \beta |^2 = 1 ,$$
then the spectral function becomes:
\begin{eqnarray}\label{spectral_2}
\Lambda_U (\lambda) &=& W(l,r,-,-) +\alpha W(r,l,-,+)+\bar{\alpha}W(r,l,+,-)+\beta W(r,r,-,+)\\& &-\bar{\beta}W(l,l,+,-)+e^{i\theta} W(l,r,+,+)\notag
\end{eqnarray}
In particular if we consider a single interval $[0, 2\pi]$ with trivial Riemannian metric,  the fundamental solutions to the equation Eq. (\ref{eigen_alpha}) have the form $\Psi^1 = e^{i\sqrt{2\lambda}x}$ and $\Psi^2 = e^{-i\sqrt{2\lambda}x}$.  Then we have:
\begin{eqnarray*} 
W(l,r,-,-) &=& -2i(1+2\lambda)\sin(2\pi\sqrt{2\lambda})-4\sqrt{2\lambda}\cos (2\pi\sqrt{2\lambda}), \\ 
W(l,l,+,-) &=& 4\sqrt{2\lambda}, \\
W(r,r,-,+) &=& 4\sqrt{2\lambda}, \\
W(r,l,-,+) &=& 2i(1-2\lambda)\sin(2\pi\sqrt{2\lambda}),\\
W(r,l,+,-) &=& 2i(1-2\lambda)\sin(2\pi\sqrt{2\lambda}),\\
W(l,r,+,+) &=& -2i(1+2\lambda)\sin(2\pi\sqrt{2\lambda})+4\sqrt{2\lambda}\cos (2\pi\sqrt{2\lambda}),
\end{eqnarray*}
and finally we obtain the spectral function $\Lambda_U(\lambda)$: 
\begin{eqnarray*} \Lambda_U( \lambda ) &=& -2i(1+2\lambda)\sin(2\pi\sqrt{2\lambda})-4\sqrt{2\lambda}\cos (2\pi\sqrt{2\lambda}) +4i\operatorname{Re}(\alpha) (1-2\lambda)\sin(2\pi\sqrt{2\lambda})\\
& &+  8\operatorname{Im}(\beta) \sqrt{2\lambda}+ e^{i\theta}[-2i(1+2\lambda)\sin(2\pi\sqrt{2\lambda})+4\sqrt{2\lambda}\cos (2\pi\sqrt{2\lambda})].
\end{eqnarray*}

\subsubsection{Quantum wires and quantum Kirchhoff's law}

The discussion in the previous section allows to discuss a large variety of self--adjoint extensions
of 1D systems whose original configuration space  $\Omega = \sqcup_{\alpha=1}^n [a_\alpha, b_\alpha]$ consist of a disjoint union of closed intervals in $\mathbb{R}$.
It is clear that some boundary conditions $U\in U(2n)$ will lead to a quantum system with configuration space a 1D graph whose edges will be the boundary points $\{ a_1, b_1, \ldots, a_n,b_n\}$ of the original $\Omega$ identified among themselves according to $U$ and with links $[a_\alpha, b_\alpha]$. 

We will say that the self--adjoint extension determined by a unitary operator $U$ in $U(2n)$ defines a quantum wire made of the links $[a_\alpha, b_\alpha]$ if there exists a permutation $\sigma$ of $2n$ elements such that Asorey's condition for $U$ implies that $\psi (x_\alpha) = e^{i\beta_\alpha}\psi (x_{\sigma(\alpha)}$, and $x_\alpha$ such that $x_\alpha = a_\alpha$ if $\alpha = 1, \ldots, n$, or $x_\alpha = b_{\alpha-n}$ if $\alpha = n+1, \ldots, 2n$.

Notice that Asorey's condition:
$$\psi -i \dot{\psi} = U(\psi + i \dot{\psi}) $$
guarantees that the evolution of the quantum system is unitary, i.e., if we consider for instance a wave packet localized in some interval $[a_k,b_k]$ at a given time, after a while, the wave packet will have spread out across the edges of the circuit, however the probability amplitudes will be preserved.  In this sense we may consider Asorey's equation above as the quantum analogue of Kirchhoff's circuit laws, or quantum Kirchhoff's laws for a quantum wire.

%%%%%%%%%%%%%%%%%%%%%%%%%%%%%%
%%%%%%%%%%%%%%%%%%%%%%%%%%%%%%

\newpage

\section{Self-adjoint extensions and semiclassical boundary conditions}\label{section:Feyman}

\subsection{Classical boundary conditions and path integrals}
The action principle governs the classical and quantum dynamics 
of unconstrained systems. The classical dynamics is given by stationary
trajectories from the variational action principle
$$\frac{\delta S({\bf x}(t))}{\delta {\bf x}(t)}\Big |_{{\bf x}(t)={\bf x}_{\rm cl}(t)}=0,$$
 and the quantum dynamics 
is automatically implemented in the path integral formalism by the weight that
the classical action provides for classical trajectories
\begin{equation}\label{path}
\displaystyle 
K_T(x,y)={\rm e}^{-TH}(x,y)= 
\int_{\tiny{\begin{array}{c} x(T)=y \\ x(0)=x\end{array}}} \delta[x(t)] \,\, {\rm e}^{-\frac12\int_0^T \, S({\bf x}(t)) dt}.
\end{equation}
However,
for particles evolving  in a bounded domain $\Omega\subset \R^n$ the variational problem is
not uniquely defined. It is necessary to specify the evolution of the particles
after reaching the  boundary. On the other hand, the  very nature of the physical 
boundary imposes some constraints on the trajectories contributing
to the path integral. 

In fact,  the boundary imposes more severe constraints on the
classical dynamics than to the quantum one. This is due to the point-like
nature of the particle which requires that after reaching the boundary the
individual particle has to emerge back either at  the same point or at a  
different one of the boundary. The only freedom the particle has is where it 
emerges back  and with which momentum it emerges back. The emergence 
of the particle at a different point covers  the possibility that the domain
can be  folded and glued at the boundary  giving rise to
non-trivial topologies. In summary, the classical boundary
conditions are given by two maps:  an isometry of the boundary
$$\alpha:\partial\Omega\to\partial\Omega $$
{\rm and}
a positive density function
$$ \rho:\partial\Omega \to {\rm \R^+}$$
which specify the change of position and normal component of 
momentum of the trajectory of the particle upon reaching the boundary. 
The isometry $\alpha$ encodes the possible geometry and
topology generated by the folding 
of the boundary and the function $\rho$  is associated to 
the reflectivity (transparency or stickiness) properties of the boundary.
Once these two functions are specified the classical variational problem is
restricted to trajectories which satisfy the boundary conditions \cite{aim2}:
\begin{equation}
{\bf {\bf x}}(t_+)=\alpha({\bf x}(t_-)),
\label{cero}
\end{equation}
\begin{equation}
\boldsymbol{\nu}({\bf x}(t_+)))\cdot\dot 
{\bf x}(t_+) =-\rho({\bf x}(t_-))\, \boldsymbol{\nu}({\bf x}(t_-))\cdot\dot {\bf x}(t_-) 
\label{uno}
\end{equation}
and
\begin{eqnarray}
\label{dos}
 &&
\alpha_\ast( \dot {\bf x}(t_-) -  [\boldsymbol{\nu}({\bf x}(t_-))\cdot\dot {\bf
  x}(t_-)] \, \boldsymbol{\nu}({\bf x}(t_-))
\nonumber\\
&&\phantom{adadfa}
=\dot {\bf x}(t_+)- [\boldsymbol{\nu}({\bf x}(t_+))\cdot \dot {\bf x}(t_+)]  \, \boldsymbol{\nu}({\bf
   x}(t_+)
\end{eqnarray}
for any $t$ such that  ${\bf  x}(t)\in \partial\Omega$, where $\boldsymbol{\nu}$ denotes the exterior normal derivative
at the boundary  $\partial\Omega$ and $$ {\bf x}(t_\pm)=\lim_{s\to0} {\bf  x}(t\pm s).$$

This definition of classical boundary conditions is motivated by
the standard physical heuristic interpretation of boundary
conditions. Linear momentum is not conserved because it is 
partially or totally absorbed by the boundary. 
The major constraints on the choice of 
boundary conditions come first by the preservation of the notion
of point-like particle which requires that any trajectory which reaches
the boundary has to emerge as a single trajectory from the same 
boundary. The second requirement concerning the permitted 
changes of linear momentum at the boundary have to be compatible
with the action principle. This implies that classical trajectories 
are determined by the stationary points of the classical action, which for simplicity
we assume to be that of a free particle
$$S({\bf x})=\int dt\, g_{ij}{\,\dot x^i(t)}{\dot x^j(t)}.$$
The variational principle yields the celebrated Euler-Lagrange motion  
equations $\ddot{\bf x}(t)=0$ provided that the boundary term
\begin{equation}
\sum_{m=1}^N\biggl[ \delta{\bf x}(t_m^+)\cdot\dot 
{\bf x}(t_m^+)- \delta {\bf x}(t_m^-)\cdot\dot {\bf x}(t_m^-) \biggr]
\end{equation}
%$$\sum_{m=1}^N \delta {\bf x}(t_m)[1- \rho({\bf x})](t_m){\bf\dot x}(t_m)$$
vanishes, where the sum is over all points  $t_m$  where the trajectories
reach the boundary. The simpler way of fulfilling this requirement is
by imposing the vanishing of each individual term on the sum.
These conditions reduce to the boundary conditions (\ref{uno})(\ref{dos}).
provided that the 
permitted variations are tangent to the boundary. In this case the normal
component of $\delta {\bf x}(t_m)$ vanishes, i.e. the  points of trajectories
which reach the boundary are only allowed to move along the boundary.
This condition is reminiscent of Dirichlet condition for D-branes in 
string theory. The analogue of Neumann boundary conditions is
senseless for point-like particles, because it will require to consider 
only trajectories which reach the boundary with null linear momentum.
  
Simple but interesting types of boundary conditions already arise in the
Sturm-Liouville problem, $\Omega=[0,1]$. In such a case
the boundary of the configuration space is a discrete  two-points set, 
$\partial\Omega=\{0,1\}$. Examples of classical boundary conditions in
such a case are \cite{gift,aim2}:

\begin{enumerate}

\item[i)] $\phantom{ii}$ Neumann  (total absorption): $\alpha=I$, $\rho(0) = \rho(1)= \infty$.

\item[ii)] $\phantom{i}$Dirichlet (total reflection): $\alpha=I$, $\rho(0) = \rho(1) = 1$.

\item[iii)] Periodic:  $\alpha(0)=1$, $\alpha(1)=0$, $\rho(0) = \rho(1) = 1$.

\item[iv)]  Quasi-periodic:  $\alpha(0)=1$, $\alpha(1)=0$,
$\rho(0) = \rho(1) = \epsilon$.

\end{enumerate}

All these classical  boundary conditions have a quantum counterpart which
can be derived from the Feynmann's path integral approach \cite{Feynman, Feynman-Higgs,FeyKac}.

%%%%%%%%%%%%%%%%%%%%%%%%%%

\subsection{Path integrals and quantum boundary conditions}

The quantum implementation of classical boundary conditions is straightforward 
via the path integral method. The only paths to be considered in the Feynman's
path integral  \cite{Feynman, Feynman-Higgs,FeyKac} given by Eq. \eqref{path}
are those that satisfy the classical boundary conditions,
The corresponding quantum boundary operators
are
$$\varphi(\alpha(x)) +i \dot\varphi(\alpha(x)) = U \left[\varphi(\alpha(x)) +i \dot\varphi(\alpha(x))\right]=
-\frac{1-\rho(x)+i}{1-\rho(x)-i}\left[
\varphi(x)-i\dot \varphi(x)\right].$$
In the one--dimensional case of Sturm-Liouville problem the space of
quantum boundary conditions   is a four--dimensional Lie  group $U(2)$,
whereas the space of classical boundary conditions is the union of two 
disconnected two--dimensional manifolds,

\begin{eqnarray}
\label{doscon}
{\cal M}_1 &=&\{\psi\in L^2([0,1]), \psi(0)=({1-\rho_1})\dot \psi(1), \psi(1)=({1-\rho_0})\dot
\psi(0)  \} \\
\label{doscond}
{\cal M}_0 &=&\{\psi\in L^2([0,1]), \psi(0)=({1-\rho_0})\dot \psi(0), \psi(1)=({1-\rho_1})\dot \psi(1)  \} 
\end{eqnarray}

Thus, the  Feynman path integral approach
does not  cover the whole  set of boundary conditions.
 One of the reasons behind the failure of the path integral picture is 
the single valued nature of trajectories. Many conditions 
describe a scattering by a singular
potential sitting on the boundary. 
There are two different  types of quantum interactions with the boundary: 
reflection and diffraction. A classical
description of the phenomena without including a potential term will
require a splitting of the ongoing classical trajectory into two outgoing paths
one pointing forward and another one backwards. This picture destroys
the pure  point-like particle approach and leads to multivalued trajectories
which dramatically changes  the simple Feynman's description of path
integrals. Furthermore, there are boundary conditions where one single
trajectory upon reaching the boundary has to be split into an infinite
amount of outgoing trajectories. This behaviour can be explicitly
pointed out by noticing that the quantum evolution of a narrow
wave packet evolves backward after being scattered by  boundary 
as a quite widespread wave packet emerging from all points of the boundary.

In order to have a path integral description of all boundary
conditions we need to incorporate some random behaviour 
for the trajectories reaching the boundary  and complex phases
for those trajectories.  This is possible because the wave functions
are complex and the evolution operator involves complex amplitudes.
Although in this way we are able to describe any type of
unitary evolution in the bounded domain the method goes  far beyond
Feynman's pure action approach.

The prescription is quite involved and proceeds by considering 
instead of the Euclidean time evolution propagator $K_T$  the
resolvent operator $C_z$ of the Hamiltonian 
\be
\displaystyle C_z (x,y)={( z \I + H )^{-1}}(x,y)=\int_0^ \infty\frac{dT}{T} {\rm e}^{-zT} K_T(x,y). 
\ee

The Euclidean time propagator can be recovered from the resolvent by means of the following countour integral
\be
K_T(x,y)= \frac{1}{2\pi i}\oint \, C_z(x,y) {\rm e}^{zT} dz
\ee
along a contour which encloses the spectrum of $H$ on the real axis.

Boundary conditions can be  easily implemented into  the resolvent, whereas as we
shall see, the implementation in the Euclidean time propagator is much harder. 
Let us consider a fixed boundary condition, e.g.
the Neumann boundary conditions $U_0=\I$, and  consider the corresponding Hamiltonian $H_0$ as a
 background selfadjoint operator. The   selfadjoint extension of $H$ defined on the domain 
\be
 i(\I+U)\dot{\varphi}=(\I - U){\varphi}
\ee
by the unitary operator $U$
has a resolvent given by Krein's formula \cite{krein2ref}
\begin{equation}
\displaystyle C^U_z (x,y)= C^0_z (x,y) - \int_{\partial\Omega} dw\, \int_{\partial\Omega} dw'\,  C^0_z (x,w) R_z^U(w,w') C^0_z(w',y)
\label{krein2}
\end{equation}
 where $R^U$ is the operator of $ L^2(\partial\Omega) $ defined by 
\be
R_z^U=((\I - U) C^0_z -i (I+U))^{-1}(\I - U).
\ee

A similar formula could be obtained choosing another boundary condition as background
boundary condition instead of Neumann's condition.

The inverse transform permits to recover a formula for the propagator kernel of the type
\be
K_T(x,y)= K^0_T(x,y)- \frac{1}{2\pi i}\oint \, dz \,{\rm e}^{z T} \int_{\partial\Omega} dw\, \int_{\partial\Omega}
 dw'\,  C^0_z (x,w) R_z^U(w,w') C^0_z(w',y).
\ee
It is easy to rewrite  $K^0_T(x,y)$ as a path integral as in  (\ref{path})
restricting the trajectories to the interior of the domain $\Omega$ and counting twice the trajectories hitting
the boundary  ${\partial\Omega}$. However, in general,  the kernel  $K_T(x,y)$ cannot be rewritten
in terms of a path integral.  Only for a few boundary conditions the reduction can be achieved, but for generic
boundary conditions the kernel $K_T(x,y)$ has to be considered as a genuine boundary condition kernel
containing information about the  boundary jumps amplitudes and phases associated to the different
trajectories hitting the boundary.  The complex structure of this kernel reduces the utility of the path integral
approach and points out the behaviour of the boundary as a genuine quantum device.
This behaviour can be explicitly
pointed out by noticing that under certain boundary conditions  the quantum evolution of a narrow
wave packet  is scattered backward by the boundary 
as a quite widespread wave packet emerging from all points of the boundary.
However, there are cases \cite{amm} where this kernel adopts a simple form and the path integral approach
can be formulated in a very explicit way. In particular, for  Dirichlet boundary conditions  $U=-\I$,
\be
R_z^U=(C^0_z )^{-1}
\ee
and 
\be
C^D_z(x,y) = C^0_z(x,y)  - \int_{\partial\Omega} dw\, \int_{\partial\Omega} dw'\,  C^0_z (x,w) (C^0_z)^{-1}(w,w') C^0_z(w',y) 
\ee
which leads to a propagator  kernel given by the path integral  (\ref{path}) but restricted to paths which do not reach the
boundary  $\partial\Omega$.

The  method of images also permits   to  use unconstrained 
path integrals  to  describe systems with
non-trivial boundary conditions \cite{groschea, amm}.
However, in the case of higher dimensions   the method is not useful in the presence of  non symmetric boundaries 
and  the path integral cannot be
defined by a  simple prescription as in the Feynman original formulation.

However, the method is only restricted to similar  cases and  for generic  boundary conditions a closed form 
expression is not available. 
In higher dimensions the number of boundary conditions for
which the path integral method is useful to describe the
quantum evolutions is even more  limited. %\cite{sant}.

In summary, it is possible to generalise the   Feynman  approach
to describe the dynamics of quantum systems constrained
to bounded domains. For some  boundary conditions 
the modification of the path integral formula includes a phase factor
or a boundary weight for the trajectories which reach the boundary.
However, the method becomes not useful for generic boundary conditions
because the prescription becomes very intricate.

%%%%%%%%%%%%%%%%%%%%%%%%%%%%%%
%%%%%%%%%%%%%%%%%%%%%%%%%%%%%%
%%%%%%%%%%%%%%%%%%%%%%%%%%%%%%

\newpage

\section{The space of self-adjoint elliptic boundary conditions}\label{section:Grassmann}

%%%%%%%%%%%%%%%%%%%%%%%%%%%%%%

\subsection{The elliptic Grassmannian}\label{section:elliptic}

In the previous sections it was shown that self-adjoint extensions
of Dirac and Laplace operators are defined
by a family of subspaces (unitary or self-adjoint respect.) of the Hilbert space
of boundary data for each operator.  However we have not
considered whether or not the extensions $\Dsl_W$ obtained in this way for a Dirac or
Laplacian operator $\Dsl$, define elliptic operators or not, i.e., if the
boundary data given by the subspace $W$ constitute an elliptic boundary problem for
$\Dsl$.  

This is a crucial point because if the extensions considered were not elliptic, this
could affect dramatically the structure of the spectrum (for instance, loosing 
its discreteness), hence affecting the physical properties of the system 
in unwanted ways.  Thus, looking for elliptic extensions of the operator
$\Dsl$ is a way of restricting to a physical sector of the possible theories with 
`good' spectral properties.  

As it was commented in the introduction, the modern theory of elliptic boundary problems for Dirac operators on closed manifolds was developed in the seminal series of papers by Atiyah, Patodi and Singer \cite{At68} and, on manifolds with boundary, \cite{At75}.  The boundary conditions introduced there to study the index theorem for
Dirac operators on even-dimensional spin manifolds with boundary are
nowadays called Atiyah-Patodi-Singer boundary conditions (APS BC's).    
The crucial observation there was
that global boundary conditions were needed in order to obtain an elliptic
problem, contrary to the situation for second order operators
where (local) Dirichlet conditions, for instance, are elliptic.   Such boundary
conditions were extended to include also odd dimensional spin manifolds
with boundary (see \cite{Da94} and references therein).  

Later on, E. Witten (\cite{Wi88}, \S II), pointed out the link between
elliptic boundary conditions for the Dirac operator on 2 dimensions and 
the infinite dimensional Grassmannian manifold.  The infinite
dimensional Grassmannian was introduced previously in the analysis of
integrable hierarchies and discussed extensively by
Segal and Wilson (see \cite{Se85} and references therein).   More
recently Schwarz and Friedlander \cite{Fr97} have
extended Witten's analysis to arbitrary elliptic operators on arbitrary
dimensional manifolds with boundary.  The particular analysis for Dirac
operators we are using follows from \cite{At75} but it can be extended also to higher order operators.  More
comments on this will be found later on. 

The basic idea is that the space of zero modes of a Dirac (or
Laplace) operator $\Dsl$, $\ker \Dsl = \set{\xi\in \Gamma (S) \mid \Dsl\xi = 0}$,
induces a subspace at the boundary $b(\ker \Dsl)$ that in general will be
infinite-dimensional, hence a way to restore ellipticity will be to
restrict to a subspace such that the kernel and cokernel of the
operator defined on this subspace will be finite dimensional.  We will perform
such analysis for the Dirac and Laplace operators and we will refer to
\cite{Fr97} for the general case.

The analysis of such projection requires a detailed description of
solutions near the boundary.   Assuming that the
Riemannian metric on $\Omega$ is a product near the boundary
$\partial\Omega$, we can decompose the operator $\Dsl$ in a collar
neighborhood  $U_\epsilon (-\epsilon,0]\times \partial \Omega$ of the boundary as (Eq. \eqref{DpartialOmega}):
$$
\Dsl = \nu\cdot (\partial_\nu + \Dsl_{\partial\Omega}) \, ,
$$
where $\Dsl_{\partial\Omega}$ is the Dirac operator on the boundary bundle
$S_\pO $.   A natural set of boundary conditions for our problem will be constructed as
follows.

Recall from \S \ref{squaring} that $\Dsl_{\partial\Omega}$ is an
elliptic and, because $\partial\Omega$ is closed, essentially self-adjoint differential operator on
$\partial\Omega$ that anticommutes with $J$, i.e.
$\Dsl_{\partial\Omega} J = - J \Dsl_{\partial \Omega}$.  The Dirac Laplacian
$\Dsl_{\partial\Omega}^2$ is a non-negative self-adjoint elliptic operator
with a real discrete spectrum $\spec \Dsl_{\partial\Omega}^2 = \set{\mu_k
\mid 0 \leq \mu_0 < \mu_1 < \cdots }$ and such that the
eigenspaces $E(\lambda_k)$ are finite dimensional, $E(\mu_k ) =
\set{\phi_k \in \sp H_D \mid \Dsl_{\partial\Omega}^2\phi_k = \mu_k
\phi_k}$.  The kernel $K$ of $\Dsl_{\partial\Omega}$ agrees with $\ker
\Dsl_{\partial\Omega}^2$ and with $E(0)$.  We have thus the following
orthogonal decomposition of $\sp H_D$,
$$ 
\sp H_D = \bigoplus_{k=0}^\infty E(\mu_k) = K \oplus
\bigoplus_{k=1}^\infty \left( E_+(\lambda_k) \oplus E_-(\lambda_k) \right) \, ,
$$
where we have set $E(\mu_k) = E_+(\lambda_k) \oplus E_-(\lambda_k)$, with $\lambda_{\pm k} = \pm \sqrt{\mu_k}$.
Then the spectrum of $\Dsl_{\partial \Omega}$ is given by $\{Ê\pm \lambda_k \}$.
Moreover if we denote by $\spec_+$ the non-negative spectrum of $\Dsl$, and by $\spec_-$ its
negative spectrum, then we may write
$$ 
\mathcal{H}_D = \bigoplus_{\lambda \in \spec_+} E_+(\lambda) \oplus \bigoplus_{\lambda \in \spec_-} E_-(\lambda) \, .
$$ 
The subspaces $T_+ = \bigoplus_{\lambda \in \spec_+} E_+(\lambda)$ and $T_- = \bigoplus_{\lambda \in \spec_-} E_-(\lambda)$
define a polarisation of $\mathcal{H}_D$.  Denoting by $P_\pm \mathcal{H}_ \to T_\pm$ the corresponding orthogonal
projectors, the celebrated Atiyah-Patodi-Singer boundary conditions (APS BC) are given by the (non-local)
condition:
$$
P_+b(\xi) = 0 \, , \qquad \xi \in H^1(\Omega, S) \, .
$$
In other words, APS BC select the negative spectrum of the boundary Dirac's operator $\Dsl_\pO$.

On the other hand the polarization $\sp H_D = \sp H_+ \oplus \sp H_-$
defined by the compatible complex structure $J_D$, $J_D (\sp H_\pm) =
\mp i\sp H_\pm$, induces a decomposition of the eigenspaces
$E(\lambda_k)$ as
$$ E(\lambda_k ) = E_\pm (\lambda_k) = E(\lambda_k) \cap \sp H_\pm .$$
Moreover, $\Dsl_{\partial\Omega}$ restricts to a map $\Dsl_k =
\Dsl_{\partial\Omega} Ê\mid_{E(\lambda_k)} \colon E(\lambda_k) \to
E(\lambda_k)$ and because it anticommutes with $J$, we have that $\Dsl_k
\colon E_\pm (\lambda_k) \to E_\mp (\lambda_k)$, thus $\Dsl_k$ has the
following block structure,
$$ \Dsl_k = \matriz{c|c}{0 & \Dsl_k^+ \\ \hline \Dsl_k^- & 0 } ,$$
and because $\Dsl_k$ is self-adjoint, $(\Dsl_k^-)^\dagger = \Dsl_k^+$.  On the
other hand $\Dsl_k^2 = \Dsl_{\partial\Omega}^2\mid_{E(\lambda_k)} = \lambda_k
I$, hence the spectrum of $\Dsl_k$ on $E(\lambda_k)$ consists of
$\pm\sqrt{\lambda_k}$.   The operator $\Dsl_k$ is invertible in
$E(\lambda_k)$ for $k\geq 1$, hence $\dim E_+ (\lambda_k) = \dim E_-
(\lambda_k)$.  Moreover $K = K_+ \oplus K_-$, and $\dim K_+ = \dim
K_-$.  In fact the index of the operator $\Dsl_0^+$ is zero because
$\partial \Omega$ is cobordant to $\emptyset$ and the index is cobordant
invariant.   We can choose an orthonormal basis
$\phi_{k,\alpha}^\pm \in E_\pm (\lambda_k)$, $\alpha = 1,\ldots, \dim
E_\pm (\lambda_k)$,  such that
$$ \Dsl_k \phi_{k,\alpha}^\pm = \pm i \sqrt{\lambda_k} \phi_{k,\alpha}^\mp
.$$
The Cayley transform discussed in \S \ref{Cayley_D} diagonalizes the operators
$\Dsl_k$, and we have 
$$\xi_{k,\alpha}^\pm = \phi_{k,\alpha}^+ \pm i \phi_{k,\alpha}^- \in \sp
L_{\pm} ,$$
and
$$ \Dsl_k \xi_{k,\alpha}^\pm = \pm \sqrt{\lambda_k} \xi_{k,\alpha}^\pm .$$
Then, it is clear that $b(\ker \Dsl_{\mathrm{APS}}) = \sp L_+$.  Moreover the orthogonal
projectors $\pr_\pm \colon\sp  H_D \to \sp L_\pm$ are pseudodifferential
operators whose complete symbol depends only on the coefficients of $\Dsl$. 
Thus, elliptic boundary conditions will be defined by subspaces $W\subset
\sp H_D$ such that $W\cap \sp L_+$ will be finite dimensional (notice
that such intersection corresponds to solutions of $\Dsl\xi = 0$ with
boundary values on $W$), this means that the projection
$\pr_+\mid_W$ will have a finite dimensional kernel.  Moreover, the
cokernel of $\pr_+$ will have to be finite-dimensional if
$\Dsl^\dagger$ is elliptic too.  Finally, if the extension $\Dsl_W$ is
elliptic, then there will exists left and right parametrics for it, and
this will imply that the projection $\pr_-\mid_W$ will have to be
compact operators.  It is sometimes convenient to restrict the last
assumption to operators of Hilbert-Schmidt class.   We can conclude that
the set of closed subspaces $W$ of $\sp H_D$ verifying the
following conditions:

\begin{enumerate}

\item[i.]  $\pr_+\mid_W \colon W \to \sp L_-$ is Fredholm.

\item[ii.] $\pr_-\mid_W \colon W \to sp L_+$ is Hilbert-Schmidt,

\end{enumerate}

determines all elliptic extensions of the Dirac operator $\Dsl$.  Such
space will be called the elliptic infinite dimensional Grassmannian of
$\Dsl$, or elliptic Grassmannian for short and will be denoted by $\mathrm{Gr}(\sp L_-, \sp L_+)$.   

The elliptic Grassmannian can be constructed also in terms of the
polarization $\sp H_+ \oplus \sp H_-$ instead of $\sp L_+ \oplus \sp
L_-$.   This is the approach taken for instance in \cite{Da94}.  In such
case, we will relate self-adjoint extensions of $\Dsl$ with unitary
operators $U\colon \sp H_+ \to \sp H_-$, hence elliptic boundary
conditions correspond to unitary operators $U$ such that the projection
from its graph to $\sp H_+$ would be Fredholm and the projection onto
$\sp H_-$ would be Hilbert-Schmidt.  It is obvious that the Cayley
transform $C$ defines a one-to-one map from $\mathrm{Gr}(\sp H_-,\sp H_+)$ into
$\mathrm{Gr}(\sp L_-, \sp L_+)$ (the map is actually a diffeomorphism, see
below), but it is important to keep in mind that the objects in the two
realizations of the Grassmannian are different.  

We will call in what
follows the elliptic boundary conditions defined by points in the
elliptic Grassmannian, generalized APS boundary conditions.   It is
important to point it out here that a parallel discussion takes place
for the discussion of elliptic extensions of the Laplacian operator. 
In fact replacing $\sp H_D$ by $\sp H_B$, $J_D$ by $J_B$ and
considering the boundary Laplacian $-\Delta_{\partial\Omega}$ we will
obtain that the elliptic extensions of $-\Delta$ are in one-to-one
correspondence with points in the elliptic Grassmanniand $\mathrm{Gr}(\sp L_-,\sp
L_+)$, and similarly, by using the Cayley
transform, with points in the elliptic Grassmannian $\mathrm{Gr}(\sp H_-, \sp
H_+)$.  Hence, in what follows we will omit the subindex in the notation
of the different boundary Hilbert spaces and operators and we will refer
simultaneously to the Dirac and/or Laplacian operators, and the elliptic
Grassmannian defining their elliptic extensions will be denoted simply
by $\mathbf{Gr}$.

The elliptic infinite dimensional Grassmannian has an important
geometrical and topological structure.  We must recall first (see for
instance Pressley and Segal \cite{Pr86} for more details) that $\mathbf{Gr}$ is
a smooth manifold whose tangent space at the point $W$ is given by
the Hilbert space of Hilbert-Schmidt operators $\sp J_2 (\sp L_-,\sp
L_+)$, from $\sp L_-$ to $\sp L_+$.  The group of linear continuous
invertible operators $GL (\sp H)$ does not act on $\mathbf{Gr}$ but only a
subgroup of it, the restricted general linear group $GL_\mathrm{res} (\sp H)$,
which defines the restricted unitary group $U_\mathrm{res}(\sp H) =
GL_\mathrm{res}(\sp H)Ê\cap U(\sp H)$.   The groups $GL(\sp H)$ and $U(\sp H)$
are contractible but $GL_\mathrm{res}(\sp H)$ and $U_\mathrm{res}(\sp H)$ are not.  The manifold $\mathbf{Gr}$ is
not connected and is decomposed in its connected components defined by
the virtual dimension of their points which is simply the index of the
Fredholm operator $pr_+\mid_W$, then, $\mathbf{Gr} = \cup_{k\in \Z} \mathrm{Gr}^{(k)}$. 

The Grassmannian $\mathrm{Gr}(\sp L_-, \sp L_+)$ carries a natural K\"ahler
structure defined by the hermitian structure given by
$$ h_W (\dot A, \dot B ) = \Tr \dot A^\dagger \dot B ,$$
where $\dot A, \dot B\in T_W \mathrm{Gr}(\sp L_-, \sp L_+)$ are Hilbert-Schmidt
operators from $\sp L_-$ to $\sp L_+$.  The imaginary part defines a
canonical symplectic structure $\omega$,
\begin{equation}\label{symp} \omega_W (\dot A, \dot B) = -\frac{i}{2} \Tr (\dot
A^\dagger \dot B - \dot B^\dagger \dot A ) \, .
\end{equation}
The elliptic Grassmannian is quasicompact in the sense that the only holomorphic
functions are constant.

\subsection{The space of self--adjoint extensions: the self-adjoint
Grassmannian and elliptic self-adjoint extensions}

We have characterized the self-adjoint extensions of a given Dirac or
Laplacian operators as the space $\sp M$ of self-adjoint subspaces of a
boundary Hilbert space $\sp H$ carrying a polarization $\sp H = \sp L_-
\oplus \sp L_+$.   On the other hand, we have seen in the previous
section that the Grassmannian $\mathrm{Gr}(\sp L_-, \sp L_+)$ describes the
elliptic extensions of such operator.  Then, the elliptic self-adjoint
extensions of the given operators will be given by the intersection
$\sp M \cap \mathrm{Gr}(\sp L_-, \sp L_+)$.  This space will be called the
elliptic self-adjoint Grassmannian or the self-adjoint Grassmannian for
short.  It is possible to see that the
self-adjoint Grassmannian is a smooth submanifold of the Grassmannian
and decomposes in connected components which are submanifolds of the
components $\mathrm{Gr}^{(k)}$.  We will denote the elliptic self-adjoint
Grassmannian as $\sp M_{\ellip}$.  The most relevant topological and
geometrical aspects of $\sp M_{\ellip}$ are contained in the following
theorem.

\begin{theorem}\label{Ellip_Lag}  The elliptic self-adjoint Grassmannian $\sp M_{\ellip}$ is a Lagrangian
submanifold of the infinite dimensional Grassmannian $\mathbf{Gr}$.
\end{theorem}

\begin{proof} That $\sp M_{\ellip}$ is an isotropic submanifold of $\mathrm{Gr}(\sp
L_-, \sp L_+)$ follows immediately from Eq. (\ref{symp}) and the
observation that tangent vectors to $\sp M_{\ellip}$ at $W$ are defined
by self-adjoint operators.  

Now, all we have to do is to compute $T_W \sp M_{ellip}^\perp$ at $W = 0$
because of the homogeneity of the Grassmannian.  Hence, if $\dot A \in 
T_0 \sp M_{ellip}^\perp$, this means that 
$$ \Tr (\dot A^\dagger \dot B - \dot B \dot A) = 0 ,$$  
for every self-adjoint $\dot B \in \sp J_2 (\sp L_-,Ê\sp L_+ )$, hence
$\dot A^\dagger - \dot A = 0$, and $\dot A$ is self-adjoint, then lying
in $T \sp M_{\ellip}$.
\end{proof}

A Lagrangian submanifold of a compact manifold carries a characteristic
class called the Maslov class which is an element of first cohomology
group of the manifold with integer coefficients.  The Maslov class is
the dual of the Maslov cycle as constructed by Arnold \cite{Ar67}.  For
reasons that will be clear later on we will call the dual of the Maslov
class for $\sp M_{\ellip}$ the Cayley-Maslov surface and we will devote
the remaining part of this section to describe this cycle and the Maslov
class of $\sp M_{\ellip}$, $\nu \in H^1 (\sp M_{\ellip}, \Z)$.

The Cayley-Maslov surface $\sp C$ is the subspace of the self-adjoint
Grassmannian $\sp M_{\ellip}$ that contains the
self-adjoint subspaces that cut $L_-$ in a space of dimension $\geq
1$.  The Cayley-Maslov subspace $\sp C$, is a stratified manifold, $\sp
C = \cup_{k \geq 1}\sp C^{(k)}$,  that contains an open dense submanifold
consisting of the space of self-adjoint subspaces whose intersection with
$\sp L_-$ is exactly 1, the space denoted by $\sp C^{(1)}$.   This is
enough for topological purposes to analyze the crossings of the
Cayley-Maslov surface.   In what follows we will denote $\sp C^{(1)}$
simply by $\sp C$ and we will call it the Cayley-Maslov surface.  In
\cite{Ar67} it is proved that the Cayley-Maslov surface is two sided. 
The proof works exactly the same in our setting, i.e., there is a
non-vanishing vector field transversal to it and given by the tangent
vector to the curve of self-adjoint operators,
$$
A_t = \frac{A- \tan \frac{t I}{2}}{\tan \frac{t A}{2} + I} \, ,
$$
which is simply the image by the Cayley transform of the curve in the
space of unitary operators obtained by multiplication by $\exp it$. 
We will call the positive side of the Cayley surface the exiting of
the previous curve and the negative side the exiting of the curve
$A_{-t}$.

Given a continuous curve $\gamma \colon [0,1] \to \sp M_{\ellip}$ we will
define its index as the sum of positive crossing minus the sum of
negative crossings, where a crossing is positive if it is done from the
negative to the positive side and negative conversely.  For curves in
general position such number will be finite.

The Cayley-Maslov surface defined in this way can be called the
Cayley-Maslov cycle of the self-adjoint Grassmannian and the Cayley index
is simply the intersection index of the Cayley cycle.  The identification
of the self-adjoint  Grassmannian with a subgroup of unitary operators
allows for an alternative cohomological way of computing such index.

We should remark first that if $W_A$ is a generic element in the
self-adjoint Grassmannian, the corresponding unitary operator $U_A$
given by its Cayley transform is of the form
$I + K_A$ where $K_A$ is Hilbert-Schmidt.  
In fact, it is clear that
$$K_A = \frac{2iA}{I - iA} ,$$
hence,
$$K_A^\dagger K_A = \frac{4A^2}{I + A^2} ,$$
and then,
$$ \Tr K_A^\dagger K_A = 4\Tr \frac{A^2}{I + A^2} \leq 4 \Tr A^2 < \infty
.$$ 
Then, we can define the determinant of $U_A$ in an standard way using
the regularized determinant
$$\lg {\det}' \, U = \sum_{i=1}^\infty \lg \frac{1+k_i}{e^{k_i}} ,$$
which is finite for all unitary operators differing from the unity by
a Hilbert-Schmidt operator.

Given a closed curve $\gamma \colon S^1 \to U(\sp H_+, \sp H_-)$
whose image lies in the image by the Cayley transform of the
self-adjoint  Grassmannian, i.e. $\gamma \in C(\sp M_{\ellip})$, we will
define the index of $\gamma$ as the winding number of the curve
$\det' \circ \gamma \colon S^1 \to S^1$, in other words,
$$\sharp (\gamma) = \frac{1}{2\pi} \int_0^{2\pi} {\det} '  (\gamma
(\theta)) d\theta = \frac{1}{2\pi} \int_{S^1} ({\det} ' )^* d\theta
.$$

In order to show that the winding number $\sharp \gamma$ coincides
with the Cayley index of $\gamma$ we are going to introduce an
alternative way of computing such winding number.

Given a unitary operator we will define its degenerate dimension as
the dimension of the eigenspace with eigenvalue $1$.  If $U$ is of
the form above, $U = I + K$ with $K$ Hilbert-Schmidt, then the
dimension of the eigenspace  of eigenvalue $1$ is finite and the
degenerate dimension of the operator is finite.  We shall denote such
number by $\nu (U)$.   If $U_t$ is a curve $\gamma$ of such unitary
operators, then we define the index of such curve as
\begin{equation}\label{nu} \nu (\gamma ) = \int_0^1 \nu (U_t ) dt .
\end{equation}
We can see that $\nu (U_t)$ is of bounded variation because of the
continuity on the norm topology on the space $C(\sp M_{\ellip})$, then the
integral in Eq. (\ref{nu}) is finite.

\begin{theorem}  $\sharp \gamma = \nu (\gamma)$.
\end{theorem}

\begin{proof}  The crucial observation to prove the
formula above is to realize that $\nu (\gamma)$ is the net number of
eigenvalues of $U_t$ that cross through $-1$.  On the other hand after
substracting a global term in the definition of the determinant
$\det^\prime U_t$ which corresponds to the eigenvalues that do not
move away of a compact set, the others, a finite number, wind around the
unit circle, and the determinant counts the sum of the winding number
of all of them.  Then, the equality follows.
\end{proof}

%\bigskip

The previous index will be called the Cayley-Maslov class and it defines a
nontrivial cohomology class in $H^1 (\sp M_{\ellip}, \Z)$.

\begin{theorem}
The Cayley index of a curve $\gamma$ and the winding number $\sharp
\gamma$ of $\gamma$ agree.
\end{theorem}

This follows easily from the fact that crossing the eigenvalue 
$-1$ for a curve $U_t$ of unitary operators is equivalent to $A_{U_t}$
crossing the Cayley surface in $\Lambda_S$.  Thus counting the
crossings in both pictures gives the same number.

%\subsection{The index of the Cayley surface and adiabatic change of
%boundary conditions}

If we perform an adiabatic change in the boundary conditions defining
self-adjoint extensions of the operator $\H$, the spectrum of such
extensions will change.   It could happen that in such
deformation process we will approach an unstable sector of the
theory.   This will happen when crossing the Cayley
surface.  On the other hand, as it was discussed before, the number of crossings of the Cayley surface with
appropriate signs defines an integer number that
is related to the topology of the space of elliptic self-adjoint extensions. 
We will call it the Cayley index and it has a similar
geometrical origin as the Maslov-Arnold index but quite a different
physical meaning.

\newpage

%%%%%%%%%%%%%%%%%%%%%%%%%%%%%%%%
%%%%%%%%%%%%%%%%%%%%%%%%%%%%%%%%
%%%%%%%%%%%%%%%%%%%%%%%%%%%%%%%%

\section{Self-adjoint extensions of dissipative systems}\label{section:dissipative}

In discussing self-adjoint extensions of symmetric elliptic
operators we have found that they are described by a Lagrangian
submanifold in the universal infinite dimensional elliptic Grassmannian. 
The remaining describes extensions that are not
self-adjoint.  We will start in this section a discussion on the meaning
of such extensions, both mathematically and physically.    From this
discussion we will learn that the self-adjointness of these extensions
can be restored adding an external ``effective'' Hilbert space to our
system.  

%%%%%%%%%%%%%%%%%%%%%%%%%%%%%%%%%
%%%%%%%%%%%%%%%%%%%%%%%%%%%%%%%%%

\subsection{Non-self-adjoint extensions and local evolution}

To set up the discussion in concrete terms we will consider a particle
moving freely on a Riemannian manifold $(\Omega, \eta)$ with boundary $\partial \Omega \neq \emptyset$, and
described in quantum mechanical terms by the Laplace-Beltrami Hamiltonian $H_0 = -\frac{1}{2}\Delta$.  

Let $W$ be a non-selfadjoint subspace of the elliptic Grassmannian
$\mathbf{Gr}$, i.e., $W \neq W^\dagger$.  If $W$ is the graph of an
operator $T$, this means that $T$ is non-selfadjoint.  Moreover the Cayley
transform on $T$ will define a linear isomorphism $C_T \colon \sp H_+
\to \sp H_-$ that will not be unitary.  

This lack of unitarity reflects
the fact that probability is not preserved at the boundary, i.e., the
non-unitary evolution semigroup defined by the extended operator will not
preserve the norm of states and, for instance, states could
``evaporate''.  

Another way of putting it is that this situation is describing a dissipation
of some type acting on the system.  Because of the structure of the
system (the operator is symmetric on the interior of
$\Omega$) the only place where this dissipation can occur is at the
boundary, however localized as it is, it affects instantaneously the 
system as a whole.  We may discuss this aspect briefly.   

It is well known that a free non-relativistic wave packet localized in a bounded region at time $t=0$ spreads instantaneously over all of space \cite{He98}.  Thus
if we consider a quantum system defined on a manifold with boundary with self-adjoint boundary conditions, and 
we modify them to non-self-adjoint ones, i.e., the system becomes dissipative, such modification, even if it is performed adiabatically, will affect the state of the system instantaneously even if such state is localized far away from the boundary contrary to what a naive perturbative analysis would suggest.

To make this analysis more precise, let us consider a fixed smooth section $\Psi_0$ with compact
support $K = \mathrm{sup} \Psi$ on the interior of $\Omega$, that is $\Psi \in C_0^\infty(\Omega)$.
Then we may consider a larger open set $\mathcal{U} \subset \overline{\mathcal{U}} \subset \Omega \backslash \partial \Omega$.   Then consider the smooth manifold $\Omega' = \overline{\mathcal{U}}$ with boundary $\partial \Omega' = \overline{\mathcal{U}}- \mathcal{U}$. If we denote by $\iota \colon \Omega' \to \Omega$ the canonical embedding, then we equip $\Omega'$ with the Riemannian metric $\iota^*\eta$ that we will denote $\eta'$.
In the same way we may pull-back to $\Omega'$ any further structure on $\Omega$, a vector bundle $E\to \Omega$, a connection $\nabla \colon \Gamma(E) \to \Gamma(E\otimes T^*\Omega)$, etc.,
that will be denoted in the same fashion $E'$, $\nabla'$, etc.
Now we may consider the Hilbert space $L^2(\Omega',E')$ of square integrable sections of the pull-back of the bundle $E$ over $\Omega$ to $\Omega'$ with respect to the metric $\eta'$.   The Laplace-Beltrami operator on $\Omega'$ defines a symmetric operator on $L^2(\Omega',E')$ and Dirichlet's boundary conditions provide a self-adjoint extension of it.

There is a natural continuous isometry from $L^2(\Omega',E')$ to $L^2(\Omega,E)$, induced by the embedding map $\iota_* \colon C_c^\infty(E') \to C_c^\infty(E)$ given by $(\iota_*\Psi)(x) = \Psi(x')$, if $x = \iota(x')$, $x'\in \Omega'$, and $0$ otherwise.  Notice that:
$$
\parallel \iota^*\Psi \parallel_{L^2(\Omega,E)}^2 = \int_\Omega \mid \Psi (x) \mid^2 \vol_{\eta} = \int_{\Omega'} \mid \Psi (x') \mid^2 \vol_{\eta'} =  \parallel \Psi \parallel_{L^2(\Omega',E')}^2 \, ÊÊ\forall \Psi \in C_c^\infty(E') \, ,
$$
hence, because both spaces $\colon C_c^\infty(E')$ and $C_c^\infty(E)$ are dense in $L^2(\Omega',E')$ and $L^2(\Omega,E)$ respectively, $\iota_*$ extends to an isometry $\hat{\iota} \colon L^2(\Omega',E') \to L^2(\Omega,E)$.   In this sense we may consider $L^2(\Omega',E')$ as a closed subspace of $L^2(\Omega,E)$.    Notice now that $\Psi_0 \in L^2(\Omega',E')$ by construction.

Let $H_W$ be the elliptic extension of the Laplace-Beltrami operator determined by the boundary conditions defined by the subspace $W$.   Because of the correspondence between elliptic boundary conditions and subspaces $N$ of the deficiency space $\mathcal{N}$ of the operator $H_0$, then the extension of the operator defined by the subspace $N$ has a domain $D_W = D_0 \oplus N$,
and it acts on an state $\Psi = \Psi_0 + \xi$,  $\Psi_0 \in D_0$,$\xi \in N$, as $H_W(\Psi) = H_0\Psi_0 + K\xi$.
But in the situation above the state $\Psi_0$ that we will use as initial data for the evolution problem:
\begin{equation}\label{evolution_W}
i\frac{\partial}{\partial t} \Psi = H_W \Psi \, , \qquad \Psi\in D_K \, , 
\end{equation}
lies in $D_0$, the minimal extension domain, hence:
$$
H_W\Psi_0 = H_0\Psi_0= H_D\Psi_0\, ,
$$
where $H_D$ is the minimal selfadjoint extension corresponding to Dirichlet boundary conditions
 $\Psi\in D_D(\Omega') \, ,  \Psi\mid_{\partial \Omega'} = 0$.

However, one can show that even if any power $n$ of the Hamiltonian satisfies:
$$
H^n_W \Psi_0=H^n_D \Psi_0 \, ,
$$
the actual evolutions of the system governed by $H_W$ is very different from
that governed by Dirichlet boundary conditions. This different behaviour is due to
the fact that the local time evolution of the quantum system is not perturbative in $t$. 

%Hence if for a period of time (short enough) the solution of the equation Eq. \eqref{evolution_W} will remain in the domain $D_0$, its evolution will be exactly the same
%as the solution of the initial value problem:
%$$
%i\frac{\partial}{\partial t} \Psi = H_0 \Psi \, , \qquad \Psi\in D_0(\Omega') \, , \quad \Psi\mid_{\partial \Omega'} = 0 \, , 
%$$
%with Dirichlet's boundary conditions, hence it will be unitary.  

%We may argue that because the radius of the support of a section on $C_c^\infty(\Omega)$ grows linearly in time, for times short enough, the state $\Psi_t$ describing the solution of the initial value problem Eq. \eqref{evolution_W} will remain inside the open set $\mathcal{U}\subset \Omega'$ because the distance $\mathrm{list}(K,\partial \Omega') > 0$,
%however this is not the case, because the dynamics at the boundary represented by the operator $K$ will create a contribution to the state of the system near the boundary for all time $t>0$. 

In fact, that is exactly what from a physical perspective we should
expect; the system under study is in contact with an exterior
system represented by the boundary.  The interaction between them
is represented by an ``effective'' action described by the boundary
conditions and they contribute instantaneously to the evolution of the system.
 
For instance we can imagine that the boundary is an actual
boundary made of a semitransparent mirror or membrane, with a given
coefficient of reflection and transmission.  Then, a given fraction of
the probability amplitude will be transmitted to the exterior part of
the membrane and the evolution from the point of view of the system
inside the membrane will not be unitary.  This kind of situations have
been studied in a variety of situations (see for instance a detailed
discussion of this type of boundary conditions for interfaces of two
quantum systems by Popov \cite{Po95} and references therein).

%%%%%%%%%%%%%%%%%%%%
%%%%%%%%%%%%%%%%%%%%

\subsection{Unitarization of non-self-adjoint boundary conditions}

Thus, dissipation at the boundary from the previous viewpoint
will be modeled by a non-unitary
isomorphism $F$ at the boundary.  Such isomorphism will replace the exterior system 
that is in contact with the inner one.  

It is important to observe that the full
system, the initial or interior system plus the exterior system, being closed has to
be described by unitary evolution.  In this sense the non-unitarity of the evolution
of the system under the boundary condition $F$ is
restored adding and external system that takes into account the
dissipation at the boundary.  

Thus, given a dissipative quantum system $H$  described by a non-selfadjoint
extension $F$, the idea to restore its unitarity (within a larger system, of course) will be to construct an
enlarged unitary quantum system such that our system will be a
(non-unitary) subsystem.   A natural requirement to ask to
this enlarged quantum system is to be as small as possible, that is minimal in the class of such
enlargements, i.e., the smallest possible unitary quantum system in which we can embed the non-unitary one. 
We will call such enlargement an unitarization of the dissipative system $(H,F)$.
 
At the classical level this idea can be implemented easily by using
symplectic geometry.   Different approaches can be taken that will
eventually lead to the construction of adequate quantum models.  We will
only sketch them here, leaving a detailed discussion of them and their
quantum counterparts to further work.   

Let $M$ be a symplectic manifold with symplectic form $\omega$, for
instance $M$ could be the cotangent bundle $T^*\Omega$ of a mechanical
system with configuration space $\Omega$.   Let us consider now a
vector field $X$ on $M$ representing the dynamical evolution of a
classical physical system which is not necessarily Hamiltonian, i.e.,
such that $\sp L_X \omega \neq 0$.   Let $\alpha_X$ denote the exact
2-form such that $\alpha_X =  \sp L_X \omega$, $\alpha_X = d(i_X\omega)$.  

If we solve now the equation
\begin{equation}\label{beta_alpha} 
\sp L_X \beta = \alpha_X ,
\end{equation}
with the requirement that $\beta$ is a closed 2-form with maximal
dimensional kernel (notice that $\omega$ is a solution with minimal dimensional kernel), we can redefine the structure form $\omega$ as
$$
\sigma = \omega - \beta \, ,
$$
and obviously, $\sp L_X\sigma = 0$.   

Now our vector field $X$ is a Hamiltonian vector field for the closed 2-form $\sigma$ (which is non-unique), however
the form $\sigma$ can fail to be symplectic because its rank can be
strictly lower than that of $\omega$, i.e., $\sigma$ will define a
presymplectic structure on $M$.  We must point it out that the eq.
(\ref{beta_alpha}) has always  a solution which is $\omega$ itself, but
if it were the only one, then the 2-form $\sigma = 0$.

The triple $(M,\sigma, X)$ can be naturally extended to a Hamiltonian
system using the well-known coisotropic embedding theorem  \cite{Go81},
that states that there is a symplectic manifold $(P,\tilde{\sigma})$ and 
a Hamiltonian vector field $\tilde{X}$ on it such that there is a canonical 
embbeding $j\colon M\to P$ such that $j^*\tilde{\sigma} = \sigma$, $\tilde{X}$ is
tangent to the submanifold $j(M) \subset P$ and $\tilde{X}\mid_{j(M)} = j_*X$.

The total space $P$ is a tubular neighborhood of the zero section of
the bundle $\ker \sigma^* \to M$, where $\ker\sigma^*$ is the dual of the subbundle $\ker\sigma \subset TM$. 
The vector field $X$ induces a function on $P$ as follows
$$
P_X (x,\zeta ) = \langle \zeta (x) , X(x) \rangle \, , \qquad \forall x\in
M, \quad \zeta \in \ker\omega^* \, ,
$$ 
and the corresponding Hamiltonian vector field $\tilde{X}$ on $P$ restricts to $X$
on the submanifold $M$. In this picture the minimal extension is obtained
by ``adding'' the dual of the kernel of $\sigma$ to our original space.  
Notice  again that if Eq. (\ref{beta_alpha}) had only one solution $\beta =
\omega$, then $\sigma = 0$ and $P = T^*M$.  The Hamiltonian vector field $\tilde{X}$
becomes the complete lift $X^c$ of $X$ to $T^*M$.  

There is an alternative way to present the previous discussion.  It
consists in considering again a vector field 
$X$ which is not Hamiltonian.  This vector field will
represent the non-unitary evolution semigroup at the quantum level.  The graph of
the vector field defines a submanifold, denoted again by $X$, of $TM$.

If the vector field were Hamiltonian, the submanifold $X$ would be
Lagrangian with respect to the natural symplectic form $\dot\omega$ in
$TM$.  In general we will obtain that $TX\neq TX^\perp$, where
$\perp$ means the symplectic orthogonal with respect to $\omega$.   
The distribution on $X$ defined by
$TX\cap TX^\perp$ (provided that the intersection is clean) is
integrable as it is easily seen by computing $\omega ([U,V],Z)$ for $U,V\in
TX\cap TX^\perp$ and $Z\in TX$. 

Then the quotient $TX/ TX\cap TX^\perp$
is a symplectic bundle.   Denoting by $\sp F$ the foliation defined
by $TX\cap TX^\perp$, if the space of leaves $\mathcal{S}_X = X/\sp F$ of this foliation is a manifold, then $TX/TX\cap
TX^\perp $ can be identified with its tangent bundle $T(X/\sp F)$. Then,
the induced form from $\hat{\omega}$ on $TX/TX\cap TX^\perp $ will define
a non-degenerate, closed, smooth 2-form on $\mathcal{S}_X$
making in this way  $\sp S_X$ into
a symplectic manifold.  

This symplectic manifold $\sp S_X$
measures the ``non-Hamiltonianess'' of the vector field $X$.  
The submanifold $\mathcal{S}_X$ is the classical analogue of 
von Neumann's deficiency spaces $\mathcal{N}_\pm$ for a symmetric
operator.

Inspired by the same idea as von Neumann's theorem, Thm. \ref{vonNeuman}, one way to make $X$
into a Hamiltonian vector field would be to ``remove'' this symplectic
manifold $\sp S_X$ converting it into a Lagrangian submanifold of a bigger
space.   The details of this construction will be discussed elsewhere.

However the two constructions proposed above are not really addressing
the classical analogue of the problem of unitarization of a symmetric operator 
because they are not ``boundary
problems''.   

The non-Hamiltonian character of the vector field $X$ above does
not come from any boundary condition for the classical system.  Such non-Hamiltonian character is 
local in the interior of the manifold $M$ because the Lie derivative appearing on 
Eq. \eqref{beta_alpha} is
defined locally, contrary to what happens with the effect of boundary conditions
in quantum evolution as it was discussed in the previous section.

Boundary conditions in classical Hamiltonian systems will be described as follows.   
Let us consider a classical
mechanical system with configuration space again a smooth Riemannian manifold $\Omega$ with non-empty
boundary $\partial\Omega$.   Now we impose boundary conditions for the
free system on $\Omega$, but contrary to the discussion in the Introduction, Eq. \eqref{CBCs}, by means of a non-canonical map 
$S\colon T^*\partial \Omega \to T^*\partial \Omega$, $S^*\omega_{\partial \Omega} \neq \omega_{\partial \Omega}$,
with $\omega_{\partial \Omega}$ the canonical symplectic structure on $T^*\partial \Omega $.

It is clear that the mechanical effect of such
non-canonical boundary condition is going to be related to dissipation
of volume density of $T^*\Omega$ at the boundary.  Thus we can think
that this volume density is transmitted to a ``mirror space'' or
external space that has been put in contact with the original one
through its boundary.   Thus, the natural way to recover a symplectic
(hence volume-preserving) evolution, would be to double the space by
adding a mirror image of $T^*\Omega$ and pasting the two of them by
means of the boundary condition $S$.  This requires some care because
$S$ is not a map from $\partial (T^*\Omega)$ into itself, but rather a
map between the symplectic boundaries of the two spaces \cite{gift}.  
As in previous discussions we will not pursue the description of the
classical situation leaving it for later developments and we will
concentrate on the quantum situation.

It has become clear from the previous comments at the classical level,
that a good strategy to restore unitarity for non-self-adjoint
extensions of symmetric operators, in particular the Hodge Laplacian,
would be to double our state space using a mirror image of the original
one and then using the boundary conditions to ``paste'' the domains of
the original operator and its mirror image in such a way that the
dissipation introduced by the former will be transmitted to the later \cite{BW95}. 

Analytically we will proceed as follows.   Let us denote as in Section \ref{section:Laplace} by
$H^2 (\Omega )$ the Sobolev Hilbert space defining the maximal extension of the operator
$-\Delta$.   The boundary data space will be denoted as usual by $\sp
H_B$ and an elliptic extension of $-\Delta_0$
will be defined by the subspace $W\in \mathbf{Gr}$.  In particular
we will assume that $W$ is the graph of a non-self-adjoint operator $A\colon \sp L_+ \to
\sp L_-$, $\dot\varphi = A \varphi$.   

We introduce a mirror Hilbert space $H^2 (\Omega)_{\mirror}$ which is
a copy of $H^2(\Omega)$ and, in the direct sum Hilbert space $H^2(\Omega) \oplus H^2
(\Omega)_{\mirror}$, we will define an extended operator $-\Delta_{\ext}$
as follows.   In the domain $H^2 (\Omega) \oplus H^2
(\Omega)_{\mirror}$, the operator $-\Delta_{\ext}$ is the direct sum of
$-\Delta_0 \oplus -\Delta_0$. Thus the operator $-\Delta_{\ext}$ with domain $H_0^2 (\Omega) \oplus H_0^2
(\Omega)_{\mirror}$ is symmetric. 

 If we denote now by
$\Psi_{\ext}\in H^2(\Omega)
\oplus H^2 (\Omega)_{\mirror}$ a vector on the enlarged space, we will denote by $\Psi$ the projection $\pi(\Psi_\mathrm{ext})$ of
$\Psi_{\ext}$ into its first factor and by $\Psi_{\mirror}$ the
projection $\pi_2(\Psi_\mathrm{ext})$ onto the second factor $H^2(\Omega)_{\mirror}$.  Then, given
$\Psi_{\ext}$ we will define the ordinary boundary values
$b(\pi_1(\Psi_{\ext})) = (\varphi, \dot\varphi )$ and the mirror boundary
values $b (\pi_2 (\Psi_{\ext})) = b(\psi_{\mirror}) = (\varphi_{\mirror},
\dot\varphi_{\mirror})$.  

Then, we define the domain of
$-\Delta_{\ext}$ associated to the operator $A$ as the space of functions
$\psi_{\ext}$ such that:

\begin{equation}\label{bound_diss} 
\dot\varphi_{\mirror} = A^\dagger \varphi, \qquad
\dot\varphi = A \varphi_{\mirror} \, .
\end{equation}

We shall denote this subspace as $b^{-1}(W_A)_{\ext}$.

The following computation shows that $-\Delta_{\ext}$ is self-adjoint in
$b^{-1}(W_A)_{\ext}$.

\begin{eqnarray*} \langle -\Delta_{\ext} \Psi_{\ext}, \Psi_{\ext}' \rangle
&=& \langle -\Delta_{\ext} (\Psi, \Psi_{\mirror}), (\Psi',\Psi_{\mirror}')
\rangle \\ &=& \langle (-\Delta\Psi, -\Delta\Psi_{\mirror}), (\Psi',
\Psi_{\mirror}') \rangle  \\ &=&
\langle -\Delta\Psi, \Psi' \rangle + \langle -\Delta\Psi_{\mirror},
\Psi_{\mirror}' \rangle \\ &=&
\langle \Psi, -\Delta\Psi' \rangle + \langle \Psi_{\mirror},
-\Delta\Psi_{\mirror}' \rangle \\ && + \langle \dot\varphi, \varphi'\rangle
- \langle \varphi, \dot\varphi'\rangle +
\langle \dot\varphi_{\mirror}, \varphi_{\mirror}'\rangle -
\langle \varphi_{\mirror}, \dot\varphi_{\mirror}'\rangle .
\end{eqnarray*}
But using the boundary conditions Eq. (\ref{bound_diss}), the last four
terms in the previous equation become,
\begin{eqnarray*} \langle \dot\varphi, \varphi'\rangle 
- \langle \varphi, \dot\varphi'\rangle &+& 
\langle \dot\varphi_{\mirror}, \varphi_{\mirror}'\rangle
- \langle \varphi_{\mirror}, \dot\varphi_{\mirror}'\rangle \\  &=& 
\langle A\varphi_{\mirror}, \varphi'\rangle
- \langle \varphi,A\varphi_{\mirror}'\rangle \\ &+& 
\langle A^\dagger\varphi, \varphi_{\mirror}'\rangle
- \langle \varphi_{\mirror}, A^\dagger\varphi'\rangle = 0.
\end{eqnarray*}

Hence the operator is self-adjoint as claimed.

We have proved the following theorem.

\begin{theorem} Given the dissipative quantum system defined on a Riemannian manifold $\Omega$ with non-empty boundary $\partial \Omega$ by 
the Hamiltonian $H_0 = -\frac{1}{2}\Delta_A$, with $\Delta_A$ the Bochner Laplacian determined by the metric and a connection $A$, and non-self-adjoint elliptic boundary conditions defined by the non-self-adjoint boundary operator $A$, there exists an unitarization of the system on the enlarged Hilbert space $L^2(\Omega)\oplus L^2(\Omega)_\mirror$
determined by the boundary conditions given by Eq. \eqref{bound_diss}. 

\end{theorem}

\newpage

%%%%%%%%%%%%%%%%%%%%%%%%
%%%%%%%%%%%%%%%%%%%%%%%%
%%%%%%%%%%%%%%%%%%%%%%%%

\section{Self-adjoint extensions of elliptic operators with symmetry}\label{section:symmetry}

This section will be devoted to the analysis of the structure and the
global properties of self-adjoint extensions of elliptic operators
invariant with respect to a Lie group of transformations.
As before we will discuss the theory for Dirac operators and the ideas
extend in a natural way to higher order differential elliptic
operators.  

%%%%%%%%%%%%%%%%%%%%%%%%%%

\subsection{Dirac bundles with symmetry}\label{section:Dirac_symmetry}

Let us consider the following geometrical setting.  Let $G$ be a Lie group acting on $\Omega$
smoothly, i.e., there is a smooth map
$\Phi \colon G\times \Omega \to \Omega$ such that $\Phi (e,x) = x$ for
all $x\in \Omega$, $\Phi (h, \Phi (g,x)) = \Phi (hg, x)$, for all
$g,h\in G$, $x\in \Omega$, and $\Phi (g,x) \in \partial \Omega$ for
every $x\in \partial \Omega$.  As usual the action $\Phi (g,x)$ will be
denoted simply by $gx$, and the induced action of $G$ on $\partial \Omega$ will
be denoted with the same symbol.  

The space of orbits of the
action will be denoted by $\Omega/G$ and if the action of $G$ on
$\Omega$ is proper and free the quotient space $\Omega/G$ will be a
smooth manifold with boundary $\partial (\Omega/G) = \partial \Omega /G$.
The Riemannian structure $\eta$ can be chosen to be invariant if
the group is compact.  In fact, in that case we can average an arbitrary Riemannian
structure to obtain an invariant Riemannian structure
on $\Omega$.

The action of the group $G$ lifts naturally to the tangent bundle
$T\Omega$ and the action is given by the tangent maps of the
diffeomorphisms defined by the group elements $g\in G$ and the corresponding
action on the space of vector fields on $\Omega$ will be denoted by $g_*X$, $X\in \mathfrak{X}(\Omega)$.

Let us consider as in the previous sections a Dirac bundle $\pi\colon
S\to \Omega$ such that there exists a lifting of the action of $G$ on
$\Omega$ to the total space $S$ of the bundle, i.e.,
there exists an action map $\Psi\colon G\times S \to S$ such that it
commutes with the natural projection maps, that is, 
$$ 
\pi \circ \Psi = \Phi \circ \pi \, ,
$$
and the action of $g\in G$ on $S$ maps linearly the fibre over $x$ into
the fibre over $gx$.  Thus, the action $\Psi$ preseves the boundary
bundle $S_{\partial \Omega}$ over $\partial \Omega$.   We will assume that 
we can choose the Hermitean structure on the Dirac bundle $S$ to be $G$
invariant, i.e.,
$$ 
(g\xi,g\zeta )_{gx} = (\xi, \zeta )_x\, , \qquad \forall g\in G, \quad \xi,
\zeta \in S_x \, ,
$$
and the group $G$ will be represented unitarily on the bundle $S$,
as well as the Hermitean connection $\nabla$, that is, because the group $G$ acts
in the space of sections $\Gamma (S)$ as $(g\cdot \sigma) (x) = \Psi(g, \sigma(\Phi(g^{-1},x)))$, $x\in \Omega$, $g\in G$, then
$$
\nabla_{g_*X} (g\cdot \sigma) = g\cdot \left( \nabla_X \sigma \right)\,, \qquad \forall \sigma \in \Gamma(S)\, , \quad  X \in \mathfrak{X}(\Omega) \, , \quad  g \in G \, .
$$

As the Riemannian metric $\eta$ is $G$-invariant, the action of the group
lifts to the Clifford algebra bundle $\Cl (\Omega )$ over $\Omega$.   The
action $\rho$ of the Clifford algebra bundle $\Cl (\Omega )$ on the Dirac
bundle $S$ defines a homomorphism of algebra bundles $\rho \colon \Cl
(\Omega) \to \mathrm{End} (S)$, where $\mathrm{End} (S)$ denotes the algebra bundle of
endomorphisms of the vector bundle $S$.  The group $G$ acts on the
vector bundle $S$ by endomorphisms, thus this action extends to the
algebra bundle $\mathrm{End} (S)$ in a natural way, that is, if $h\colon S\to S$ is a
bundle homomorphism, than  $h^g = g^{-1}\circ h \circ g$, $g\in G$. 
Thus we have two $G$-spaces, $\Cl (\Omega)$ and $\mathrm{End} (S)$ and a map $\rho$
between them, then if the group $G$ is compact, we can choose this map to be equivariant by averaging.  In
fact let
$$ \rho_G (u) = \int_G \rho (gu)^{g^{-1}} d\mu_G(g) ,$$
where $d\mu_g$ denotes (the normalized Haar measure on the group $G$. 
Then,
$$ \rho_G (hu) = \int_G \rho (ghu)^{g^{-1}} d\mu_G (g) = 
 \int_G \rho (k u)^{kh^{-1}} d\mu_G (k) = \int_G \rho {(k u)^k}^{h^{-1}}
d\mu_G (k) $$ 
$$= {\left( \int_G \rho {(k u)^k}
d\mu_G (k)\right) }^{h^{-1}} = \rho_G (u)^{h^{-1}} .$$

Finally, if the group $G$ is compact the connection $\nabla$ in $S$ can be chosen to be
equivariant by averaging again a given connection.  It is easy to check that if the given
connection were verifying the derivation property (\ref{derD}), then the
averaged connection will satisfy it again. 

Summarizing the previous discussion, we have arrived at the following
result.

\begin{proposition}\label{equiv_dirac}  Let $S$ be a Dirac bundle over the Riemannian
manifold with boundary $\Omega$ and let $G$ be a compact Lie group
acting on the bundle $S$ by bundle isomorphisms, then there exists a
Dirac bundle structure on $S$ which is $G$-invariant.  Besides
the Dirac operator constructed using it will commute with the
action of $G$ and is topologically equivalent to the initial one.
\end{proposition}

Proof.  It is immediate from the previous considerations and the fact
that the space of connections and metrics is contractible, then there
exists a continuous path $\Dsl_t$ connecting the Dirac operator $\Dsl$ and
the averaged one $\Dsl_{inv}$. \hfill$\Box$

%%%%%%%%%%%%%%%%%%%%%%%%%%%%%

\subsection{The quotient Dirac operator}

Under the conditions stated in Prop. \ref{equiv_dirac} we have constructed an equivariant Dirac operator, that we simply denote again by $\Dsl$.  If the action of the group
$G$ on $S$ is ``good enough'', e.g., proper and free, then the quotient
total space $S/G$ will be a smooth bundle over the quotient manifold
$\Omega/G$ with smooth boundary
$\partial \Omega /G$.   Moreover the structures on $S$ will be related to the corresponding structures on the
quotient and the bundle $S/G \to \Omega/G$ will be again a Dirac bundle with Dirac
operator $\Dsl_G$.  

We will denote by $\pi_G$ the projection map between the quotient spaces $S/G$ and $\Omega/G$ above defined as:
$\pi/G ([\xi]) = [\pi(\xi)]$, where $[\xi]$ denotes the orbit of $\xi\in S$ under the action of $G$ and, similarly, $[x]$ is the orbit of $x$ in $\Omega$.  
The Dirac bundle $\pi_G \colon S/G \to \Omega/G$ will be  called the quotient Dirac bundle.  

The space of sections $\Gamma (S/G)$ of the quotient
Dirac bundle $S/G$ are in one-to-one correspondence with the space of equivariant sections of $S$ under the action of $G$, that is: $\Gamma(S/G) = \Gamma(S)^G$, with 
$$
\Gamma(S)^G = \{ \sigma\in \Gamma(S) \mid g\cdot \sigma = \sigma \} \, .
$$
Notice that if $\sigma \in \Gamma(S)^G$ we may define $\tilde{\sigma}([x]) = \sigma (x)$ for all $[x] \in \Omega/G$.
then $\tilde{\sigma}$ defines a section of $E/G$ as it satisfies that $[\sigma([x])] = [\sigma([x'])]$ whenever $[x] = [x']$.
On the other hand, if $\tilde{\sigma} \colon \Omega/G \to S/G$ is a section, then we may define $\sigma (x) = \int_G g^{-1}(\tilde{\sigma}(\pi(x))) d\mu_G(g)$, $x \in \Omega$ which is an invariant section of $S$.

Because of the $G$-equivariance of the Dirac operator
$\Dsl$ on $S$, it will induce an operator on the quotient space that we
will denote by $\Dsl/G$.    Actually, notice that if $\sigma \in \Gamma(S)^G$,
then 
\begin{equation}\label{equiv_sections}
\Dsl g\cdot \sigma = e_i \cdot \nabla_{e_i} g\cdot \sigma = g\cdot \Dsl \sigma
\end{equation}
and then it makes sense to define 
$$
\Dsl_G [\sigma ] = \Dsl ([\sigma]) \, , \qquad \forall [\sigma] \in \Gamma \, .
$$

\begin{remark}
It is clear that under the previous conditions
for the action of $G$ on $S$, then $\Dsl/G = \Dsl_G$.   We shall remark here
that even if the quotient space $S/G$ fails to be a smooth bundle over a
smooth manifold $\Omega/G$, this will happen for instance if the action of $G$ is not
free, the induced operator $\Dsl/G$ will still be defined as the following discussion shows. 
\end{remark}

The group $G$ acts naturally on the space of smooth sections of $S$, as indicated above, i.e.,
\begin{equation}\label{uni_bulk}
(g\cdot \xi) (x) = g(\xi (g^{-1}x))\, ,  \qquad \forall g\in G, x\in \Omega \, .
\end{equation}
By continuity this action extends unitarily to the spaces of sections
$H_0^1(\Omega,S)$, $H^1(\Omega,S)$ and $L^2(\Omega,S)$ because:
\begin{eqnarray*} 
\langle g\cdot \xi, g\cdot \zeta \rangle &=& \int_\Omega (g\cdot \xi
(x), g\cdot \zeta (x) )_x \vol_\eta (x) \\
&=& \int_\Omega (\xi (g^{-1}x),
\zeta (g^{-1}x) )_x \vol_\eta (x) \\
&=& \int_\Omega (\xi (x), \zeta (x) )_x
\vol_\eta (x) = \langle \xi, \zeta \rangle \, ,
\end{eqnarray*}
where we have used that $G$ acts by isometries of $\eta$, then
it preserves the volume form $\vol_\eta$, hence the Jacobian of the
diffeomorphism $g$ will be trivial.  Besides $g$ acts by unitary
transformations on the hermitian bundle $S$.  Thus, the Hilbert spaces
of sections before support unitary representations of the group $G$. 

Similarly, the boundary data Hilbert space $\sp H_D$ will also define a
unitary representation of the group $G$.  The boundary map $b$ is
equivariant because the pull-back map $i\colon \partial \Omega \to
\Omega$ commutes with the action of $G$ and the following diagramme is
commutative
\begin{equation}\label{Dirac_traceable}
\begin{array}{ccccc} H^1(\Omega,S) & & \begin{array}{cc}V(g) \\ \to \\ \phantom{} \end{array} & & H^1(\Omega,S) \\
\downarrow & & & & \downarrow \\ H^{1/2}(\pO,S_\pO) && \begin{array}{cc}\phantom{} \\ \to \\ v(g) \end{array} && H^{1/2}(\pO,S_\pO)
\end{array} \, ,
\end{equation}
where $V(g) \colon H^1(\Omega,S) \to H^1(\Omega,S)$ denotes the unitary representation of the group $G$ defined by Eq. \eqref{uni_bulk} and $v(g)\colon H^{1/2}(\pO,S_\pO) \to H^{1/2}(\pO,S_\pO)$ is the corresponding 
unitary representation induced in the restriction of the Dirac bundle $S$ to the boundary (see next section, Sect. \ref{subsec:6-2}, for more details).

We will not pursue this analysis here, but it is obvious that the
multiplicities of the irreducible representations of $G$ contained in
the Hilbert space $H^1 (S)$ will be related to the multiplicities of
the corresponding ones in the boundary Hilbert space $\sp H_D$.  More
elaborate comments on this will be done later. 

Now we will define the quotient operator $\Dsl/G$.
Let $\Gamma_0 (S)^G$ and $\Gamma (S)^G$ be the subspaces of smooth invariant
sections of $S$ of $\Gamma_0(S)$ and $\Gamma (S)$, the spaces of compact supported smooth sections of $S$ and
smooth sections of $S$, respectively.  Clearly, this amounts to $\xi$ be a fixed point for the action of
$G$ on $\Gamma (S)$.  

Then we will define $\Dsl/G$ as a linear map
$\Gamma_0 (S)^G \to \Gamma_0 (S)^G$ by means of
$$ 
(\Dsl/G) (\xi ) = \Dsl (\xi ) \, , \qquad  \forall \xi \in \Gamma_0^G (S) \, .
$$
Notice that because $\Dsl$ is a differential operator, 
hence local, and  Eq. \eqref{equiv_sections}, then $(\Dsl/G)(\xi) \in \Gamma_0 (S)^G$.

Clearly the operator $\Dsl/G$ is symmetric on the domain $\Gamma_0 (S)^G$ and
we can search for its self-adjoint extensions.  Of course, $\Gamma_0
(S)^G$ is not dense in $H^1(\Omega,S)$ but it is dense in the intersection of the $L^2$-closure of
$\Gamma (S)^G$ and $H^1 (S)$.  Such space, denoted in what follows by
$H^1 (S)^G$, will play the role of the Hilbert space of sections of
Sobolev class 1 in the quotient bundle space $S/G$.  Notice that $H^1
(S)^G$ coincides with the subspace of fixed sections under
the action of $G$ in $H^1 (S)$.  Thus we have,

\begin{proposition}  If we denote by $\mathrm{Fix}_G (H^1 (S))$ the fixed set of
the unitary action $V(g)$ of $G$ in $H^1(\Omega,S)$, then, 
$$ 
\mathrm{Fix}_G (H^1 (S)) = H^1 (S)^G \, ,
$$
and similarly for $H_0^1(\Omega,S)$.  Moreover, if the quotient spaces $S/G$,
$\Omega/G$ are smooth manifolds and the canonical projection is a smooth
submersions, then 
$$ 
H^1 (S)^G  \cong H^1 (S/G) \, ,
$$
where the later identification means that there is a natural
unitary transformation from the Hilbert space $H^1 (S)^G$ and the Hilbert space of order 1 Sobolev 
sections of the bundle $S/G \to \Omega /G$. 
\end{proposition}

Once the quotient operator $\Dsl/G$ has been defined and it has been shown
to be symmetric in a dense domain of the Hilbert space $H^1(\Omega,S)^G$ we
can, as we did for the Dirac operator $\Dsl$, compute and characterize all
its selfadjoint extensions.   Of course, as it was discussed before, if the action of $G$ defines a
quotient Dirac bundle $S/G \to \Omega/G$, 
the space $H^1(\Omega,S)^G$ is precisely the space $H^1(S/G)$ of sections of the
quotient bundle, and the quotient operator $\Dsl/G$ is precisely the Dirac
operator $\Dsl_G$ in the quotient bundle, thus using the results and the
discussion in Section \ref{Cayley_D}, Thm. \ref{sa_Dirac}, its self-adjoint extensions are given by the
self-adjoint  Grassmannian $\sp M (\Dsl_G)$ on the boundary Hilbert space
$\sp H_{\dsl_G} = H^{1/2}(\pO,S_\pO)$ defined on the boundary $\partial \Omega /G$ and the
problem will be solved.  

In spite of this, we would like to
characterize such self-adjoint extensions in terms of self-adjoint
extensions for $\Dsl$ on $\Omega$, i.e., we are asking how to obtain $\sp
M(\Dsl_G)$ directly from $\sp M (\Dsl)$.   Apart from the intrinsic interest
of being able to compute things in quotient spaces without having to go
to the quotient, avoiding the inherent difficulties of taking
quotients, this approach to the problem has the advantage of
providing an effective method to construct the self-adjoint extensions of the
quotient Dirac operator $\Dsl/G$ when $S/G$ is not a manifold, a situation
which is often found in all sort of problems.  

\begin{figure}[h]
\centering
\includegraphics[width=12cm]{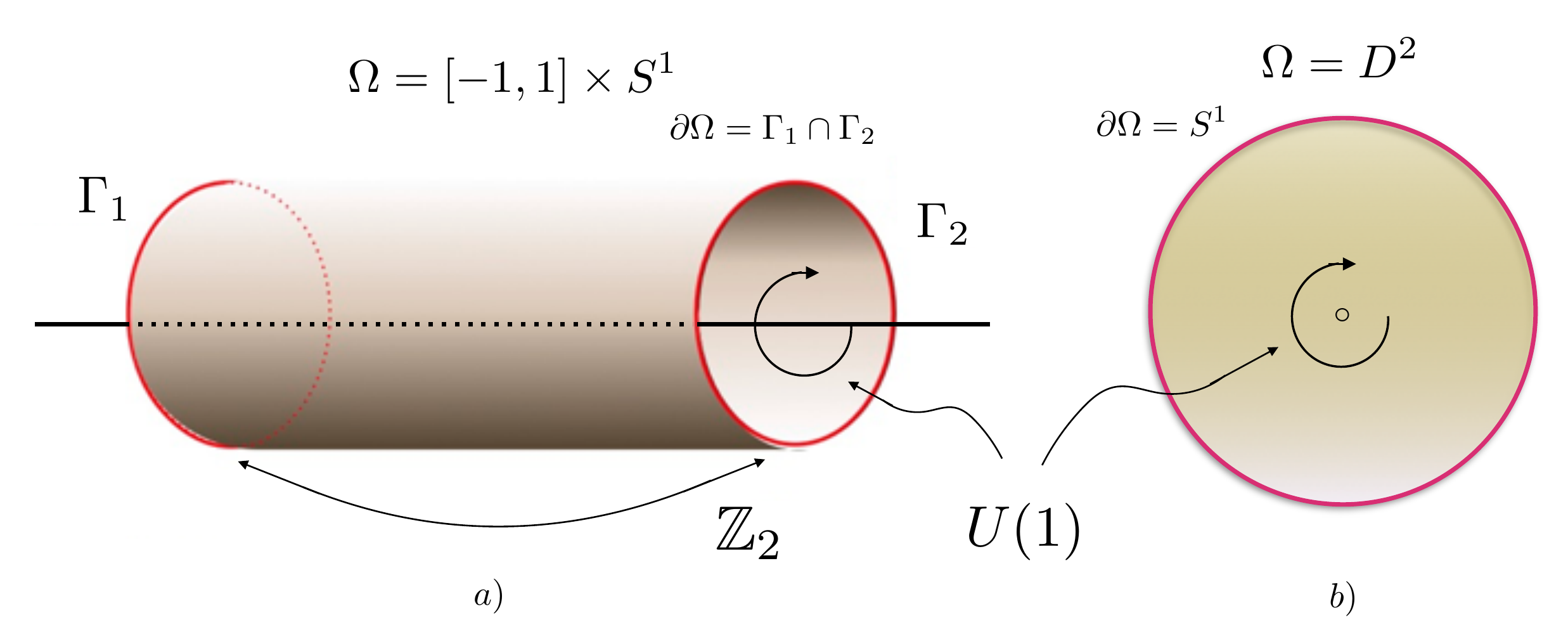}
\caption{The cylinder (a) and disk (b) with the groups $U(1)$ and $\mathbb{Z}_2$ acting on them.}\label{fig:cilindro}
\end{figure}

For instance, consider the following two simple examples.  Let $\Omega =
S^1 \times [0,1]$ be the cylinder with boundary $\partial \Omega =
S^1\times \set{0} \cup S^1 \times \set{1}$ and consider the natural
action of the group $U(1)$ on $\Omega$ by rotations along the symmetry axis of the
cylinder (see Fig. \ref{fig:cilindro}).  
Then the quotient space $\Omega/U(1)$ is clearly the smooth manifold with boundary
$[0,1]$, and we could expect that the $U(1)$-invariant self-adjoint
extensions of a Dirac operator defined on a given Dirac bundle over
$\Omega$ will correspond to the self-adjoint extensions of the Dirac
operator defined by the projection of that bundle to $[0,1]$.  
It is easy to check that if $S$ is a complex line bundle (the Spin bundle of the manifold), then
$\mathcal{M}(\Dsl_{U(1)}) \cong U(1)$.   

Consider now, instead of the cylinder the unit disk. 
That is, $\Omega = \set{z\in \C \mid |z| \leq 1}$ with boundary
$\partial \omega = S^1$.  Consider now the natural action of $U(1)$ on
$\Omega$ by complex multiplication, i.e., rotation around the origin in
$\C$.  The action of $U(1)$ is not free, however the quotient space is
a smooth manifold with boundary, $[0,1]$ again.  What happens to the
self-adjoint extensions of the quotient operator now?  Do we obtain
all self-adjoint extensions of $\Dsl_{U(1)}$ by looking at the
$U(1)$-invariant ones in $\Omega$?  Clearly no, new self-adjoint
extensions on the quotient space arise because of the non-trivial
nature of the action of the group z(the origin is a fixed point for the action).  

We will analye in what follows
these matters.   We will restate some of the notions
introduced above concerning Dirac operators in a slightly more general context,
and we will  proceed then to a direct construction of the
`quotient' self-adjoint  Grassmannian.   

%%%%%%%%%%%%%%%%%%%%%%%%%%%%
%%%%%%%%%%%%%%%%%%%%%%%%%%%%

\subsection{Unitaries at the boundary and $G$-invariance}\label{subsec:6-2}

Let us consider now a Hermitean bundle $\pi \colon E \to \Omega$, with $(\Omega, \eta)$ a Riemannian manifold
with smooth boundary $\partial \Omega$.

Let $G$ be a Lie group and $V\colon G \to \mathcal{U}(L^2(\Omega,E))$ be a continuous unitary representation of $G$, 
on the Hilbert space of square integrable sections of $E\to \Omega$, 
i.e., for any $\Phi\in L^2(\Omega,E)$ the map
\[
 G\ni g\mapsto V(g)\Phi \quad 
\]
is continuous in the $L^2$-norm $|| \Phi ||^2 = \int_{\Omega} || \Phi (x) ||^2 \vol_\eta$.

Notice that if the unitary representation $V$ leaves invariant the subspace $H^1(\Omega,E)$, i.e., 
$V(g) H^1(\Omega) \subset H^1(\Omega)$, and it leaves invariant the quadratic form $Q(\Phi) = || \nabla\Phi ||^2$,
where $\nabla$ is a Hermitean connection on $E$ and $Q$ is defined on Neumann's domain $H^1(\Omega,E)$, (we call $Q$ with such domain the Neumann's quadratic form), that is 
$$
Q(V(g) \Phi ) = Q(\Phi) \, \qquad \forall g \in G, \, \Phi \in H^1(\Omega,E) \, ,
$$
then $V$ defines also a continuous unitary representation on $\mathcal{H}^{1}(\Omega)$ with its corresponding
Sobolev scalar product (see for instance \cite{Ib14c}).  Now we can extend the property of equivariant Dirac operators expressed by Eq. \eqref{Dirac_traceable}, that can be used also to study self-adjoint extensions of Laplace operators.

\begin{definition} \label{def:traceable}
The representation $V\colon G\to L^2(\Omega,E)$ has a trace (or is traceable) along the boundary
$\partial \Omega$, if it leaves invariant Neumann's quadratic form $Q$ and there exists another continuous, unitary representation 
$v \colon G\to \mathcal{U}( L^2(E_{\partial \Omega}))$ 
such that 
\begin{equation}\label{VPhi=vphi}
b ( V(g) \Phi ) = v(g)  \gamma (\Phi)  \, ,
\end{equation} 
for all $\Phi \in H^1(\Omega,E)$ and $g\in G$ or, in other words, that the following diagram is commutative:
$$
\begin{array}{ccc}   H^1(\Omega,E) & \overset{V(g)}{\longrightarrow} & H^1(\Omega,E) \\
b \downarrow & & \downarrow b \\
 H^{1/2}(\pO,E_\pO) & \overset{v(g)}{\longrightarrow}& H^{1/2}(\pO,E_\pO)
\end{array}
$$
We will call $v$ the trace of the representation $V$.
\end{definition}

Notice that if the representation $V$ is traceable, its trace $v$ is unique.  
It is not difficult to prove the following theorem: (see the proof in the case of
Laplace operators in \cite{Ib14c}).

\begin{theorem}\label{repcommutation}
Let $G$ be a Lie group, $\pi\colon E \to \Omega$ a Dirac bundle over $\Omega$ and $\Dsl$ a Dirac operator on it.  Let $V\colon G\to \mathcal{U}(L^2(\Omega,E))$ be a traceable continuous, unitary representation of $G$ with unitary trace
$v\colon G\to\mathcal{U}(L^2(\pO,E_\pO))$ along the boundary $\partial\Omega$. Denote by $(\Dsl_U,\D_U)$ the 
self-adjoint extension of the Dirac operator $\Dsl$ determined by the unitary operator $U \colon H^{1/2}(\pO,E_\pO) \to H^{1/2}(\pO,E_\pO)$. Then we have that $[v (g)\,,\,U]=0$ for all $g\in G$ iff $\Dsl_U$ is $G$-invariant.
\end{theorem}

%%%%%%%%%%%%%%%%%%%%%%%%%%%%%%%%%%%%%%%%%%%%%%%%%%%%%%%%%%%%%%%%%%%%%%%%%%%%%%%%%%%%%%%%%%%%%%%%%%%%%%%%%%%%%%%%%%%%%%%%%%%%%%%%%%%%%%%%%%%%%%
\subsection{Examples:  Groups acting by isometries}

In this section we will discuss some examples with unitaries which satisfy the conditions mentioned in the statements above derived by actions of groups by unitary transformations of the bundle $E$ covering isometries on $\Omega$.

Thus, assume as in Section \ref{section:Dirac_symmetry}, that the group $G$ acts smoothly by isometries on the Riemannian manifold $(\Omega,\pO,\eta)$ and this action can be lifted to an unitary action on the bundle $\pi\colon E \to \Omega$.  
Any $g\in G$ specifies a bundle isomorphism $g_E\colon E \to E$ such that $(g_E\cdot \xi , g_E\cdot \zeta) = (\xi, \zeta)$, for all $\xi, \zeta \in E$ and $( \cdot, \cdot )$ denotes the inner product along the fibers of $E$.  The element $g\in G$ defines also a diffeomorphism $g_\Omega \colon \Omega\to\Omega$ such that $\pi(g_E\cdot \xi) = g_\Omega \cdot \pi(\xi)$.  Moreover, we have that $g_\Omega^*\eta=\eta$,
where $g_\Omega^*$ stands for the pull-back by the diffeomorphism $g$. 
These diffeomorphisms restrict to isometric diffeomorphisms on the Riemannian manifold at the boundary $(\pO,\partial\eta)$ (see, e.g., \cite[Lemma 8.2.4]{marsden01}), 
$$
(g|_{_{\pO}})^*\partial\eta=\partial\eta\; ,
$$ 
hence the action on $E$ restricts to an action on $E_\pO$, the pull-back of the bundle $E$ to the boundary $\pO$.
These actions of the group $G$ induce unitary representations of the group on  the space of square integrable sections of the bundles $E\to \Omega$ and $E_\pO \to \pO$\,.
In fact, consider the following representations:

\begin{eqnarray*}
&& V\colon G\to\mathcal{U}(L^2(\Omega, E))\;, \qquad V(g)\Phi = g_E\cdot \Phi\circ g_\Omega^{-1}\, , \Phi\in L^2(\Omega, E)\, , \\
&& v\colon G\to\mathcal{U}(L^2(\pO,E_\pO))\;, \quad
v (g)\varphi = g_E|_{_{E_\pO}}\cdot \varphi \circ (g|_{_{\pO}})^{-1}\, , \varphi \in L^2(\pO, E_\pO)\, .
\end{eqnarray*}
Then a simple computation shows that,
$$
\scalar{V(g^{-1})\Phi}{V(g^{-1})\Psi} =\scalar{\Phi}{\Psi}\;,
$$
where we have used the change of variables formula and the fact that isometric diffeomorphisms preserve the Riemannian volume, 
i.e., $g^*\negthinspace\d\mu_\eta=\d\mu_{\eta}$\,. The result for the boundary is proved similarly. 
The induced actions are related with the boundary map as in Eq. \eqref{def:traceable}, $(V(g)\Phi)=v (g)b(\Phi)$, $g\in G$, $\Phi\in H^1(\Omega)$, and therefore the unitary representation $V$ is traceable along the boundary of $\Omega$ with trace $v$.

Moreover we have that Neumann's quadratic form $Q$ is $G$-invariant.

\begin{proposition}\label{prop:dphiinvariant}
Let $G$ be a Lie group that acts by unitary bundle isomorphisms on the Hermitean bundle $\pi \colon E \to \Omega$ over the Riemannian manifold with boundary $(\Omega,\pO,\eta)$ and let $V \colon G\to \mathcal{U}(L^2(\Omega,E))$ be the associated unitary representation. 
Then, Neumann's quadratic form $Q_N (\Phi) = \scalar{\nabla\Phi}{\nabla\Phi}$ with domain $H^1(\Omega, E)$ is $G$-invariant, where $\nabla$ is a $G$-invariant connection.
\end{proposition}

\begin{proof}
Let us consider the simpler case of a trivial line bundle over $\Omega$ and trivial unitary action of $G$ along the fibres.
Then the connection $\nabla$ is trivial. The general case 
is a trivial extension.

First notice that the pull-back of a diffeomorphism commutes with the action of the exterior differential. 
Then we have that $$\d(V(g^{-1})\Phi)=\d(g^*\Phi)=g^*\negthinspace\d\Phi\;.$$ 
Hence
\begin{subequations}\label{dphiinvariant}
\begin{eqnarray}
\scalar{\d(V(g^{-1})\Phi)}{\d(V(g^{-1})\Psi)}&=& \int_\Omega \eta^{-1}(g^*\negthinspace\d\Phi,g^*\negthinspace\d\Psi)\d\mu_\eta \nonumber\\
&=& \int_\Omega g^*\negthickspace\left( \eta^{-1}(\d\Phi,\d\Psi) \right)g^*\negthinspace\d\mu_\eta \nonumber\\
&=& \int_{g\Omega}\eta^{-1}(\d\Phi,\d\Psi)\d\mu_\eta \nonumber\\
&=& \scalar{\d\Phi}{\d\Psi}\;,
\end{eqnarray}
\end{subequations}
where in the second inequality we have used that $g:\Omega\to\Omega$ is an isometry and therefore 
$$\eta^{-1}(g^*\negthinspace\d\Phi,g^*\negthinspace\d\Psi)=g^*\negthinspace \eta^{-1}(g^*\negthinspace\d\Phi,g^*\negthinspace\d\Psi)
=g^*\negthickspace\left( \eta^{-1}(\d\Phi,\d\Psi) \right)\;.$$ 
The equations \eqref{dphiinvariant} guaranty also that $V(g) H^1(\Omega)= H^1(\Omega)$ since $V(g)$
is a unitary operator in $L^2(\Omega)$ and the norm 
$\sqrt{\norm{\operatorname{d}\cdot\;}+\norm{\cdot}^2}$ is equivalent to the Sobolev norm of order 1.
\end{proof}

\begin{remark}
Before making explicit the previous structures in concrete examples we notice that the previous discussion 
works in a similar way with the covariant Laplacian $\Delta_A$ discussed in Section \ref{section:Laplace}. Thus if we are given a group acting by unitary bundle isomorphisms on an Hermitean bundle $E \to \Omega$ (and by isometric diffeomorphisms on the Riemannian manifold $\Omega$), then any unitary operator $U$ at the boundary, 
(that in addition satisfies the conditions of possessing gap and being admissible, \cite{Ib14c}, that guarantee that the quadratic form constructed from the operator $\nabla$ with boundary conditions dictated by $U$, read more about self-adjoint extensions determined by quadratic forms in  \cite{Ib14} and \cite{Ib15} this volume), and that verifies the commutation relations of Theorem~\ref{repcommutation} describes a $G$-invariant quadratic form. The closure of this quadratic form characterizes uniquely a  $G$-invariant self-adjoint extension of the Laplace-Beltrami operator.
\end{remark}

%%%%%%%%%%%%%%%%%%%%%%%%%%%%%%%%%%%%%%%%%%%%%%%%%%%%%%%%%%%%%%%%%%%%%%%%%%%%%%%%%%%%%%%%%%%%%%%%%%%%%%%%%%%%%%%%%%%%%%

\begin{example} Discrete and compact groups of isometries\label{sec:SymmetryExamples}

We will discuss now two particular examples of $G$-invariant quadratic forms. In the first example we are considering a 
situation where the symmetry group is a finite, discrete group. In the second one we consider $G$ to be a compact Lie group.

\begin{enumerate}

\item
Let $\Omega$ be the cylinder $[-1,1]\times[-1,1]/\negthickspace\sim$\,, where $\sim$ is the equivalence relation $(x,1)\sim(x,-1)$\,. 
The boundary $\pO$ is the disjoint union of the two circles $\Gamma_1=\left\{\{-1\}\times[-1,1]/\negthickspace\sim\right\}$ and
$\Gamma_2=\left\{\{1\}\times[-1,1]/\negthickspace\sim\right\}$\,, (see Figure \ref{fig:cilindro} (a)). Let $\eta$ be the euclidean metric on $\Omega$. 
Now let $G=\mathbb{Z}_2$=\{e,f\} be the discrete, abelian group of two elements and consider the following action in $\Omega$:
\begin{eqnarray*}
e:(x,y)& &\to(x,y)\;,\\
f:(x,y)& &\to(-x,y)\;.
\end{eqnarray*}
The induced action at the boundary is
\begin{eqnarray*}
e &:& (\pm1,y)\to(\pm1,y)\;,\\
f &:& (\pm1,y)\to(\mp1,y)\;.
\end{eqnarray*}
Clearly $G$ transforms $\Omega$ onto itself and preserves the boundary. Moreover, it is easy to check that $f^*\eta=\eta$\,.

Since the boundary $\pO$ consists of two disjoints manifolds $\Gamma_1$ and $\Gamma_2$\,, the Hilbert space of 
the boundary is $L^2(\pO, E)=L^2(\Gamma_1, E)\oplus L^2(\Gamma_2, E)$. Any $\Phi\in L^2(\pO, E)$
can be written as
$$
\Phi=\begin{pmatrix}\Phi_1(y)\\ \Phi_2(y)\end{pmatrix}
$$ 
with $\Phi_i\in L^2(\Gamma_i, E)$\,. 
A nontrivial action on $L^2(\pO)$ is given by 
\[
v (f)\begin{pmatrix}\Phi_1(y)\\ \Phi_2(y)\end{pmatrix}
=\begin{pmatrix}0 & v \\ v^\dagger & 0\end{pmatrix}\begin{pmatrix}\Phi_1(y)\\ \Phi_2(y)\end{pmatrix}\; ,
\]
where $v\colon L^2(\Gamma_2, E) \to L^2(\Gamma_1, E)$ is a unitary operator.
The set of unitary operators that describe the closable quadratic forms as defined in the previous section 
is given by suitable unitary operators 
$$U=\begin{pmatrix} U_{11} & U_{12} \\ U_{21} & U_{22} \end{pmatrix}\;,$$ 
with $U_{ij}=L^2(\Gamma_j)\to L^2(\Gamma_i)$. 
According to Theorem~\ref{repcommutation} the unitary operators commuting 
with $v (f)$ will lead to $G$-invariant quadratic forms. Imposing $[v(f) , U] = 0$,
we get the conditions
\begin{eqnarray*}
 U_{12} &=& v U_{12} v\;,\\
 U_{22} &=& v^\dagger U_{11} v \;.
\end{eqnarray*}

Obviously there is a wide class of unitary operators, i.e., boundary conditions, that will be compatible with the symmetry group $G$. 
We will consider next two particular classes of boundary conditions. 
First, consider the following unitary operators 
\begin{equation}
U=\begin{bmatrix} e^{\mathrm{i}\beta_1}\mathbb{I}_1 & 0\\ 0 & e^{\mathrm{i}\beta_2}\mathbb{I}_2 \end{bmatrix}\;,
\end{equation}
where $\beta_i\in C^\infty\left(S^1,[-\pi+\delta,\pi-\delta]\cup\{\pi\}\right)$ for some $\delta>0$. 
Moreover, this choice of unitary matrices corresponds to select Robin 
boundary conditions of the form:
\begin{equation}
\left. b\left(-\frac{\d\Phi}{\d x}\right)\right|_{\Gamma_1}=-\tan(\beta_1/2)b(\Phi)\mid_{\Gamma_1}\, ; 
\left.b\left(\frac{\d\Phi}{\d x}\right)\right|_{\Gamma_2}=-\tan(\beta_2/2) b(\Phi)\mid_{\Gamma_2}\;.
\end{equation}
The $G$-invariance condition above imposes $\beta_1=\beta_2$. Notice that when $\beta_1\neq\beta_2$ 
we can obtain meaningful self-adjoint extensions of the Laplace-Beltrami operator that, however, will not be $G$-invariant.

We can also consider unitary operators of the form
\begin{equation}%\label{Eqquasiperiodic}
U=\begin{bmatrix} 0 & e^{\mathrm{i}\alpha} \\ e^{-\mathrm{i}\alpha} & 0 \end{bmatrix}\;,
\end{equation}
where $\alpha\in C^\infty (S^1,[0,2\pi])$. In this case the unitary matrix corresponds to select so-called quasi-periodic boundary conditions, 
cf., \cite{asorey83}, i.e.,
\begin{equation*}
b(\Phi)\mid_{\Gamma_1}=e^{i\alpha}b(\Phi)\mid_{\Gamma_2}\; , \quad  \left.b\left(-\frac{\d\Phi}{\d x}\right)\right|_{\Gamma_1}=e^{i\alpha} \left.b\left(\frac{\d\Phi}{\d x}\right)\right|_{\Gamma_2}\;.
\end{equation*}
The $G$-invariance condition imposes $e^{i\alpha}=e^{-i\alpha}$ and therefore among all the quasi-periodic conditions only the
periodic ones, $\alpha\equiv0$\,, are compatible with the $G$-invariance. 

\item 
Let $\Omega$ be the unit, upper hemisphere centered at the origin. 
Its boundary $\pO$ is the unit circle on the $z=0$ plane. 
Let $\eta$ be the induced Riemannian metric from the euclidean metric in $\mathbb{R}^3$\,. Consider the compact Lie group 
$G=SO(2)$ that acts by rotation around the $z$-axis. 
If we use polar coordinates on the horizontal plane, then the boundary $\pO$ is isomorphic to the interval $[0,2\pi]$ with 
the two endpoints identified. 
We denote by $\theta$ the coordinate parameterizing the boundary and the boundary Hilbert space is $L^2(S^1)$\,.

Let $\varphi\in H^{1/2}(\pO)$ and consider the action on the boundary by a group element 
$g_\alpha\in G$, $\alpha\in [0,2\pi]$, given by
$$v (g^{-1}_\alpha)\varphi(\theta)=\varphi(\theta+\alpha)\;.$$
To analyze what are the possible unitary operators that lead to $G$-invariant quadratic forms it is convenient to use the Fourier 
series expansions of the elements in $L^2(\pO)$\,. Let $\varphi\in L^2(\pO)$\,, then
$$\varphi(\theta)=\sum_{n\in\mathbb{Z}}\hat{\varphi}_ne^{\mathrm{i}n\theta}\;,$$
where the coefficients of the expansion are given by $$\hat{\varphi}_n=\frac{1}{2\pi}\int_0^{2\pi}\varphi(\theta)e^{-\mathrm{i}n\theta}\d\theta\;.$$
We can therefore consider the induced action of the group $G$ as a unitary operator on ${\ell}_2$\,, the Hilbert space of square summable sequences. 
In fact we have that:
\begin{eqnarray*}
\widehat{(v (g^{-1}_\alpha)\varphi)}_n&
= & \frac{1}{2\pi}\int_0^{2\pi}\varphi(\theta+\alpha)e^{-\mathrm{i}n\theta}\d\theta\\
&= & \sum_{m\in\mathbb{Z}}\hat{\varphi}_me^{\mathrm{i}m\alpha}\int_0^{2\pi}\frac{e^{\mathrm{i}(m-n)\theta}}{2\pi}\d\theta
 =e^{\mathrm{i}n\alpha}\hat{\varphi}_n\;.
\end{eqnarray*}
This shows that the induced action of the group $G$ is a unitary operator in $\mathcal{U}(\ell_2)$ that acts 
diagonally in the Fourier series expansion.
More concretely, we can represent it as $\widehat{v (g^{-1}_\alpha)}_{nm}=e^{\mathrm{i}n\alpha}\delta_{nm}\,$\,.
From all the possible unitary operators acting on the Hilbert space of the boundary, only those whose representation in $\ell_2$
commutes with the above operator will lead to $G$-invariant quadratic forms (cf., Theorem~\ref{repcommutation}). 
Since $\widehat{v (g^{-1}_\alpha)}$ acts
diagonally on $\ell_2$ it is clear that only operators of the form $\hat{U}_{nm}=e^{\mathrm{i}\beta_n}\delta_{nm}$\,, 
$\{\beta_n\}_n\subset\mathbb{R}$\,, will lead to $G$-invariant quadratic forms.

As a particular case we can consider that all the parameters are equal, i.e., $\beta_n=\beta$, $n\in\mathbb{Z}$\,.
In this case it is clear that $(\widehat{U\varphi})_n=e^{\mathrm{i}\beta}\varphi_n$\,, which gives the following 
admissible unitary with spectral gap at $-1$:
$$U\varphi=e^{\mathrm{i}\beta}\varphi\;.$$
This shows that the unique Robin boundary conditions compatible with the $SO(2)$-invariance are those that are defined with a 
constant parameter along the boundary, i.e., 
\begin{equation}
b\left(\frac{\d\Phi}{\d \nu}\right)=-\tan(\beta/2)b(\Phi)\;,\quad\beta\in[0,2\pi]\;,
\end{equation}
where $\nu$ stands for normal vector field pointing outwards to the boundary.
\end{enumerate}
\end{example}

%%%%%%%%%%%%%%%%%%%%%%%%%%%%
%%%%%%%%%%%%%%%%%%%%%%%%%%%%

\subsection{Reduction of symplectic manifolds by fixed sets and the reduced elliptic Grassmannian}

This construction of the reduced self-adjoint Grassmannian, i.e., the space of elliptic self-adjoint extensions of the quotient Dirac operator, will be inspired by a natural
construction in symplectic geometry that we will discuss first.

Let $M$ be a symplectic manifold with symplectic form $\omega$ and $G$ a
compact group acting by symplectomorphisms on $M$.  There is a well
developed and very successful way of removing the symmetry degrees of
freedom of $M$ under the symmetry group $G$ known as symplectic
reduction (more specifically Marsden-Weinstein reduction in this particular case, see for instance
\cite[Ch. 7.4]{Ca14}  and references therein).  

However this scheme is
not appropriate for the situation we are considering.  In fact, what we
need is another reduction scheme which is based on the properties of
the singular part of $M$ with respect to the action of $G$ (see Section \ref{section:breaking} for more details
on the stratified structure of $M$).   As usual we will
denote by $G_x$ the isotropy group of $x\in M$.  Then the orbit $G\cdot
x$ through $x$ is diffeomorphic to the homogeneous space $G/G_x$.  Two
points $x,y\in M$ are said to lie in the same stratum $\Sigma$ if the
isotropy groups $G_x$ and $G_y$ are conjugate in $G$.  Thus, the points
in the same orbit are in the same stratum, but other points in different
orbits can be in the same stratum too.  In fact, it is easy to see that
the strata of $M$ are $G$-invariant and thus they are union of orbits. 

There is a natural map from the space of strata into the space of orbits
of $G$ acting  by conjugation on its lattice of subgroups.   In the space
of strata there are two distinguished ones: the maximal stratum $M_\mathrm{reg}$, which
is the union of orbits with minimal isotropy group (when the action is
effective, this is the set of points where the group has a locally free action);
and the minimal stratum, $M_{\mathrm{min}}$, made of the fixed points of the action.  If the
action is totally ineffective, i.e., trivial, this stratum is the
manifold itself. 

The manifold $M/G$ is a stratified manifold too.  The strata of the
orbit space are in one-to-one correspondence with the strata of $M$. 
The set $M/G$ decomposes in this way in a disjoint union of smooth
manifolds $S_\alpha$,  $M/G = \cup_\alpha S_\alpha$, where $\alpha$ labels
the strata in $M$.  The canonical projection $\pi\colon M \to M/G$
is a smooth submersion on each strata. Thus $\pi^{-1}(S_{\mathrm{min}}) =
\mathrm{Fix}_G(M)$ is a smooth submanifold (possibly non connected) of $M$.

  We will concetrate now in this minimal stratum $\mathrm{Fix}_G
(M)$ and we will show that it carries a (possibly trivial) symplectic
structure.

\begin{lemma}\label{fix_sym} The set $\mathrm{Fix}_G(M)$ of fixed points in $M$ under the
action of the compact group $G$ is a smooth symplectic submanifold of $M$.
\end{lemma} 

\begin{proof}  We notice first that the action of $G$ on $M$ induces a linear action of
$G$ on the vector space $T_xM$ for each $x\in\Fix_G(M)$
denoted by $g_*v$, $g\in G$, $v\in T_xM$.
Now it is simple to check that $g_*$ defines a linear representation
of $G$ on $T_x\Fix_G(M)$ (in general it is easy to check that $g_*$
defines a linear representation of $G_x$ on $T_xM$).  

Now let us suppose that $\Fix_G(M)$ is not zero dimensional (in that
case, the symplectic form will be trivial).  On the other hand let us choose a $G$-invariant metric
on $M$ compatible with $\omega$.  Such metric
exists because the group $Sp (2n,\R)$ is contractible to the subgroup
$U(n)$.  Moreover, this allows to find and adapted almost complex structure $J$ for
$\omega$, hence a metric $\eta$ which is related to $\omega$ by
\begin{equation}\label{compwJ} \eta_x (u,v) = \omega_x (J_x u,v)\, , \qquad \forall
u,v\in T_xM .
\end{equation}
Then, averaging $\eta$ and $J$ over the group $G$ we will
find the demanded structures pointwise.  

Let us take now a unitary eigenvector $v$ of $g_*$.
Consider now the linear function $f_v\colon T_x M \to \R$ defined
by $f_v(u) = \eta_x(v,u)$, for all $u\in T_xM$.  Now the 1-form $df_v =
\hat{\eta}_x (v)$ defines a hamiltonian vector field $X$ on $T_xM$ via
the symplectic structure $\omega_x$, i.e.
$$ 
i_X \omega_x = df_v \, .
$$
Because $\omega_x$ and $df_v$ are constant in $T_xM$, the same happens
for $X$ defining a constant tangent vector to $T_xM$ that can be
identified with a vector in $T_xM$.  Let us denote such vector by $v'$.
Then, computing $\omega_x (v,v')$ we get
$$ \omega_x (v',v) = i_X \omega_x (v) = df_v (v) = \eta_x (v,v) = 1 .$$ 
On the oher hand, it is clear that the covector $df_v$ is an eigenvector of $g_*$ 
with the same eigenvalue than $v$ because both $\eta_x$, $J$ are $g_*$ -invariant.  Hence,
because $\omega$ is also invariant with respect to $g_*$, $v'$ is an eigenvector
with eigenvalue the same eigenvalue for the action of $G$ on $T_xM$.  

Thus we conclude that the pair $v,v'$ just constructed span a symplectic 
subspace $V$ of dimension 2 in $T_x\Fix_G(M)$.  Then we can take the
orthogonal subspace $V^\perp$ with respect to the metric $\eta_x$ on
$T_x\Fix_G(M)$  and repeat the argument.  

Notice that because the
compatibility property of $\omega$ and $\eta$, Eq. (\ref{compwJ}), the
restriction of $\omega_x$ to $V^\perp$ is symplectic too, thus the
splitting $T_x\Fix_G(M) = V \oplus V^\perp$ is a symplectic
splitting.

\end{proof}

%\bigskip

Consider now $L$ to be a Lagrangian submanifold of a symplectic manifold $M$.  Let
$S$ be a symplectic submanifold of $M$.   Weinstein's theorem on the
local structure of Lagrangian embeddings asserts that there exists a tubular
neighborhood of $L$ on $M$ symplectomorphic to a tubular neighborhood of
$L$ on $T^*L$ as the zero section of the cotangent bundle. 
Then we can construct a Lagrangian foliation in the neighborhood of a
given point $x\in L$ by pulling back the (local) Lagrangian foliation of
$T^*L$ defined by a family of closed 1-forms $\alpha_y$ parametrized by
$y$, such that $\alpha_0 = 0$, in an open subset $U\subset \R^n$, $0\in U$, $n = \dim L$. 

Then, we will say
that the Lagrangian submanifold $L$ is transverse to the symplectic
submanifold $S$ at $x\in L\cap S$ in the symplectic category if $S$ is
transverse to a local Lagrangian foliation in a neighborhood of $x$
that contains $L$. We will say that $L$ and $S$ a symplectically
transversal if they are transversal in the symplectic category at each
$x\in L\cap S$.  Such local foliation always exist by the previous remarks
and the notion of symplectic transversality implies that the intersection
$L\cap S$ is a Lagrangian submanifold of $S$. 

\begin{lemma}\label{Lag_inter} If the Lagrangian submanifold $L$ and the symplectic
submanifold are symplectically transversal, the intersection $L\cap S$
is a Lagrangian submanifold of $S$.
\end{lemma}  

\begin{proof}  Because the submanifold $S$ is transversal to a local foliation
of $M$ then the intersection with the leaves are submanifolds.  In
particular the intersection with $L$ is a submanifold.  This holds in a
neighborhood of each point, than we have patches covering the
intersection $L\cap S$.  But it is clear from the definition of local
foliation in the neighborhood of a point in $L$ that this patches can be
chosen to be contained in local coordinate neighborhoods of $L$ hence
they define on $L$ the same differentiable structure.

It is obvious that $L\cap S$ is an isotropic submanifold of $M$, hence
it is isotropic of the symplectic submanifold $S$.  Let now $X$ be in
$T(S\cap L)^\perp$. Then, $\omega (X,Y) = 0$ for all $Y\in T(S\cap L)$.
Because $T(S\cap L) \subset TS$, then $TS^\perp \subset T(S\cap
L)^\perp$, but $TS\oplus TS^\perp = TM$.
\end{proof}

%I AM LOST AGAIN !! *****  

%%%%%%%%%%%%%%%%%%
%%%%%%%%%%%%%%%%%%

\subsection{The reduction of the Grassmannian and the space of virtual
self-adjoint extensions}

We can apply the discussion in the previous section to the
space of elliptic self-adjoint extensions.  In fact, the space of elliptic self-adjoint
extensions of the Dirac operator $\Dsl$, the elliptic self-adjoint Grassmannian $\sp
M_{\ellip}(\Dsl)$ is a Lagrangian submanifold of the infinite dimensional
Grassmannian $\mathrm{Gr} (\Dsl)$.  Moreover the group $G$ acts on $\mathrm{Gr} (\Dsl )$
because unitary representations of a compact group transform a closed subspace $W$
into another one which still is in $\mathbf{Gr}$.  In addition $G$ acts on $\mathbf{Gr}$ symplectically.  

We can check this easily because in the dense
subset of the Grassmannian made of subspaces $W_T$ which are graphs of closed
operators $T$, the group $G$ acts as $g\cdot W_T =
W_{U(g)^\dagger T U(g)}$.  Hence the group $G$ will act on tangent vectors $\dot{A} \in T_{W_T}\mathrm{Gr}$, as 
$$
g_* \dot{A} = U(g)^\dagger \dot{A} U(g) \, .
$$ 
The following computacion shows that $\omega_W$ is actually $G$-invariant (the subindex at $\omega$ 
will be omitted):
\begin{eqnarray*} \omega (g_*\dot{A}, g_*\dot{B}) &=& 
\frac{i}{2} \Tr ((g_*\dot{A})^\dagger g_*\dot{B} - (g_*\dot{B})^\dagger g_*\dot{A}) \\ 
&=& \frac{i}{2} \Tr ((U(g)^\dagger\dot{A}U(g))^\dagger U(g)^\dagger\dot{B}U(g) -
(U(g)^\dagger\dot{B}U(g))^\dagger U(g)^\dagger\dot{A}U(g)) 
\\ &=&  \frac{i}{2} \Tr (U(g)^\dagger\dot{A}^\dagger \dot{B}U(g) -
U(g)^\dagger\dot{B}^\dagger \dot{A}U(g)) = 
\\ &=& \frac{i}{2} \Tr ( \dot{A}^\dagger \dot{B} - \dot{B}^\dagger
\dot{A}) = \omega (\dot{A}, \dot{B}) \, .
\end{eqnarray*}

On the other hand, the fixed set of $G$ in $\mathrm{Gr}(\Dsl)$ is formed precisely by the
invariant subspaces in $\mathrm{Gr}(\Dsl)$. That is, 
$$
\Fix_G (\mathrm{Gr}(\Dsl)) = \set{W\in \mathrm{Gr}(\Dsl) \mid g\cdot W = W, ~ \forall g\in G}\, .
$$  
Notice that if $W = W_T$, i.e., $W$ is the graph of the operator $T$, then
$g\cdot W = W$ iff $U(g) T = T U(g)$.  That means that the set of fixed
points under $G$ is a the closure in the Grassmannian of the set of $G$-invariant
operators.

Moreover, because of Lemma \ref{fix_sym}, $\mathrm{Fix}_G (\mathrm{Gr}(\Dsl))$ is symplectic.  On
the other hand because of Thm. \ref{Ellip_Lag}, $\sp M_{\mathrm{ellip}} (\Dsl)$ is a Lagrangian submanifold in $\mathrm{Gr} (\Dsl)$.   In addition both $\mathcal{M}(\Dsl)$ and $\mathcal{M}_\mathrm{ellip}(\Dsl)$ are $G$-invariant.  Obviously, if $W$ is a self-adjoint
subspace, $g\cdot W$ is also self-adjoint because $G$ acts unitarily in $\sp H_D$. 

Finally, notice that  $\Fix_G (\mathrm{Gr} (\Dsl)) \cap \sp M_{\ellip} (\Dsl) = \Fix_G (\sp M_\ellip
(\Dsl))$.  Then, because of Lemma \ref{Lag_inter}, we conclude that $\Fix_G (\sp M(\Dsl))$
is a Lagrangian submanifold of $\Fix_G (\mathrm{Gr}(\Dsl))$ that will be
denoted by $\sp M_0$.  We have concluded the proof of the following
theorem.  

\begin{theorem}\label{Dirac_G_inv} The space $\Fix_G (\mathrm{Gr} (\Dsl))$ of $G$-invariant subspaces of the Grassmannian
$\mathrm{Gr} (\Dsl)$ is a symplectic submanifold. Moreover the space of $G$-invariant elliptic self-adjoint extensions of $\Dsl$ is a Lagrangian submanifold of $\Fix_G (\mathrm{Gr} (\Dsl))$.
\end{theorem}

As it was pointed before, the spaces $\Fix_G (\sp M_\ellip(\Dsl))$ and $\Fix_G (\sp M(\Dsl))$ are closely
related to the spaces of elliptic self-adjoint extensions and self-adjoint extensions of the
quotient Dirac operator $\Dsl/G$ respectively, but they still consist of functions on
$\partial \Omega$ and not in $\partial \Omega/G$.  In order to compare them we will have to relate $\mathrm{Gr}
(\sp H_{D/G})$ with $\Fix_G (\sp H_D)$. However that is easily done, noticing
that on each $G$-invariant subspace $W$ we can select the closed
subspace made of $G$-invariant vectors, i.e., the eigenspace with
eigenvalue 1 for the unitary representation of $G$.  This space will be
denoted by $W^G$ and it coincides with the intersection $W\cap \sp
H_D^G$.  

\begin{proposition}  For a compact group action of $G$ on $\Omega$ with smooth quotient space $\Omega/G$ and a $G$-invariant Dirac operator $\Dsl$, the spaces $\Fix_G(\mathcal{M}_\ellip(\Dsl ))$ and $\mathcal{M}_\ellip(\Dsl/G))$ are symplectically diffeomorphic. 
\end{proposition}

We conclude this discussion by remarking that even if the space $\sp H_D^G$ of invariant
vectors is always well defined for arbitrary $G$-actions on $\Omega$, it
is not the same for $\mathrm{Gr} (\Dsl/G)$ because if the space of leaves $\Omega$ is
singular, it is not obvious at all what is the meaning of $\Dsl/G$.   

Thus we define the space of `virtual elliptic self-adjoint extensions' of the
quotient Dirac operator $\Dsl/G$ as the family of closed subspaces in $\Fix_G(\mathcal{M}_\ellip(\Dsl ))$.  
We shall denote this space by $\sp M (\Dsl/G)$ and it is a Lagrangian submanifold of $\mathrm{Gr}
(\sp H_D^G)$ whenever both are manifolds. 

%%%%%%%%%%%%%%%%%%%%%%%%%%%%%%
%%%%%%%%%%%%%%%%%%%%%%%%%%%%%%
%%%%%%%%%%%%%%%%%%%%%%%%%%%%%%

\newpage

\section{Spontaneous symmetry breaking and self-adjoint extensions}\label{section:breaking}

A basic ingredient in the construction of gauge invariant theories of
interacting quantum fields was the discovery of the mechanism of
spontaneous symmetry breaking that allowed to explain the mass spectrum
of the theories \cite{Go61}, \cite{Hi64}.   

\subsection{The notion of spontaneous symmetry breaking}

Spontaneous symmetry breaking
has since then received a detailed analysis from the mathematical and
physical viewpoints \cite{Ki67}, \cite{We74}.  L. Michel and L. Radicati
in particular obtained the firsts theorems on the subject \cite{Mi68},
\cite{Mi73} and a series of refinements and developments followed (see
for instance \cite{Ga97} and references therein).

Roughly speaking, spontaneous symmetry breaking happens when the
equations describing a physical theory possess a given symmetry but the
physical states of the theory have a smaller symmetry algebra.  The way
this situation is modeled in Quantum Field Theories is based on the
fact that the symmetry of the Lagrangian has not to the same as the
symmetry of the vacuum state of the theory.  The specific mechanism that
selects the physical vacuum stated among the orbits of the symmetry
group of the theory on the Hilbert space of the quantum system is not
clearly stated.   

We will discuss a concrete mechanism of spontaneous
symmetry breaking on our simplified model of first quantized quantum
systems.    The basic idea is extremely simple.  As it was described in
the previous section, for a given symmetry group $G$, invariant Dirac
and Laplacian operators can be constructed.  Thus the
equations of the theory will be $G$-invariant.   However, if the
manifold of classical states, in our case $\Omega$,  has a
nontrivial boundary, the self-adjoint extensions of the operators could
break this theory in the sense that the domain of the given
self-adjoint extension will not be invariant under the group $G$.

As it was stated in Thm. \ref{Dirac_G_inv} the elliptic self-adjoint
extensions of the quotient Dirac operator $\Dsl/G$ are in one-to-one
correspondence with the space $\sp M_G(\Dsl)$, the space of fixed points
of the action of $G$ on the space $\mathcal{M}_\ellip(\Dsl)$ of elliptic self-adjoint extensions of $\Dsl$.  Thus all subspaces
in its complement are breaking the symmetry
of the theory, in the sense that the functions in the corresponding domains not only do not
have to be $G$-invariant, but that its transformed under $G$ could lie out
of the domain of the operator itself.  Thus the operator $\Dsl_W$ defined by $W$ not in $\Fix_G$ will
of course commute with the action of $G$, but its spectrum will
certainly not be $G$-invariant.  We shall explore this idea in detail
in what follows.

Let us consider for instance a simple example that will illustrate this idea.
Consider the 1D free particle on the interval $[-\pi,\pi]$ with Hamiltonian $H  = -d/dx^2$. The the parity group $\mathbb{Z}_2$ acts on the Hilbert space $L^2([-\pi,\pi])$ as $P\Psi (x) = \Psi (-x)$ and clearly $[P,H] = 0$ \cite{Ag01}.  
However not all allowed self-adjoint extensions of the Hamiltonian operator are invariant
with respect to $P$.   Actually the space of self-adjoint extensions of $H$ is parametrised by the 
unitary group $U(2)$. The induced action on the boundary space Hilbert space $\mathbb{C}_{-\pi} \oplus \mathbb{C}_{\pi}-$,
where $\mathbb{C}_{\pi}$ denotes the space of boundary values $\psi_\pm^\pi = \Psi(\pi) \pm i\Psi'(\pi)$ and
$\mathbb{C}_{\pi}$, $\psi_\pm^{-\pi} = \Psi(-\pi) \mp i\Psi'(-\pi)$ respectively.  The element $P$ acts on this space as 
the matrix 
$$
P = \left[ \begin{array}{cc} \, 0 & \, 1 \\ \, 1 & \, 0 \end{array}\right] \, ,
$$
hence, only those self-adjoint boundary conditions defined by unitary matrices of the form:
$$
U(\alpha, \delta) = e^{i\delta}  \left[ \begin{array}{cc} \cos \alpha & e^{i\pi/2} \sin \alpha \\ e^{i\pi/2} \sin \alpha &  \cos \alpha\end{array}\right] \, ,
$$
will define domains invariant under $P$.  Thus, for instance, quasi-periodic boundary conditions will break the symmetry $P$.

Another, less trivial, example in the same spirit is provided
by a result by Falceto and Esteve on a generalization of the Virial Theorem \cite{Es12}.
Consider a free particle in one dimension restricted to
move in $[0,\infty)$ and subject to Robin boundary conditions, i.e. 
$\Psi'(0) + \alpha\Psi(0) = 0$ with $\alpha > 0$.  In this case the free Hamiltonian has a single eigenfunction
$\Psi_0 (x) = \sqrt{2}\alpha  e^{-\alpha x}$ with eigenvalue
$$
E_0 = -\frac{\hbar^2\alpha^2}{2m} \, .
$$
If we use the Virial Theorem:
$$
2\langle \Psi_0 \mid T \mid \Psi_0 \rangle = \langle \Psi_0 \mid x V'(x) \mid \Psi_0 \rangle \, ,
$$
 to compute the expectation value of the kinetic energy in this state we obtain that it vanishes, which is in contradiction with the real result:
$$
 \langle \Psi_0 \mid T \mid \Psi_0 \rangle = E_0 \, . 
$$
The reason for this apparent contradiction is the fact that the domain of the Hamiltonian is not preserved by the generator of the group of scale transformations
$$
E = -\frac{i}{\hbar} xp - \frac{1}{2} \, ,
$$
and the virial theorem has to be modified.

The group theoretical analysis of this example is as follows.  
The group $G$ of scale transformations acts on the 
space of $L^2(\mathbb{R}^+)$ functions as:
$$ 
(\delta_s \Psi)(x) = e^{s/2}\Psi (e^s x) \, , \qquad s \in \mathbb{R} , x \in \mathbb{R}^+  \, .
$$
Notice that this action defines a unitary representation of the group as:
$\parallel \delta_s \Psi \parallel = \parallel \Psi \parallel$ for all $s \in \mathbb{R}$,
$\Psi \in L^2(\mathbb{R}^+)$.   The free particle Hamiltonian operator $H_0$ transforms
under the action of $\delta_s$ as:
$$
\delta_s^{-1} H_0 \delta_s = e^{2s} H_0 \, ,
$$
thus the virial theorem is obtained by observing that
\begin{eqnarray*}
\langle \Psi_n \mid [H,E] \mid \Psi_n \rangle &=& \langle \Psi_n \mid HE  \mid \Psi_n \rangle - \langle \Psi_n \mid EH  \mid \Psi_n\rangle  \\ &=&  \langle H\Psi_n \mid E  \mid \Psi_n \rangle  -  \langle \Psi_n \mid  EH  \mid \Psi_n \rangle  = 0  \, ,
\end{eqnarray*}
where $H$ is a sel-adjoint operator with domain $D(H)$, $\mid \Psi_n\rangle$ is an eigenvector, and $E\mid \Psi_n \rangle$ must lie in
the domain of $H$.
However if $E\mid \Psi_n \rangle \notin D(H)$ then the theorem fails.
In the situation above, it is easy to realise that the vector $E(\Psi_0)$
will not be in the domain $D_\alpha \subset L^2(\mathbb{R}^+)$ defined
by Robin's boundary conditions, because $(\delta_s\Psi) (0) = e^{s/2}\Psi(0)$, whereas 
$(\delta_s\Psi')(0) = \Psi'(0)$, hence if $\Psi \in D_\alpha$, $\alpha \neq 0$, then
$\delta_s\Psi \in D_{e^{s/2}\alpha} \neq D_{\alpha} $.

%%%%%%%%%%%%%%%%%%%
%%%%%%%%%%%%%%%%%%%

\subsection{The bifurcation diagram of the space of self-adjoint extensions}

%%%%%%%%%%%%%%%%%%%
%%%%%%%%%%%%%%%%%%%

We will analyze now in more detail the structure of the Lagrangian
Grassmannian $\mathcal{L}_\Dsl$ of the $G$-invariant Dirac operator $\Dsl$. 

Let us recall that the action of the group $G$ induces an action on
subspaces of the boundary data Hilbert space $\mathcal{\Dsl}$, hence
it induces an action on the space of self-adjoint extensions $\mathcal{M}_\Dsl$
of $\Dsl$.  Moreover it also induces and action on the elliptic Grasmmanian
$\mathrm{Gr}(\Dsl)$. Notice that because the representation $V$ of the group is unitary and continuous, then the transformed of a given subspace $W^g = V(g)^\dagger W V(g)$ will have Fredholm
and Hilbert-Schmidt projections on the corresponding polarization.
Hence the action of $G$ will map the elliptic self-adjoint Grassmannian $\mathcal{L}_{\dsl}$
onto itself.

The action of $G$ on $\mathcal{L}_{\dsl}$ will induce a stratified
structure on this manifold.  If we denote by $G_W$ the isotropy group
of the vector subspace $W\in \mathcal{L}_{\dsl}$, all the subspaces whose isotropy groups
will lie in the conjugacy class of $G_W$ will constitute a submanifold,
the stratum $\sp O_W$ of $\mathcal{L}_{\dsl}$.  

The fixed point set, i.e., the strictly invariant subspaces, that correspond to 
self-adjoint extensions of the reduced Dirac operator, define the minimal stratum $\mathcal{O}_{\mathrm{min}} = \mathrm{Fix}_G(\mathcal{L}_{\dsl})$, and the set of
completely non-invariant subspaces, i.e., those subspaces $W$ such that
$ W^g \cap W = \mathbf{0}$, constitute the maximal strata denoted by $\mathcal{O}_{\mathrm{max}}$. 
It is noticeable that $\mathcal{O}_{\mathrm{max}}$ is an open dense submanifold of $\mathcal{L}_{\dsl}$.

The manifold $\mathcal{L}_{\dsl}$ is decomposed then as a union 
$$ \mathrm{Gr}(\sp H_D) = \cup_{\alpha} \sp O_\alpha ,$$  
of the different strata $\sp O_\alpha$ labeled by an index $\alpha$
that runs over the conjugacy classes of subgroups of $G$.  
The decomposition above is such that the boundary of any
stratum is a union of strata,
$$
\partial \sp O_\alpha = \cup_{\alpha'} \sp O_{\alpha'} \, ,
$$
and under generic conditions, the strata $\mathcal{O}_\alpha$ are submanifolds with boundary 
given by the expression above.

A partial ordering can be defined on the space of strata as follows $\mathcal{O}_\alpha \prec \mathcal{O}_\beta$ if 
$\mathcal{O}_\alpha \subset \partial \mathcal{O}_\beta$. 
This partial relation induces a natural partial ordering in the space of
conjugacy classes of subgroups of $G$ which non-void strata.  Denoting
by $G_\alpha$ a representative of the conjugacy class of subgroups
of the strata $\sp O_\alpha$, we will define
$G_\alpha \prec G_\beta$ iff $\sp O_\alpha \prec \sp O_\beta$ or, in other words,
that is $\alpha \prec \beta$ iff $\mathcal{O}_\beta \prec \mathcal{O}_\alpha$. The lattice of
subgroups $\{G_\alpha \}$ with the partial order $\prec$ will be called the
bifurcation diagram of the $G$-invariant operator $\Dsl$.  

The bifurcation diagram of $\Dsl$ represents the possible schemes of
breaking the symmetry $G$.  The subgroups $G_\alpha \prec G_\beta$ can
actually be chosen to be closed subgroups of each other, hence, for finite-dimensional Lie groups, the
bifurcation diagram will necessarily have a finite number of levels (the
maximum number of subgroups needed to go from the trivial group to
$G$).  

The intersection of the space $\sp M(\Dsl)$ of self-adjoint extensions of
$\Dsl$ with the corresponding strata $\sp O_\alpha$ will describe the
self-adjoint extensions of $\Dsl$ with symmetry group
$G_\alpha$ determined by the boundary conditions.  We will say that the symmetry is broken from $G_\alpha$ to
$G_\beta$ if there is a change from the boundary conditions on $\sp
O_\beta \subset \partial \sp O_\alpha$ to the interior of $\sp
O_\alpha$.  

The problem of understanding such symmetry breaking is to
describe how the properties of the operators $\Dsl_W$ change when $W$
changes from $\sp O_\beta$ to $\sp O_\alpha$, e.g., how their spectral invariants change. 

We will only describe here the breaking from the trivial symmetry to
$G$, i.e., from the minimal strata to the maximal one.  Notice that this
can always be done because $\sp O_\mathrm{min} \subset \partial \sp O_\mathrm{max}$ and
$\sp O_\mathrm{max}$ is a dense open submanifold.

We denote by $W$ a small generic perturbation of a $G$-invariant
subspace $W_0$, i.e., $W\in \sp O_{max}$, and $W^g \cap W = \mathbf{0}$ for all $g\in G$.  Let $\xi$ a section in the domain of $\Dsl_W$, i.e.,
$\xi\mid_{\partial \Omega} \in W$.  Then defining a new variable
$\theta = \xi - \xi_0$ where $\xi_0$ is a lowest positive eigenvalue
solution of the eigenvalue equation $\Dsl_W \xi = m \xi$, we obtain that the
spectral equation for $\Dsl_W$ becomes the inhomogeneous equation for  $\theta$:
$$ 
\Dsl\theta = E\theta + (E-m_0) \xi_0 
$$ 
and $\theta\mid_{\partial\Omega} = 0$, that is Dirichlet's boundary conditions which are always symmetric boundary
condition. Thus the bifurcation $G$ to $e$ modifies the operator $\Dsl$
adding a nontrivial inhomogeneous term.  

\newpage

%%%%%%%%%%%%%%%%%%%%%%
%%%%%%%%%%%%%%%%%%%%%%

\section{Conclusions and further developments}

Unfortunately there is no room in these notes to discuss many other examples and problems from quantum systems
involving boundary conditions and/or (non-)self-adjoint extension
that deserve a detailed analysis and for which some of
the results and ideas presented are relevant.  We will just
quote some of them.

The physical role, generation and control of edges states
is paramount among them.   We have briefly discussed
how edge states are formed when approaching the Cayley
surface, however a more detailed analysis is needed. 
 
Close to this is the relation between topology change and
dynamics of boundary conditions.  Notice that states corresponding
to boundary conditions close to the Cayley manifold can exhibit
(when considered in a Quantum Field Theory) quantum corrections 
corresponding to changes in the topology of the underlying space. 
Again, these aspects should be investigated further.
 
The theory of self-adjoint extensions of Laplace or
Dirac operators on manifolds with singularities, arising either
from group actions or by suitable boundary conditions. like orbifolds, quantum systems with meromorphic
or distributional, like $\delta'$, potentials, etc.
 
%\vskip 3cm

\section*{Acknowledgments} M. A.  work has been partially 
supported by the Spanish MICINN grants  FPA2012-35453 and CPAN 
Consolider Project CDS2007-42 and DGA-FSE (grant 2014-E24/2).
A.I. was partially supported by the Community of Madrid project QUITEMAD+, S2013/ICE-2801, and MICIN MTM 2010-21186-C02-02.    M.A. and A.I. would also like to thank the warm hospitality
of the Dipartimento di Fisica dell'Universita di Napoli were
this work was started.    The authors would like to thank the Erwing
Schr\"odinger Institute for Mathematical Physics for the friendly and
stimulating atmosphere they enjoyed during their stay there and were this project
was partially developed.  The authors would like to thank the hospitality of the
NIThP at Stellenbosch were the project was completed.  We would also like to thank
the  suggestions and comments by Balachandran, Mu\~noz-Castaneda, P\'erez-Prado, 
Santagata, Rajeev, Vaydia and Falceto.

\end{document}